\documentclass[acmsmall,nonacm]{acmart} 

\usepackage{comment} 

\usepackage{thm-restate} 
\usepackage{caption} 
\usepackage{subcaption} 
\usepackage{xspace} 
\usepackage{tabularx} 
\usepackage{enumitem}

\newtheorem{claim}{Claim}[section]

\newcommand{\complexityclass}[1]{\textsf{\textup{\small #1}}\xspace}
\newcommand{\np}{\complexityclass{NP}}
\newcommand{\conp}{\complexityclass{co-NP}}
\newcommand{\pls}{\complexityclass{PLS}}

\definecolor{TUMBlue}{HTML}{0065BD}

\usepackage{tikz}
\usetikzlibrary{decorations.pathreplacing}
\usetikzlibrary{positioning,arrows,shapes,snakes,calc}

\tikzstyle{arrow}=[->,>=latex]

\tikzset{ 
protovertex/.style={
  draw,
  circle,
  inner sep=0,
  minimum size=.2cm}
}

\tikzset{ 
indicatorvertex/.style={
  fill,
  circle,
  inner sep=0,
  minimum size=.1cm}
}

\newcommand{\partition}{\pi}
\newcommand{\spartition}{\partition_0} 
\newcommand{\util}{v}
\newcommand{\potential}{\Lambda}
\newcommand{\N}{\mathcal{N}}
\newcommand{\existpb}[2]{\textsc{$\exists$-#1-Sequence-#2}\xspace}
\newcommand{\convergpb}[2]{\textsc{$\forall$-#1-Sequence-#2}\xspace}
\newcommand{\singleton}{singleton partition\xspace}

\newcommand{\decpb}[3]{\medskip
\noindent\begin{tabularx}{\columnwidth}{lX}
\hline
\multicolumn{2}{l}{#1} \\
\hline
Input: & #2 \\ 
Question: & #3 \\ 
\hline
\end{tabularx}
\medskip
}

\allowdisplaybreaks

\usepackage{multirow, pbox}
\newcolumntype{L}[1]{>{\raggedright\let\newline\\\arraybackslash\hspace{0pt}}m{#1}}
\newcolumntype{C}[1]{>{\centering\let\newline\\\arraybackslash\hspace{0pt}}m{#1}}
\newcolumntype{R}[1]{>{\raggedleft\let\newline\\\arraybackslash\hspace{0pt}}m{#1}}
\newcolumntype{N}{@{}m{0pt}@{}}

\DeclareFontFamily{U}{matha}{\hyphenchar\font45}
\DeclareFontShape{U}{matha}{m}{n}{
      <5> <6> <7> <8> <9> <10> gen * matha
      <10.95> matha10 <12> <14.4> <17.28> <20.74> <24.88> matha12
      }{}
\DeclareSymbolFont{matha}{U}{matha}{m}{n}
\DeclareMathSymbol{\positiveResultSymbol}{0}{matha}{"60}
\DeclareMathSymbol{\negativeResultSymbol}{0}{matha}{"61}

\AtBeginDocument{%
  \providecommand\BibTeX{{%
    \normalfont B\kern-0.5em{\scshape i\kern-0.25em b}\kern-0.8em\TeX}}}

\begin{document}

\title{Reaching Individually Stable Coalition Structures}

\author{Felix Brandt}
\email{brandtf@in.tum.de}
\affiliation{%
  \institution{Department of Informatics, Technical University of Munich}
  \country{Germany}
}
\author{Martin Bullinger}
\email{bullinge@in.tum.de}
\affiliation{%
  \institution{Department of Informatics, Technical University of Munich}
  \country{Germany}
}
\author{Ana\"elle Wilczynski}
\email{anaelle.wilczynski@centralesupelec.fr}
\affiliation{%
  \institution{MICS, CentraleSup\'elec, Universit\'e Paris-Saclay}
  \country{France}
}

\begin{abstract}
The formal study of coalition formation in multi-agent systems is typically realized in the framework of hedonic games, which originate from economic theory. The main focus of this branch of research has been on the existence and the computational complexity of deciding the existence of coalition structures that satisfy various stability criteria. The actual process of forming coalitions based on individual behavior has received little attention. In this paper, we study the convergence of simple dynamics leading to stable partitions in a variety of established classes of hedonic games including anonymous, dichotomous, fractional, and hedonic diversity games. The dynamics we consider is based on individual stability: an agent will join another coalition if she is better off and no member of the welcoming coalition is worse off.

Our results are threefold. First, we identify conditions for the (fast) convergence of our dynamics. To this end, we develop new techniques based on the simultaneous usage of multiple intertwined potential functions and establish a reduction uncovering a close relationship between anonymous hedonic games and hedonic diversity games. Second, we provide elaborate counterexamples determining tight boundaries for the existence of individually stable partitions. Third, we study the computational complexity of problems related to the coalition formation dynamics. In particular, we settle open problems suggested by \citet{BoJa02a}, \citet{BBS14a}, and \citet{BoEl20a}.
\end{abstract}



\maketitle

\section{Introduction}

Coalitions and coalition formation are central concerns in the study of multi-agent systems as well as cooperative game theory. 
Typical real-world examples include individuals joining clubs or societies such as orchestras, choirs, or sport teams, countries organizing themselves in international bodies like the European Union (EU) or the North Atlantic Treaty Organization (NATO), students living together in shared flats, or employees forming unions. 
The formal study of coalition formation is often realized using so-called hedonic games, which originate from economic theory. In these games, the central goal is to identify coalition structures (henceforth partitions) that satisfy various stability criteria based on the agents' preferences over coalitions, that is, over subsets of agents. A partition is defined to be stable if single agents or groups of agents cannot gain by deviating from the current partition by means of leaving their current coalition and joining another coalition, or forming a new one. The permitted deviations can be constrained by other agents, for instance by means of contracts with an existing coalition or by the necessity of consent when admitting a new member. These constraints lead to a large body of stability concepts \citep{AzSa15a}. Two important and well-studied questions in this context concern the existence of stable partitions in restricted classes of hedonic games and the computational complexity of finding a stable partition. However, stability is only concerned with the end-state of the coalition formation process and ignores how these desirable partitions can actually be reached. Essentially, an underlying assumption in most of the existing work is that there is a central authority that receives the preferences of all agents, computes a stable partition as an end-state, and has the means to establish this partition among the agents. By contrast, our work focuses on simple dynamics, where starting with some partition (e.g., the partition of singletons), agents deliberately decide to join and leave coalitions based on their individual preferences. We study the convergence of such a process and the stable partitions that can arise from it. For example, in some cases the only partition satisfying a certain stability criterion is the grand coalition consisting of all agents, while the dynamics based on the agents' individual decisions can never reach this partition and is doomed to cycle. 

The dynamics we consider is based on \emph{individual stability}, a natural notion of stability going back to \citet{DrGr80a}: an agent will join another coalition if she is better off and no member of the welcoming coalition is worse off.
Individual stability is suitable to model the situations mentioned above. For instance, by Article 49 of the Treaty on European Union, admitting new members to the EU requires the unanimous approval of the current members. Similarly, by Article 10 of their founding treaty, unanimous agreement of all parties is necessary to become a member of the NATO. Also, for joining a choir or orchestra it is often necessary to audition successfully, and joining a shared flat requires the consent of all current residents.
This distinguishes individual stability from Nash stability, which ignores the consent of members of the welcoming coalition. 

The analysis of coalition formation processes provides more insight in the natural behavior of agents and the conditions that are required to guarantee that desirable social outcomes can be reached without a central authority. 
Similar dynamic processes have been studied for matchings, which can be seen as a special domain of coalition formation where only coalitions of size~2 are allowed \citep[e.g.,][]{RoVa90a,AbRo95a,
BrWi19a}. 
More recently, dynamics of coalition formation have also come under scrutiny in the context of hedonic games \citep{BFFMM18a,HVW18a,CMM19a,FMM21a,BBT22a}.
While coalition formation dynamics
are an object of study worthy for itself, they can also be used as a means to design algorithms for the computation of stable outcomes, and have been implicitly used for this purpose before. For example, the algorithm by \citet{BoEl20a} for finding an individually stable partition in hedonic diversity games predefines a promising partition and then reaches an individually stable partition by running the dynamics from there. Similarly, the algorithm by \citet{BoJa02a} for finding an individually stable partition on games with ordered characteristics, a generalization of anonymous hedonic games, runs the dynamics using a specific sequence of deviations starting from the singleton partition.

In addition, the study of dynamics offers a more fine-grained view in games where the static concepts of stability only gives limited information. We will see that there exist classes of hedonic games in which individually stable partitions are guaranteed to exist but dynamics cycle can be doomed to cycle when executing from specific starting configurations. Two important examples are fractional hedonic games with non-negative weights and hedonic diversity games. For the former, the grand coalition is individually stable and for the latter, individually stable partitions always exist \citep{BoEl20a}. By contrast, our results show that dynamics can behave very differently, illustrating that dynamics can offer a broader picture on stability.

In many cases, the convergence of the dynamics of deviations follows from the existence of 
potential functions, whose local optima form individually stable states.
Generalizing a result by \citet{BoJa02a}, \citet{Suks15a} has shown via a potential function argument that an individually stable---and even a Nash stable---partition always exists in subset-neutral hedonic games, a generalization of symmetric additively-separable hedonic games.
Using the same potential function, it can straightforwardly be shown that the dynamics converge.\footnote{By inclusion, convergence also holds for symmetric additively-separable hedonic games. Symmetry is essential for this result to hold since an individually stable partition may not exist in additively-separable hedonic games, even under additional restrictions~\citep{BoJa02a}.}

Another example are hedonic games with the common ranking property, a class of hedonic games where preferences are induced by a common global order \citep{FaSc88a}. Here, the dynamics associated with core-stable deviations is known to converge to a core-stable partition that is also Pareto-optimal\footnote{A partition is Pareto-optimal if every partition preferred to this partition by some agent is worse for another agent.}, thanks to a potential function argument~\citep{CaKi19a}.
The same potential function implies convergence of the dynamics based on individual stability.

In this paper, we study the coalition formation dynamics based on individual stability for a variety of classes of hedonic games, including anonymous hedonic games (AHGs), hedonic diversity games (HDGs), fractional hedonic games (FHGs), and dichotomous hedonic games (DHGs).
Computational questions related to the dynamics are investigated in two ways: the existence of a \emph{path to stability}, that is the existence of a sequence of deviations that leads to a stable state, and the \emph{guarantee of convergence} where every sequence of deviations should lead to a stable state.
The former gives an optimistic view on the behavior of the dynamics and may be used to motivate the choice of reachable stable partitions (we can exclude ``artificial'' stable partitions that may never naturally form).
If such a sequence can be computed efficiently, it enables a central authority to coordinate the deviations towards a stable partition.
On the other hand, guaranteed convergence allows agents to reach a stable outcome without further coordination. This provides strong stability guarantees under more pessimistic assumptions on the agents' behavior. Whether we obtain positive or negative results concerning the convergence of the dynamics depends on various dimensions of the input concerning the initial partition, restrictions imposed on the agents' preferences, and the selection of deviations. We identify clear boundaries to computational tractability based on these specifications.
Our main results are summarized as follows.
\begin{itemize}%
\item In AHGs, the dynamics is guaranteed to converge for (naturally) single-peaked preferences. 
On the other hand, we provide a 15-agent example showing the non-existence of individually stable partitions in general AHGs. 
The previously known smallest counterexample by \citet{BoJa02a} requires 63 agents and the existence of smaller examples was an acknowledged open problem \citep[see][]{Ball04a,BoEl20a}.
\item We provide an elaborate reduction for HDGs that eventually establishes a close relationship to AHGs and show guaranteed convergence of the dynamics for strict and naturally singled-peaked preferences when starting from the \singleton and agents' deviations satisfy a weak constraint.
In contrast to empirical evidence reported by \citet{BoEl20a}, we show that all of the above assumptions are essential for the convergence result. In particular, cycling of the dynamics is possible if the starting partition is the \singleton and preferences are restricted to be strict and naturally single-peaked.
\item In FHGs, the dynamics is guaranteed to converge for simple symmetric preferences when starting from the \singleton or when preferences form an acyclic digraph.
On the other hand, we show that individually stable partitions need not exist in general symmetric FHGs,
which was left as an open problem by \citet{BBS14a}.
\item For each of the above classes and DHGs, we identify computational boundaries. In particular, we show that deciding whether there is a sequence of deviations leading to an individually stable partition is \np-hard while deciding whether all sequences of deviations lead to an individually stable partition is \conp-hard. Some of these results hold under preference restrictions and even when starting from the singleton partition.
\end{itemize}
For the sake of readability, some (parts of) proofs are omitted from the body of the paper, they can nevertheless be found in the appendix.

\section{Preliminaries and Model}

Let $N=[n]=\{1,\dots,n\}$ be a set of $n$ agents.
The goal of a coalition formation problem is to partition the agents into different disjoint coalitions according to their preferences. 
Formally, a solution is a \emph{partition} of $N$, i.e., a subset $\partition\subseteq 2^N$ such that $\bigcup_{C\in \partition} C = N$, and for every pair $C,D\in \partition$, it holds that $C = D$ or $C\cap D = \emptyset$. An element of a partition is called \emph{coalition}, and given a partition~$\partition$, we denote by $\partition(i)$ the coalition containing agent~$i$. 
Two prominent partitions are the \emph{\singleton} $\partition$ given by $\partition(i)=\{i\}$ for every agent $i\in N$, and the \emph{grand coalition} $\partition$ given by $\partition=\{N\}$.

Since we focus on dynamics of deviations, we assume that there exists an initial partition $\spartition$, which could be a natural initial state (such as the singleton partition) or the outcome of a previous coalition formation process.

\subsection{Classes of Hedonic Games}

In a hedonic game, the agents only express preferences over the coalitions to which they belong, i.e., there are no externalities.
Let $\N_i$ denote all possible coalitions containing agent $i$, i.e., $\N_i=\{C\subseteq N: i\in C\}$.
A hedonic game is defined by a tuple $(N,(\succsim_i)_{i\in N})$ where $\succsim_i$ is a weak order over $\N_i$ which represents the preferences of agent $i$, i.e., $C\succ_i C'$ means that agent $i$ strictly prefers coalition $C$ to coalition $C'$, and $C\sim_i C'$ means that agent $i$ is indifferent between coalitions $C$ and $C'$.
Since $|\N_i|=2^{n-1}$, the preferences 
are rarely given explicitly, but rather in some concise representation.
These representations give rise to several classes of hedonic games:
\begin{itemize}%

\item \emph{Anonymous hedonic games (AHGs)}~\citep{BoJa02a}: The agents only care about the size of the coalition they belong to, i.e., 
for each agent $i\in N$, there exists a weak order $\succsim_i^S$ over integers in $[n]$ (superscript $S$ for sizes) such that $\partition(i) \succsim_i \partition'(i)$ iff $|\partition(i)| \succsim_i^S |\partition'(i)|$.
\item \emph{Hedonic diversity games (HDGs)}~\citep{BEI19a}: The agents are divided into two different types (or colors). We call them red and blue agents and they are represented by the subsets $R\subseteq N$ and $B\subseteq N$, respectively, such that $N=R\cup B$ and $R\cap B = \emptyset$.
Each agent only cares about the proportion of red agents present in her own coalition, i.e., for each agent $i\in N$, there exists a weak order $\succsim_i^F$ over $\{\frac{p}{q}:p\in [|R|]\cup\{0\}, q\in[n]\}$ (superscript $F$ for fractions) such that $\partition(i) \succsim_i \partition'(i)$ iff $\frac{|R\cap \partition(i)|}{|\partition(i)|} \succsim_i^F \frac{|R\cap \partition'(i)|}{|\partition'(i)|}$.\footnote{To keep the notation concise, we abuse notation by omitting the superscripts of $\succsim_i^S$ and $\succsim_i^F$ when they are clear from the context. Hence, $\succsim_i$ may also denote agent $i$'s preference order over coalition sizes in case of an AHG, or over fractions in case of an HDG.}
\item \emph{Fractional Hedonic Games (FHGs)}~\citep{ABB+17a}: The agents evaluate a coalition by assessing how much they like each of its members on average, i.e., 
for each agent $i$, there exists a utility function $\util_i:N\rightarrow\mathbb{R}$ where $\util_i(i)=0$ such that $\partition(i)\succsim_i\partition'(i)$ iff $\frac{\sum_{j\in \partition(i)}\util_i(j)}{|\partition(i)|}\geq \frac{\sum_{j\in \partition'(i)}\util_i(j)}{|\partition'(i)|}$.
An FHG can be represented by a weighted directed graph 
$G=(N,E)$ where, for the sake of readability, only non-null utilities are mentioned, i.e., $(i,j)\in E$ iff $\util_i(j)\neq 0$, and the weight of an arc $(i,j)\in E$ is equal to $\util_i(j)$.
An FHG is \emph{symmetric} if $\util_i(j)=\util_j(i)$ for every pair of agents~$i$ and~$j$; since two opposite arcs have the same weight, 
the representation by a graph can be simplified by directly considering a weighted 
undirected graph $G=(N,E)$ with weights $\util(i,j)$ on each edge $\{i,j\}\in E$.
An FHG is \emph{simple} if $\util_i:N\rightarrow\{0,1\}$ for every agent~$i$; since all arcs have the same weight, the representation by a graph can be simplified by directly considering an unweighted directed graph $G=(N,E)$ where $(i,j)\in E$ iff $\util_i(j)=1$.
We say that an FHG is \emph{simple asymmetric} if, for every pair of agents~$i$ and~$j$, $\util_i(j)\in\{0,1\}$ and $\util_i(j)=1$ implies $\util_j(i)=0$, i.e., it can be represented by an asymmetric directed graph.
\item \emph{Dichotomous hedonic games (DHGs)}: The agents only approve or disapprove coalitions, 
i.e., for each agent $i$ there exists a utility function $\util_i:\N_i\rightarrow \{0,1\}$ such that $\partition(i)\succsim_i\partition'(i)$ iff $\util_i(\partition(i))\geq \util_i(\partition'(i))$.
When the preferences are represented by a propositional formula, such games are called \emph{Boolean hedonic games}~\citep{AHLW16a}.
\end{itemize}
An anonymous game (or hedonic diversity game) is \emph{generally single-peaked} if there exists a linear order $>$ over integers in $[n]$ (or over ratios in $\{\frac{p}{q}:p\in [|R|]\cup \{0\}, q\in[n]\}$) such that for each agent $i\in N$ and each triple of integers $x,y,z\in [n]$ (or $x,y,z\in \{\frac{p}{q}:p\in |R|\cup \{0\}, q\in[n]\}$) with $x>y>z$ or $z>y>x$, $x\succ_i^S y$ implies $y \succsim_i^S z$ (or $x\succ_i^F y$ implies $y \succsim_i^F z$).
The obvious linear order $>$ that comes to mind is, of course, the natural order over integers (or over rational numbers). We refer to such games as \emph{naturally single-peaked}.
Clearly, a naturally single-peaked preference profile is generally single-peaked but the converse is not true. 

Up to our best knowledge, most papers on hedonic games only use naturally single-peaked preferences as single-peaked preferences \citep[see, e.g.,][]{BoJa02a}, since they deal with preferences over integers or fractions.
We have introduced generally single-peaked preferences, by considering any type of given order over integers or fractions, in the spirit of the initial definition of \citet{Blac48a} for social choice, where alternatives do not have an inherent order. 
By generalizing the single-peaked notion on hedonic games, we aim to capture the frontier of tractability with respect to this concept and thus strengthen our results. 

To improve the presentation when we display preferences for a large number of agents, we also write them in the form $i \colon X \succ Y \succ Z$, where $i\in N$ is an agent, and $X$, $Y$, and $Z$ are coalitions, or depending on the context, integers (representing sizes of coalitions in AHGs) or rational numbers (representing fractions in HDGs). 
If there are multiple agents $C\subseteq N$ with identical, we also denote there preferences in the form $C \colon X \succ Y \succ Z$. 
Note that we usually do not fully specify preferences in this notation, but focus on the relevant part of the preferences which affects a certain example or proof.

\subsection{Dynamics of Individually Stable Deviations}

Starting from the initial partition, agents can leave and join coalitions in order to improve their well-being.
We focus on unilateral deviations, which occur when a single agent decides to move from one coalition to another.
A \emph {unilateral deviation} performed by agent $i$ transforms a partition $\partition$ into a partition $\partition'$ 
where $\pi(i)\neq\pi'(i)$ and, for all agents $j\neq i$, it holds that $\pi(j)\setminus\{i\} = \pi'(j)\setminus\{i\}$.

Since agents are assumed to be rational, agents only engage in a unilateral deviation if it makes them better off, i.e., $\partition'(i)\succ_i \partition(i)$.
Any partition in which no such deviation is possible is called \emph{Nash stable (NS)}.

This type of deviation can be refined by additionally requiring that no agent in the welcoming coalition is worse off when agent $i$ joins. Formally,
a unilateral deviation performed by agent $i$ who moves from coalition $\partition(i)$ to $\partition'(i)$ is an \emph{IS deviation} if $\partition'(i)\succ_i \partition(i)$ and $\partition'(i)=\partition'(j)\succsim_j \partition(j)$ for all agents $j\in\partition'(i)$. 
A partition in which no IS deviation is possible is called \emph{individually stable (IS)}.
Clearly, an NS partition is also IS.\footnote{It is possible to weaken the notion of individual stability even further by also requiring that no member of the \emph{former} coalition of agent $i$ is worse off. The resulting stability notion is called contractual individual stability and guarantees convergence of our dynamics.}
In this article, we focus on dynamics based on IS deviations.
By definition, all terminal states of the dynamics have to be IS partitions.

We are mainly concerned with whether sequences of IS deviations can reach or always reach an IS partition.
If there exists a sequence of IS deviations leading to an IS partition, i.e., a path to stability, then although the agents perform myopic deviations, they can optimistically reach (or can be guided towards) a stable partition.
The corresponding decision problem is described as follows.

 \decpb{\existpb{IS}{[HG]}}{Instance of a particular class of hedonic games [HG], initial partition $\spartition$}{Does there exist a sequence of IS deviations starting from $\spartition$ leading to an IS partition?}

In order to provide some guarantee, we also examine whether \emph{all} sequences of IS deviations terminate.  Whenever this is the case, we say that the dynamics \emph{converges}. 
 The corresponding decision problem is described below.

 \decpb{\convergpb{IS}{[HG]}}{Instance of a particular class of hedonic games [HG], initial partition $\spartition$}{Does every sequence of IS deviations starting from $\spartition$ reach an IS partition?}
 
We mainly investigate this problem via the study of its complement: given a hedonic game and an initial partition, does there exist a sequence of IS deviations that cycles? 

A common idea behind hardness reductions concerning these two problems is to exploit the existence of instances without an IS partition or instances which allow for cycling starting from a certain partition. These can be used to create prohibitive subconfigurations in reduced instances. 

\section{Anonymous Hedonic Games}

\citet{BoJa02a} showed that IS partitions always exist in AHGs under naturally single-peaked preferences, and proved that this does not hold under general preferences, by means of a 63-agent counterexample. Here, we provide a counterexample that only requires 15 agents and additionally satisfies general single-peakedness. Note that smaller counterexamples were repeatedly asked for in the literature \citep{Ball04a,BoEl20a}.

\begin{proposition}\label{prop:noIS-AHG}
There may not exist an IS partition in AHGs even when $n=15$ and the agents have strict and generally single-peaked preferences.
\end{proposition}

\begin{proof}
Let us consider an AHG with 15 agents with the following (incompletely specified) preferences ([\dots] denotes an arbitrary order over the remaining coalition sizes). The preferences for agents $5$ through $15$ are identical.

\medskip
{\centering
\begin{tabular}{r*{15}{c}}
$1:$ & $2$ & $\succ$ & $3$ & $\succ$ & $13$ & $\succ$ & $12$ & $\succ$ & $1$ & $\succ$ & $[\dots]$ \\
$2:$ & $13$ & $\succ$ & $3$ & $\succ$ & $2$ & $\succ$ & $1$ & $\succ$ & $12$ & $\succ$ & $[\dots]$ \\
$3,~4:$ & $3$ & $\succ$ & $2$ & $\succ$ & $1$ & $\succ$ & $[\dots]$ \\
$5,\dots,15:$ & $13$ & $\succ$ & $12$ & $\succ$ & $15$ & $\succ$ & $14$ & $\succ$ & $11$ & $\succ$ & $10$ & $\succ$ & $\dots$ & $\succ$ & $1$ \\
\end{tabular}}
\medskip

These preferences can be completed to be generally single-peaked with respect to axis $1>2>3>13>12>15>14>11>10>\dots>4$. 

Note that in an IS partition,
\begin{enumerate}[label=\textit{(\roman*)}]
\item agents $3$ and $4$ are in a coalition of size at most 3: Otherwise, they prefer to deviate to be alone.
\item agents $5$ to $15$ are in the same coalition:  Suppose, for the sake of contradiction, that agents $5$ to $15$ are not in the same coalition. 
By $(i)$, at most two of them are with agent $3$ and at most two of them with agent $4$. 
Agents $1$ and $2$ cannot be in a coalition of size $12$ or $13$ unless at least $10$ of the $11$ agents $5$ to $15$ are with them (agents $3$ and $4$ cannot be in such a big coalition by $(i)$). In this case, any remaining agent from $5$ to $15$ would be with them because she would join them, otherwise. Hence, the assertion is true.
Therefore, we may assume that agents $1$ and $2$ are not in a coalition of size larger than $3$. Otherwise, they would deviate to be alone.
It follows that at most two agents from $\{5,\dots,15\}$ are with agent $1$ and at most two of them with agent $2$.
Then, in the worst case, there remain only three agents within $\{5,\dots,15\}$ who are not in a coalition with agents $1$, $2$, $3$, or $4$.
These three agents cannot enter into the other coalitions but they prefer to group together, forming a coalition of size three.
Afterwards, all the agents from $\{5,\dots,15\}$ that are in coalitions with agents $1$, $2$, $3$ or $4$ will deviate to join them because they prefer to be in bigger coalitions, and they can benefit from a coalition of size at least four by joining these remaining agents, whereas they are blocked in a coalition of size at most three, a contradiction.
\item agents $3$ and $4$ are in the same coalition: Suppose for the sake of contradiction that agents $3$ and $4$ are not in the same coalition. 
By $(i)$, none of them can belong to the big coalition containing the agents $5$ to $15$, which exists according to $(ii)$. 
Moreover, if they are both alone, then they have an incentive to group together, contradicting the stability.
Therefore, at least one of them must form a coalition with agents $1$ or $2$. 

If agents $1$ and $2$ are both with agent $3$ (or $4$) and agent $4$ (or $3$) is alone, then agent $1$ has an incentive to leave the coalition $\{1,2,3\}$ (or $\{1,2,4\}$) to join agent $4$ (or $3$), contradicting the stability.
Now, consider the case where one of agents $3$ and $4$, say agent $4$, is alone. If agent $3$ forms a coalition with agent $2$, then agent $4$ would join them by an IS deviation. If agent $3$ forms a coalition with agent $1$, then agent $2$ is in a coalition of size $12$ or size $1$ and would join agent $4$ by an IS deviation.

Therefore, each agent among $1$ and $2$ must be with either agent~$3$ or agent $4$. 
But, in such a case, the agent among $3$ and $4$, say $3$, who is with agent $1$ will move to the coalition with agent $2$ and agent $4$, contradicting the stability.
Therefore, agents $3$ and $4$ must be in the same coalition.
\item agents $1$ and $2$ cannot be both alone: Otherwise, they would deviate to group together.
\end{enumerate}

From the previous observations, we get that 
agents $3$ and $4$ must be together in a coalition, 
while agents $5$ to $15$ must be together in another coalition. 
The remaining question concerns the coalitions to which agents $1$ and $2$ belong.
It is not possible that both agents $1$ and $2$ are in a coalition with agents $3$ and $4$, otherwise it would contradict condition $(i)$.
If one agent among agents $1$ and $2$ is alone and the other one is with agents $5$ to $15$, then the alone agent can deviate to join them, contradicting the stability.
The remaining possible partitions 
are present in the following cycle of IS deviations (the deviating agent is written on top of the arrows).

{\centering\footnotesize
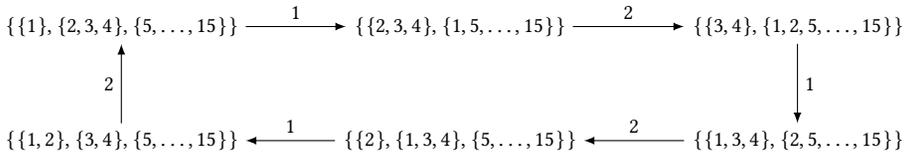
\begin{tikzpicture}
\node (p1) at (-4.5,0) {$\{\{1\},\{2,3,4\},\{5,\dots,15\}\}$};
\node (p2) at (0,0) {$\{\{2,3,4\},\{1,5,\dots,15\}\}$};
\node (p3) at (4.5,0) {$\{\{3,4\},\{1,2,5,\dots,15\}\}$};
\node (p4) at (4.5,-1.5) {$\{\{1,3,4\},\{2,5,\dots,15\}\}$};
\node (p5) at (0,-1.5) {$\{\{2\},\{1,3,4\},\{5,\dots,15\}\}$};
\node (p6) at (-4.5,-1.5) {$\{\{1,2\},\{3,4\},\{5,\dots,15\}\}$};
\draw[arrow] (p1) -- node[midway,above] {1}  (p2);
\draw[arrow] (p2) -- node[midway,above] {2}  (p3);
\draw[arrow] (p3) -- node[midway,right] {1}  (p4);
\draw[arrow] (p4) -- node[midway,above] {2}  (p5);
\draw[arrow] (p5) -- node[midway,above] {1}  (p6);
\draw[arrow] (p6) -- node[midway,left] {2}  (p1);
\end{tikzpicture}\par}

Hence, there is no IS partition in this instance.
\end{proof}

The construction in Proposition~\ref{prop:noIS-AHG} does not seem to leave room for improvements, and we conjecture that the counterexample may even be minimal, that is, an IS partition always exists when $n<15$. However, even when $n<15$ and IS partitions do exist, there may still be cycles in the dynamics.

\begin{proposition}\label{prop:cycleAHG-singleton}
The dynamics of IS deviations may cycle in AHGs even when starting from the \singleton or grand coalition, preferences are strictly generally single-peaked, and $n = 7$.
\end{proposition}

\begin{proof}
Let us consider an AHG with 7 agents with the following (incompletely specified) preferences ([\dots] denotes an arbitrary order over the remaining coalition sizes).

{\centering
\begin{tabular}{r*{11}{c}}
$1:$ & $2$ & $\succ$ & $3$ & $\succ$ & $5$ & $\succ$ & $4$ & $\succ$ & $1$ & $\succ$ & $[\dots]$ \\
$2:$ & $5$ & $\succ$ & $3$ & $\succ$ & $2$ & $\succ$ & $1$ & $\succ$ & $4$ & $\succ$ & $[\dots]$ \\
$3,4:$ & $3$ & $\succ$ & $2$ & $\succ$ & $1$ & $\succ$ & $[\dots]$ \\
$5,6,7:$ & $5$ & $\succ$ & $4$ & $\succ$ & $3$ & $\succ$ & $2$ & $\succ$ & $1$ & $\succ$ & $[\dots]$ \\
\end{tabular}\par}

They can be completed to be generally single-peaked with respect to the axis $1>2>3>5>4>6>7$.
We represent below a cycle of IS deviations. 

{\centering\footnotesize
\begin{tikzpicture}
\node (p1) at (-4,0) {$\{1,2\},\{3,4\},\{5,6,7\}$};
\node (p2) at (0,0) {$\{1\},\{2,3,4\},\{5,6,7\}$};
\node (p3) at (4,0) {$\{2,3,4\},\{1,5,6,7\}$};
\node (p4) at (4,-1) {$\{3,4\},\{1,2,5,6,7\}$};
\node (p5) at (0,-1) {$\{1,3,4\},\{2,5,6,7\}$};
\node (p6) at (-4,-1) {$\{2\},\{1,3,4\},\{5,6,7\}$};
\draw[arrow] (p1) -- node[midway,above] {$2$}  (p2);
\draw[arrow] (p2) -- node[midway,above] {$1$}  (p3);
\draw[arrow] (p3) -- node[midway,right] {$2$}  (p4);
\draw[arrow] (p4) -- node[midway,above] {$1$}  (p5);
\draw[arrow] (p5) -- node[midway,above] {$2$}  (p6);
\draw[arrow] (p6) -- node[midway,left] {$1$}  (p1);
\end{tikzpicture}\par}
This cycle can be reached from the \singleton or the grand coalition. Indeed, the partition $\{\{1,2\},\{3,4\},\{5,6,7\}\}$ can be reached from the \singleton by forming each coalition. It can also be reached by the grand coalition by having agents $1$, $2$, $3$, and $4$ leave and form their desired coalitions. 
\end{proof}

Note that $\{\{1\},\{3,5,6\},\{2,4,7\}\}$ is an IS partition in the example of the previous proposition.

We know that it is \np-complete to recognize instances for which an IS partition exists in AHGs, even for strict preferences~\citep{Ball04a}.
We prove that both checking the existence of a sequence of IS deviations ending in an IS partition and checking convergence are hard. 

\begin{restatable}{theorem}{existAHG}\label{thm:hardness-existencepath-AHG}
\existpb{IS}{AHG} is \np-hard and \convergpb{IS}{AHG} is \conp-hard, even for strict preferences.
\end{restatable}

However, these hardness results do not hold under naturally single-peaked preferences, even if preferences may be weak.
Indeed, we show in the next proposition that \emph{every} sequence of IS deviations is finite under such a restriction.
We thus
complement a result by \citet{BoJa02a} 
who have proved that there always exists an IS partition if the preferences are naturally single-peaked with a dynamical version. \citet{BoJa02a} provide a constructive proof involving specific IS deviations starting
from the \singleton which constructs an IS partition that is additionally weakly Pareto-optimal. By contrast, dynamics can reach every IS partition, and therefore may end up in partitions that are not Pareto-optimal.

\begin{theorem}\label{thm:convAHG}
The dynamics of IS deviations always converges to an IS partition in AHGs for naturally single-peaked preferences.
\end{theorem}

\begin{proof}
	Let an AHG be given with naturally single-peaked preferences. Assume for contradiction that there exists a cycle of IS deviations. We will identify a specific deviation within this cycle that cannot be repeated throughout the execution of the dynamics, obtaining a contradiction. 
	
	Let $C$ be a coalition within the cycle of smallest cardinality such that there exists an agent $d\in C$ that performs a deviation leaving $C$. Such a coalition exists, because there are only finitely many different deviations performed in the cycle. We will argue that it is impossible to alter the coalition $C\setminus \{d\}$ within the execution of the cycle. Hence, the coalition $C$ will never be reached again. 
	
	First, our assumption of minimality implies that $C\setminus \{d\}$ cannot be altered by having an agent leave this coalition. Therefore, we have to show that it cannot happen that an agent ever joins $C\setminus \{d\}$. Assume for contradiction that there exists an agent $x$ that joins $C\setminus \{d\}$. Using our minimality assumption again, the deviation of agent $x$ when joining $C\setminus \{d\}$ must originate from a coalition $C'$ with $|C'| > |C|$. Hence, agent $x$ deviates towards a smaller coalition. Hence, single-peakedness implies that the peak $p_x$ of agent $x$ must satisfy $p_x < |C'|$. In particular, it follows by single-peakedness that it cannot be the case that $y\succ_x^S |C'|$ for any $y \ge |C'|$. 
	
	We claim that, within the cycle, it is impossible that agent~$x$ ever reaches a coalition of size at least $|C'|$ again. To see this, let $C_k$ be the $k$-th coalition that $x$ is part of after leaving $C'$, i.e., $C_1 = (C\setminus \{d\})\cup \{x\}$ and $C_{k+1}$ evolves from $C_k$ by having some agent join or leave $C_k$, or $C_{k+1}$ is the new coalition of $x$ if $x$ performs a deviation to leave $C_k$. 
	We will show by induction over $k$ that, for every $k$, $|C_k| < |C|$ or $C_k\succ_x C'$. Since $C_1\succ_x C'$, the claim is true for $k = 1$. Now, let $k\ge 1$ and assume that $|C_k| < |C|$ or $C_k\succ_x C'$. Consider first the case that $|C_k| < |C|$. By our minimality assumption, $x$ is not allowed to perform a deviation and therefore $C_{k+1}$ evolves by having an agent leave or join $C_k$. Also by the minimality assumption, no other agent may leave the coalition. If an agent joins the coalition, then the size remains to be smaller than $|C|$, or is exactly $|C|$ and we already know that $|C|\succ_x^S |C'|$. 
	
	It remains to consider the case that $C_k\succ_x C'$. If $C_{k+1}$ forms via a deviation of agent~$x$, then $C_{k+1}\succ_x C_k \succ_x C'$ and the claim is true. If some agent joins $C_k$ to form $C_{k+1}$ then $x$ has to approve this, and we can conclude that $C_{k+1}\succsim_x C_k \succ_x C'$. Hence, it remains to consider the case that some agent leaves $C_k$ and the remaining coalition is $C_{k+1}$. 
	As we have argued above, $|C_k|\succ_x^S |C'|$ implies that $|C_k| < |C'|$. If $|C_{k+1}| < |C|$, then the induction hypothesis is true for $k+1$. If $|C_{k+1}| = |C|$, then $C_{k+1}\succ_x C'$, and the induction hypothesis is also true. 
	Finally, it remains the case that $|C| < |C_{k+1}| = |C_k| - 1 < |C_k| < |C'|$, and we know that $|C|\succ_x^S |C'|$ and $|C_k|\succ_x^S |C'|$. 
	Hence, single-peakedness implies that $C_{k+1}\succ_x C'$. This completes the proof of the induction hypothesis.
	
	It follows that agent~$x$ cannot reach a coalition of size at least $|C'|$ again, a contradiction. 
	Hence, there cannot be an agent joining $C\setminus \{d\}$. This shows that $C\setminus \{d\}$ can never be altered again, our final contradiction.
\end{proof}

Our final goal in this section is to provide a polynomial bound on the running time of the dynamics. Unfortunately, our proof relies on strict preferences, leaving the case of weak preferences as an interesting open problem. Even the proof under strict preferences needs far more sophisticated methods than the proof of existence. The key idea is to distinguish deviations towards a smaller and a larger coalition, and to make use of a potential function that aggregates values for agents and coalitions to bound the number of deviations towards a larger coalition by $n^2$. Using a second, much simpler potential function yields an overall polynomial running time. 

\begin{theorem}\label{thm:fastConvAHG}
The dynamics of IS deviations always converges to an IS partition in AHGs for strict naturally single-peaked preferences in $\mathcal O(n^3)$ steps.
\end{theorem}

\begin{proof}
	Let an AHG be given with strict and naturally single-peaked preferences where the peak of agent $j$ is at position $p_j$. Consider a sequence of IS deviations starting at some initial partition $\spartition$. Assume that the deviations lead to the sequence $(\partition_k)_{k = 0}^m$ where, for $k = 0,\dots, m-1$, $\partition_{k+1}$ evolves from $\partition_k$ through an IS deviation of agent $d_k$. We call a deviation an R-move (or L-move) if $d_k$ deviates towards a larger (or smaller) partition. The main part of the proof provides a bound of $n^2$ for the number of R-moves. As we will see, this implies that there are at most $n^3$ L-moves.

	The idea is to define a potential function that is based on a value $v^k_j$ for each agent $j\in N$ and a value $v^k_C$ for each coalition $C\in \partition_k$. This potential function will not only depend on the partition $\partition_k$, but also on the starting partition and the specific sequence of deviations to derive $\partition_k$. It will be increased strictly during an R-move and will not decrease during an L-move. We also need to keep track of the last agent $l^k_C$ that entered a coalition $C\in \partition_k$ if this agents plays a `special role' within her coalition.\footnote{Basically, we must be careful not to increase the potential function by too much after an R-move. If an agent performs an R-move, she might land up far right of her peak, and we keep track of this possibility by labeling the agent as `special'. Intuitively, by the strictness and single-peakedness of preferences, only the last agent that has joined some coalition can be special.}
	The potential function will have a close relationship to the peaks of agents. It maintains the invariant that an agent has a value that is always smaller than her peak (but also reflects her coalition size). 
	
	Let us define the values that lead to the potential function.
	Initially, define $v^0_j = 0$ for all agents $j\in N$, and $v^0_C = 0$ for all coalitions $C\in \spartition$. Also, there is no last agent that entered a coalition so far, so we initiate $l^0_C = \bot$. Now, assume that we transition from partition $\partition_{k}$ to partition $\partition_{k+1}$ through an IS deviation of agent $d_k$. Denote $D_k = \partition_{k}(d_k) \setminus \{d_k\}$ and $E_k = \partition_{k+1}(d_k)$. Note that since the only change in the partition is caused by agent~$d_k$, this corresponds exactly to the two new coalitions in $\partition_{k+1}$ compared to $\partition_{k}$. 
	Also, denote $l_k = l^k_{\partition_{k}(d_k)}$, which will be the last agent that entered $\partition_{k}(d_k)$ (unless $l^k_{\partition_{k}(d_k)} = \bot$, in which case such an agent does either not exist or its identity is unimportant for the updates of the values). 
	We specify first the updates that are done independently of the kind of deviation.
	
	We set $v^{k+1}_j = v^{k}_j$ for all $j\in N \setminus (D_k \cup E_k)$, i.e., the value of agents not involved in the deviation does not change. Similarly, we set $v^{k+1}_C = v^{k}_C$ and $l^{k+1}_C = l^{k}_C$ for all coalitions $C\in \partition_k \cap \partition_{k+1} = \partition_{k+1} \setminus \{D_k,E_k\}$. The updates for the last agents do not depend on the kind of move. Set $l^{k+1}_{E_k} = d_k$, and set $l^{k+1}_{D_k} = \bot$ if $l_k = d_k$ and $l^{k+1}_{D_k} = l_k$, otherwise. 
	
	It remains to specify new values for the agents which are part of one of the coalitions involved in the deviation, and the values of these coalitions. The intuition for defining the agent and coalition values is as follows. The value of an agent is always strictly smaller than her peak. In addition, it represents a size that is preferred at most as much as an agent's current coalition size. 
		
	We first consider the deviating agent. 
	If she performs an R-move, then an appropriate bound for her value is the size of the abandoned coalition. By 	this, we do not increase her value by too much if she ends up far to the right of her peak. 
	In case of an L-move, the deviating agent can only be the last agent who joined the coalition, unless she abandons a coalition that was never joined by an agent (this intuition follows from the third invariant in the formal proof below). 
	Hence, we assign the `right' value by maintaining the current value.

	Formally, we set $v^{k+1}_{d_k} = |\partition_{d_k}(\partition_k)| = |D_k| + 1$ if the deviation was an R-move and  $v^{k+1}_{d_k} = v^{k}_{d_k}$ if it was an L-move. 
	
	Next, we consider the agents which are joined. Their consent to increase the coalition size together with strictness and single-peakedness of the preferences imply that they move towards their peak. 
	We set $v^{k+1}_j = |E_k| - 1$ for all $j\in E_k\setminus \{d_k\}$, which represents a bound that is strictly smaller than their peak.

	This leads us to the coalition value of $E_k$. The role of this value is to anticipate behavior of the deviating agent that goes against the drift of coalition sizes to the right. In all cases, we simply update $v^{k+1}_{E_k} = v^{k+1}_{d_k}$, which will be sufficient to serve this purpose. The intuition of this value is that we reach a critical point in the deviation when a coalition has been abandoned so often that its size reaches the value of the special last agent. In this case, the special last agent cannot be in the situation anymore where she can be too far on the right of the peak. We can then `spread' the coalition value among all agents (of which there is exactly the right amount).

	The most complicated part of the update rules is to specify new values for the agents in the abandoned coalition. Here, we have to update very differently dependent on the behavior of the dynamics. We distinguish several cases that we will reference to in the following proof. The first considers the case where the deviating agent performs an L-move. The four other cases consider different situations during an R-move. The second and third case consider two simple situations where either the abandoned agents did not change their value since initialization, and where $d_k = l_k$, i.e., the deviating agent is the special last agent in the abandoned coalition. Note that the update rules are the same for the first three cases whereas the values of the agents can still be very different. The final two cases consider the situation where $d_k \neq l_k$. The update rules distinguish whether we reach a threshold specified by the value of $l_k$ in the forth case, or whether not in the fifth case. Note that the case distinction is exhaustive because of invariant~(\ref{cl:notbotcoal}) in the proof of Claim~\ref{cl:AHGaux}.
	
	Note that the forth case covers the case where $l_k = \bot$ but agents in the abandoned coalition already have a value because it .
	\begin{itemize}
		\item[(i)] If the deviation was an L-move, we update $v^{k+1}_j = v^{k}_j$ for all $j\in D_k$ and $v^{k+1}_{D_k} = 0$.
\end{itemize}
In all other cases, we assume that the deviation was an R-move.
\begin{itemize}	
		\item[(ii)] If $v^{k}_j = 0$ for all $j\in D_k$, we maintain $v^{k+1}_j = v^{k}_j$ for all $j\in D_k$ and $v^{k+1}_{D_k} = 0$. 
		\item[(iii)] If $d_k = l_k$, we update $v^{k+1}_j = v^{k}_j$ for all $j\in D_k$ and $v^{k+1}_{D_k} = 0$.
\end{itemize}
For the remaining cases, we assume that $d_k \neq l_k$. 
\begin{itemize}	

		\item[(iv)] 
		Consider the case that $l_k = \bot$. Then, set $v_j^{k+1} = |D_k|$ for all $j\in D_k$, $v^{k+1}_{D_k} = 0$, and maintain $l^{k+1}_{D_k} = \bot$.
		
		\item[(v)] Otherwise, 
		there exists an agent $l_k\in N$ and $v^{k}_{l_k} \le |D_k|$ (shown in the proof of Claim~\ref{cl:AHGaux}).
		If $v^{k}_{l_k} = |D_k|$, then set $v_j^{k+1} = |D_k|$ for all $j\in D_k$, $v^{k+1}_{D_k} = 0$, and $l^{k+1}_{D_k} = \bot$ (we update this last agent \emph{again} because all agents play the same role now within $D_k$). 
		\item[(vi)] If $v^{k}_{l_k} < |D_k|$, set $v^{k+1}_j = |D_k|-1$ for all $j\in D_k\setminus \{l_k\}$, $v^{k+1}_{l_k} =v^{k}_{l_k}$, and $v^{k+1}_{D_k} = v^{k}_{D_{k}\cup\{d_k\}}$ (the last agent of $D_{k}$ does not need a second update since it still plays a special role).
	\end{itemize}
	
	We are ready to specify our potential function.
	Given a partition $\partition_k$ that occurs in the dynamics, define its potential $\Lambda(\partition_{k}) = \sum_{j \in N} v^{{k}}_j + \sum_{C\in \partition_k}v^{k}_C$. Again, note that this potential can depend both on the starting partition and the specific sequence of deviations. The core of the proof is the following claim.
	\begin{claim}\label{cl:AHGaux}
		For all $k = 1,\dots, m$, $\Lambda(\partition_{k})\ge \Lambda(\partition_{k-1})$. If $\partition_{k+1}$ evolved from $\partition_k$ by an R-move, then $\Lambda(\partition_{k+1}) > \Lambda(\partition_{k})$.
	\end{claim}
	
	\begin{proof}
		\renewcommand\qedsymbol{$\vartriangleleft$}
	We will prove this main claim along with the following useful invariants by induction over $k = 0,\dots, m$:
	
	\begin{enumerate}
		\item \label{cl:coalboun} For all $C\in \partition_k$, $v_C^{k} \le |C|-1$. 
		\item \label{cl:agposcoal} For all $C\in \partition_k$ and $j\in C$, if $v_C^{k} > 0$, then $l_C^k \neq \bot$ and $v_j^{k} \le |C|-1$.
		\item \label{cl:peakinvar} For all agents $j\in N$ with $v_j^k>0$, $v_j^{k}\preceq_j |\partition_k(j)|$ and $p_j > v^{k}_j$.
		\item \label{cl:notbotcoal} For all $C\in \partition_k$ with $l^k_C\neq \bot$, $v^k_C = v^k_{l^k_C}$ or $v^k_C = 0$. Moreover, in this case, $v^k_{l^k_C}\le |C|-1$.
		\item \label{cl:botcoal}  For all $C\in \partition_k$ with $l_C^k = \bot$, then $v_j^{k} = |C|$ for all $j\in C$ or $v_j^{k} = 0$ for all $j\in C$. Moreover, in this case,  $v_C^k = 0$.
	\end{enumerate}
	
	For $k=0$, the main claim is vacant, and the four invariants hold by our initialization of the agent and coalition values. So assume that all of them are true for some $0\le k<m$. We use the notation for the agents $d_k$ and $l_k$, and the coalitions $D_k$ and $E_k$ as in the description of the updates of the values. We will proof that all invariants are true for partition $k+1$.
	
	\begin{enumerate}
	\item 	The first invariant follows from the update rules and by induction. Let us provide the details for the two affected coalitions. 
	Specifically, $v_{D_k}^{{k+1}} = 0$, unless we are in the last case of the update rule, where $v_{D_k}^{{k+1}} \overset{\mathrm{(v)}}{=} v_{D_k\cup \{d_k\}}^{{k}} \overset{(\ref{cl:notbotcoal})}{=} v_{l_k}^{{k}} \overset{\mathrm{(v)}}{\le} |D_k|-1$.
	
	 Now, let us consider coalition $E_k$. If $v_{E_k}^{k+1} = 0$, then the invariant holds. Otherwise, if $d_k$ performed an R-move, then $v^{k+1}_{E_k} = |D_k|+1\le |E_k| - 1$. On the other hand, assume that $d_k$ performed an L-move. Recall that by the update rule for $E_k$, $v_{E_k}^k = v^k_{d_k}$. Then, it holds that $v^k_{d_k} = 0$ (and we are done), or $d_k$ was the last agent who joined $\partition_k(d_k)$ (other agents in $\partition_k(d_k)$ can only perform R-moves due to strict preferences and single-peakedness). Assume for contradiction that $v_{E_k}^k \ge |E_k|$. If $v_{E_k}^k = |E_k|$, then induction for invariant~(\ref{cl:peakinvar}) yields $|E_k| = v_{d_k}^k \preceq_{d_k} |\partition_k(d_k)|$, contradicting the fact that $d_k$ must improve her coalition after her deviation. Hence, using invariant~(\ref{cl:peakinvar}) again, we have that $p_{d_k} > v^k_{d_k} > |E_k|$ and $v^k_{d_k}\prec_{d_k} p_k$. Hence, single-peakedness implies that $|E_k| \prec_{d_k} v^k_{d_k}$, and the deviation was again not improving, a contradiction. 
	\item The second invariant follows by induction and the update rules. In particular, it holds for the agent $d_k$ if $d_k$ performed an L-move to join $E_k$, because then $v^{k+1}_{d_k} = v^{k+1}_{E_k}$ and the invariant follows from the first invariant. The value of the coalition $D_k$ will either be set to $0$ or the invariant will be maintained for all agents in $D_k$.
	\item The third invariant holds by induction for the agents who do not change their value. For the agents in $E_k \setminus \{d_k\}$, it holds by definition of an IS deviation, because they improve when agent~$d_k$ joins. The same is true for $d_k$ if she performed an R-move because then $v_{d_k}^{k+1}=|D_k|+1$. Otherwise, if $d_k$ performs an L-move, we can apply induction, because she improved her utility. 
	
	Finally, we have to consider the agents in $D_k$. The invariant holds if $v_j^k = 0$ for all $j\in D_k$. Otherwise, all agents $j\in D_k\setminus \{l_k\}$ have improved when the last agent joined and therefore $p_j \ge |D_k| + 1 > |D_k|$. By the update rules, $v^{k+1}_j \le |D_k|$. In particular, single-peakedness implies that $v_j^{k+1}\preceq_j |D_k| = |\partition_{k+1}(j)|$. Finally, for $l_k$, the argument is the same if $v_{l_k} = |D_k|$ and we can apply induction (and single-peakedness) if $v_{l_k} < |D_k|$. Note that the case $v_{l_k} > |D_k|$ is excluded by induction for invariant~(\ref{cl:notbotcoal}).
	\item 	The forth invariant is immediate for all coalitions in $\partition_{k+1}\setminus \{D_k\}$. The only case where $v^{k+1}_{D_k}\neq 0$ is case $(v)$ of the update rules. Then, $v^{k+1}_{l_k} =v^{k}_{l_k}$, and $v^{k+1}_{D_k} = v^{k}_{D_{k}\cup \{d_k\}}$.	
	 If $v^{k+1}_{l_k}\neq 0$, then $v^{k}_{l_k}\neq 0$, and induction yields for this case that $v^{k+1}_{D_k} = v^{k}_{D_{k}\cup\{d_k\}} = v^k_{l_k} = v^{k+1}_{l_k}$. Moreover, in case $(v)$ of the update rule, it holds that $v^k_{l_k} \le |D_k| - 1$, which achieves the last part of this invariant.
	\item Again, the only coalitions to consider are $D_k$ and $E_k$. Since, $l_{E_k}^{k+1} = d_k$, the invariant is vacant for this coalition. For $D_k$, the invariant can only apply in cases (ii), (iv), and (v) of the update rules. For all of these cases, it holds that $v_{D_k}^{k+1} = 0$, and $v_j^{k+1} = |D_k|$ for all $j\in D_k$ or $v_j^{k+1} = 0$ for all $j\in D_k$.
	\end{enumerate}
	
	An immediate consequence of invariant~(\ref{cl:peakinvar}) is that for all agents $j\in N$, it holds that $v_j^{k} \le |\partition_k(j)|$. Indeed, otherwise $|\partition_k(j)| < v_j^{k} < p_j$ and since $v_j^k \prec_j p_j$, single-peakedness and strictness would imply $|\partition_k(j)| \prec_j v_j^{k}$, contradicting invariant~(\ref{cl:peakinvar}). We call this observation~$(\Delta)$.
	
	Finally, we can show our main claim. We can restrict our attention to coalitions $D_k$ and $E_k$ and agents within these. Note that either $v_{E_k\setminus \{d_k\}}^{k} = 0$, or, by invariant~(\ref{cl:coalboun}), $v_{E_k\setminus \{d_k\}}^{k}\le |E_k|-2$. Moreover, for every $j\in E_k\setminus\{d_k\}$, it holds that  $v_j^{{k+1}} = |E_k| - 1$ according to the update rule, and $v_j^{k} \le (|E_k|-1)-1$ according to invariant~(\ref{cl:agposcoal}) if $v_{E_k\setminus \{d_k\}}^{k} > 0$. Thus, in this case $v_j^{{k+1}} - v_j^{k} \ge 1$ for all $j\in E_k\setminus \{d_k\}$. Hence, in every case
	\begin{align*}
		\sum_{j\in E_k\setminus \{d_k\}}v^{{k+1}}_j-v^{k}_j - v_{E_k\setminus \{d_k\}}^{k} > 0\text.\tag{$*$}
	\end{align*} 
	
	We make a case distinction according to the different update cases for agents in~$D_k$.
	We start with with the cases covering situations with $d_k = l_k$. 
	
	\begin{enumerate}
		\item[(i)] 	If an L-move was performed, then $d_k = l_k$ or the agents in $D_k\cup \{l_k\}$ did not participate in a deviation, yet. In either case, $v_{E_k}^{{k+1}} \ge v^k_{{D_k}\cup\{d_k\}}$ (using the update rules for the values of $l_k$ and $E_k$, and the forth invariant). Hence,
		\begin{align*}
			\Lambda(\partition_{k+1})- \Lambda(\partition_{k}) \overset{(*)}{\ge} v_{D_k}^{{k+1}} + \sum_{j\in D_k\cup \{d_k\}} v^{{k+1}}_j-v^{k}_j \overset{(i)}{=} 0\text.
		\end{align*} 
	\end{enumerate}

	Assume now that an R-move was performed. Then, it holds that $v_{E_k}^{{k+1}} = |D_k|+1$ and by the invariant~(\ref{cl:coalboun}), $v_{D_k\cup\{d_k\}}^k \le |D_k|$. Hence,
	\begin{align*}
		v_{E_k}^{{k+1}} - v_{D_k\cup\{d_k\}}^k > 0\text.\tag{$\mathit{**}$}
	\end{align*} 
	
	\begin{enumerate}
		\item[(ii)] If the agents in $D_k$ had value $0$, then $\Lambda(\partition_{k+1})- \Lambda(\partition_{k})\ge v_{d_k}^{{k+1}} + v_{E_k}^{{k+1}} = 2 (|D_k| + 1) > 0$.

		\item[(iii)] Next, assume that the agents in $D_k$ do not have value $0$, and that $d_k = l_k$. 
		Then, 
		\begin{align*}
			\Lambda(\partition_{k+1})- \Lambda(\partition_{k}) \overset{(*)}{\ge} \underbrace{v_{d_k}^{{k+1}} - v_{d_k}^{{k}}}_{\ge\,0}+ \underbrace{v_{D_k}^{{k+1}}}_{=\, 0\textnormal{ by }  \mathrm{(iii)}} + \underbrace{v_{E_k}^{{k+1}}  - v_{D_k\cup\{d_k\}}^{{k}}}_{>\, 0\textnormal{ by }  (**)} > 0\text.
		\end{align*}
	\end{enumerate}
	
	Now, consider the case that $d_k\neq l_k$. 
	
	\begin{enumerate}
		\item[(iv)] Assume that $l_k = \bot$. Then, by invariant~(\ref{cl:botcoal}), $v_{D_k\cup\{d_k\}}^k = 0$. Moreover, by the update rule, for all $j\in D_k$, $d_j^{k+1} = |D_k|$, while observation $(\Delta)$ yields $d_j^{k} \le |D_k|+1$. Hence, 
		\begin{align*}
			\sum_{j\in D_k} d_j^{k+1} - d_j^{k} \ge -|D_k|\text.
		\end{align*}
	
	Hence,
	\begin{align*}
		\Lambda(\partition_{k+1})- \Lambda(\partition_{k}) \overset{(*)}{\ge} \underbrace{v_{d_k}^{{k+1}} - v_{d_k}^{{k}}}_{\ge\,0}+ \underbrace{v_{D_k}^{{k+1}}}_{=\, 0\textnormal{ by } \mathrm{(iv)}} + \underbrace{v_{E_k}^{{k+1}}}_{=\, |D_k|+1} + \underbrace{\sum_{j\in D_k} d_j^{k+1} - d_j^{k}}_{\ge\, -|D_k|} - \underbrace{v_{D_k\cup\{d_k\}}^{{k}}}_{=\, 0\textnormal{ by }  (\ref{cl:botcoal})} > 0\text.
	\end{align*}
	\end{enumerate}

	Now assume that $l_k\neq \bot$. By invariant~(\ref{cl:notbotcoal}), $v_{l_k}^{k} \le |D_k|$. We consider the two corresponding cases.
	
	\begin{enumerate}
		\item[(v)] Assume that $v_{l_k}^{k} = |D_k|$. If $v_{D_k\cup\{d_k\}}^k > 0$, then 
		$v_j^{{k+1}} = |D_k| \overset{(\ref{cl:agposcoal})}{\ge} v_j^{{k}}$ for all $j\in D_k$. 
		Therefore,
		\begin{align*}
			\Lambda(\partition_{k+1})- \Lambda(\partition_{k}) \overset{(*)}{\ge} \underbrace{v_{d_k}^{{k+1}} - v_{d_k}^{{k}}}_{\ge\,0}+ \underbrace{v_{D_k}^{{k+1}}}_{=\, 0\textnormal{ by } \mathrm{(v)}} + \underbrace{v_{E_k}^{{k+1}}  - v_{D_k\cup\{d_k\}}^{{k}}}_{>\, 0\textnormal{ by }  (**)} > 0\text.
		\end{align*}

		In the case that $v_{D_k\cup\{d_k\}}^k = 0$, we have 
		$v_j^{{k+1}} = |D_k| \overset{(\Delta)}{\ge} v_j^{{k}} - 1$ for all $j\in D_k$. 
		Hence, 
		\begin{align*}
			v_{E_k}^{k+1} - \sum_{j\in D_k} v_j^{{k+1}} -v_j^{{k}} = v_{E_k}^{k+1} - |D_k| > 0\text.\tag{${***}$}
		\end{align*}
		
		Together,
		\begin{align*}
			\Lambda(\partition_{k+1})- \Lambda(\partition_{k}) \overset{(*),(***)}{>} \underbrace{v_{d_k}^{{k+1}} - v_{d_k}^{{k}}}_{\ge\,0}+ \underbrace{v_{D_k}^{{k+1}}}_{=\, 0\textnormal{ by } \mathrm{(vi)}} - \underbrace{ v_{D_k\cup\{d_k\}}^{{k}}}_{=\, 0} > 0\text.
		\end{align*}
	\item[(vi)] If $v_{l_k}^{k} \neq |D_k|$, then, by invariant~(\ref{cl:notbotcoal}), $v_{l_k}^{k} < |D_k|$. Hence, as in the beginning of case (iv), $\sum_{j\in D_k}v_j^{k}-v_j^{{k+1}}\le |D_k|-1$. Moreover, $v_{E_k}^{{k+1}}= |E_k|-1\ge |D_k| + 1$. Also, according to the update rule for case (vi), $v^{k+1}_{D_k} = v^{k}_{D_{k}\cup\{d_k\}}$. Hence, 
	\begin{align*}
		\Lambda(\partition_{k+1})- \Lambda(\partition_{k}) \overset{(*)}{\ge} \underbrace{v_{d_k}^{{k+1}} - v_{d_k}^{{k}}}_{\ge\,0}+ \underbrace{v_{D_k}^{{k+1}}- v_{D_k\cup\{d_k\}}^{{k}}}_{=\, 0\textnormal{ by } \mathrm{(vi)}} + \underbrace{v_{E_k}^{{k+1}}}_{=\, |D_k|+1} + \underbrace{\sum_{j\in D_k} d_j^{k+1} - d_j^{k}}_{\ge\, -(|D_k|-1)} > 0\text.
	\end{align*}
	\end{enumerate}
	
	This completes the induction.
\end{proof}
\renewcommand\qedsymbol{$\square$}
	Now, note that for all $0\le k\le m$, $\Lambda(\partition_k) \ge 0$ and the potential is integer-valued. Moreover, $\Lambda(\partition_m) = \sum_{j \in N} v^m_j + \sum_{C\in \partition_k}v^m_C = \sum_{C\in \partition_k} (v^m_C +  \sum_{j \in C} v^m_j) \le \sum_{C\in \partition} |C|^2 \le n^2$. We use that, by $(\Delta)$, if $v_C^m = 0$, then $v_j^m \le |C|$ for all $j\in C$, and we also use the first and second invariant. Hence, there are at most $n^2$ R-moves.
	
	To get a global bound on the number of moves, we consider the simple potential that counts the pairs of agents that form common coalitions. Given a coalition $C$, this value is precisely $|C|(|C|-1) / 2$. Define therefore the potential $\Gamma(\partition) = \sum_{C\in \partition}|C|(|C|-1) / 2$. Note that $0\le \Gamma(\partition)\le n(n-1)/2$. Now, every R-move raises the potential $\Gamma$ by at most $n-1$, and every L-move diminishes it by at least $1$. Hence, there can be at most $\Gamma(\partition_0) + n^2(n-1) \le n^3$ L-moves. This bounds the total length of the execution of the dynamics by $n^2 + n^3 = \mathcal O(n^3)$.
\end{proof}

\section{Hedonic Diversity Games}

Hedonic diversity games take into account more information about the identity of the agents, changing the focus from coalition sizes to proportions of given \emph{types} of agents. Following the definition of HDGs, we assume throughout this section that the agent set is always partitioned into sets $R$ and $B$ of red and blue agents, respectively. It is known that IS partitions always exist in HDGs, even without restrictions such as single-peakedness of preferences~\citep{BoEl20a}. 
However, we prove that the dynamics of IS deviations may cycle, even under very strong restrictions. This stands in contrast to empirical evidence for the convergence of dynamics based on extensive computer simulations by \citet{BoEl20a}. To this end, we consider natural restrictions of the preferences, of the starting partitions, and specific selection rules for the performed deviations. We show that most combinations of them still allow for infinite dynamics. Most surprisingly, we can show that the dynamics may cycle even if we start from the \singleton and the preferences are strict and single-peaked.\footnote{This corrects a statement in the conference version of this paper \citep{BBW21a}.} However, if we add an arguably weak selection rule, we obtain convergence of the dynamics. To define this rule, we call a coalition $C\subseteq N$ \emph{homogeneous} if it consists only of agents of one type, i.e., $C\subseteq R$ or $C\subseteq B$. We say that a deviation satisfies \emph{solitary homogeneity} if, whenever the target coalition of the deviator is homogeneous, then it is a singleton coalition. Note that whenever an agent can perform an IS deviation, then she can perform a deviation satisfying solitary homogeneity, simply by forming the homogeneous singleton coalition instead of joining existing homogeneous coalitions. Hence, assuming solitary homogeneity of deviations yields valid selection rules, i.e., whenever a deviation is possible, then solitary homogeneity does not prohibit all possible deviations.

In the second part of this section, we show that combining all considered restrictions leads to convergence of the dynamics. In other words, the IS dynamics may cycle if and only if any of the four properties of Theorem~\ref{thm:cycleHDG} is violated. For the proof of Theorem~\ref{thm:cycleHDG}, we make use of a lemma which already highlights the special role of homogeneous coalitions.

\begin{restatable}{lemma}{HDGhom}\label{lem:HDG-hom}
	Given a set $R_a$ of red (or set $B_a$ of blue) agents (the subscripts indicate that we use these agents as \emph{auxiliary} agents) whose preferences satisfy $\frac 23 \succ 1 \succ \frac 12$ (or $\frac 13 \succ 0 \succ \frac 12$), it is possible to create the homogeneous coalition $R_a$ (or $B_a$) by means of an individual dynamics starting from singleton coalitions.
\end{restatable}

We are ready to prove the theorem.

\begin{restatable}{theorem}{cycleHDG}\label{thm:cycleHDG}
The dynamics of IS deviations may cycle in HDGs even if any three of the following restrictions apply:
\begin{enumerate}
\item preferences are naturally single-peaked,
\item preferences are strict,
\item the starting partition is the \singleton, or
\item all deviations satisfy solitary homogeneity.
\end{enumerate}
\end{restatable}

\begin{proof}
	We provide examples for any triple of the four restrictions. The example where all properties except the condition on the starting partition, and where all properties except deviation selection according to solitary homogeneity are satisfied are closely related. First we show how to deal with the former case. Then, we show how to reach a configuration within the cycle of this case by starting from the singleton coalition. However, for reaching this cycle, some of the performed deviations violate solitary homogeneity. We defer the other two examples to the appendix. 
	
\begin{enumerate}
\item[$(\lnot 3)$] We start with an example of an HDG where all preferences are strict and naturally single-peaked and all agents' deviations satisfy solitary homogeneity. Therefore, let us consider an HDG with 26 agents: 12 red agents and 14 blue agents.
There are four deviating agents: red agents $r_1$ and $r_2$ and blue agents $b_1$ and $b_2$, and four fixed coalitions $C_1$, $C_2$, $C_3$ and $C_4$ such that:
\begin{itemize}
\item $C_1$ contains 2 red agents and 4 blue agents;
\item $C_2$ contains 5 red agents;
\item $C_3$ contains 3 red agents and 2 blue agents;
\item $C_4$ contains 6 blue agents.
\end{itemize} 

The relevant part of the preferences of the agents is given below.\footnote{\label{fot:relevantprefs}Throughout this proof, we only specify the relevant part of preferences. Since agents are only in coalitions with the specified ratios, all missing values can be inserted arbitrarily, possibly respecting single-peakedness. Note that if these preferences are naturally single-peaked, then they can easily be completed by inserting fractions in the right intervals. 
	For instance, the preferences by $b_1$ in the following can be completed as follows. The peak is at $\frac 38$, then we prefer values to the right of the peak over values to the left of the peak. Within the interval $[\frac 38, 1]$, smaller values are preferred to larger values. Within the interval $[0,\frac 38]$, larger values are preferred to smaller values. In particular, the displayed part of the preferences is naturally single-peaked because $\frac 27 < \frac 38 < \frac 57 < \frac 56$.}
\smallbreak
\renewcommand{\arraystretch}{1.25}
\begin{minipage}[c]{0.5\columnwidth}
\begin{center}
\begin{tabular}{rl}
$b_1:$ & $\frac{3}{8} \succ \frac{5}{7} \succ \frac{5}{6} \succ \frac{2}{7}$ \\
$b_2:$ & $\frac{5}{7} \succ \frac{4}{7} \succ \frac{1}{2} \succ \frac{5}{6}$ \\
$r_1:$ & $\frac{4}{7} \succ \frac{1}{4} \succ \frac{1}{7} \succ \frac{2}{3}$ \\
$r_2:$ & $\frac{1}{4} \succ \frac{3}{8} \succ \frac{3}{7} \succ \frac{1}{7}$ \\
\end{tabular}
\end{center}
\end{minipage}
\begin{minipage}[c]{0.5\columnwidth}
\begin{center}
\begin{tabular}{rl}
$C_1:$ & $\frac{3}{8} \succ \frac{3}{7} \succ \frac{1}{3}$ \\
$C_2:$ & $\frac{5}{7} \succ \frac{5}{6} \succ 1$ \\
$C_3:$ & $\frac{4}{7} \succ \frac{1}{2} \succ \frac{3}{5}$ \\
$C_4:$ & $\frac{1}{4} \succ \frac{1}{7} \succ 0$ \\
\end{tabular}
\end{center}
\end{minipage}
\smallbreak
Consider the sequence of IS deviations in Figure~\ref{fig:HDG-cycle-SSPSH} that describes a cycle of the dynamics. 
The four deviating agents of the cycle $b_1$, $b_2$, $r_1$ and $r_2$ are marked in bold and the specific deviating agent between two states is indicated next to the arrows.

\begin{figure}
	\centering
\resizebox{\textwidth}{!}{\scriptsize
\begin{tikzpicture}
\node (1) at (0,0) {\begin{tabular}{|cccc|}\hline $C_1\cup\{\mathbf{b_1},\mathbf{r_2}\}$ & $C_2$ & $C_3 \cup \{\mathbf{b_2}\}$ & $C_4 \cup \{\mathbf{r_1}\}$ \\ \hline 
$3/8$ & $1$ & $1/2$ & $1/7$ \\ \hline\end{tabular}};
\node (2) at ($(1)+(0.5,2)$) {\begin{tabular}{|cccc|}\hline $C_1 \cup \{\mathbf{b_1}\}$ & $C_2$ & $C_3 \cup \{\mathbf{b_2}\}$ & $C_4 \cup \{\mathbf{r_1},\mathbf{r_2}\}$ \\ \hline 
$2/7$ & $1$ & $1/2$ & $1/4$ \\ \hline  \end{tabular}};
\node (3) at ($(1)+(4,4)$) {\begin{tabular}{|cccc|} \hline $C_1 \cup \{\mathbf{b_1}\}$ & $C_2$ & $C_3 \cup \{\mathbf{b_2},\mathbf{r_1}\}$ & $C_4 \cup \{\mathbf{r_2}\}$ \\ \hline 
$2/7$ & $1$ & $4/7$ & $1/7$ \\ \hline \end{tabular}};
\node (4) at ($(2)+(7.5,0)$) {\begin{tabular}{|cccc|} \hline $C_1$ & $C_2 \cup\{\mathbf{b_1}\}$ & $C_3 \cup \{\mathbf{b_2},\mathbf{r_1}\}$ & $C_4 \cup \{\mathbf{r_2}\}$ \\ \hline 
$1/3$ & $5/6$ & $4/7$ & $1/7$ \\ \hline \end{tabular}};
\node (5) at ($(1)+(8.5,0)$) {\begin{tabular}{|cccc|} \hline $C_1 \cup \{\mathbf{r_2}\}$ & $C_2 \cup\{\mathbf{b_1}\}$ & $C_3 \cup \{\mathbf{b_2},\mathbf{r_1}\}$ & $C_4$ \\ \hline 
$3/7$ & $5/6$ & $4/7$ & $0$ \\ \hline \end{tabular}};
\node (8) at ($(1)+(0.5,-2)$) {\begin{tabular}{|cccc|} \hline $C_1\cup\{\mathbf{b_1},\mathbf{r_2}\}$ & $C_2 \cup \{\mathbf{b_2}\}$ & $C_3$ & $C_4 \cup \{\mathbf{r_1}\}$ \\ \hline 
$3/8$ & $5/6$ & $3/5$ & $1/7$ \\ \hline  \end{tabular}};
\node (6) at ($(8)+(7.5,0)$) {\begin{tabular}{|cccc|} \hline $C_1 \cup \{\mathbf{r_2}\}$ & $C_2 \cup \{\mathbf{b_1},\mathbf{b_2}\}$ & $C_3\cup \{\mathbf{r_1}\}$ & $C_4$ \\ \hline 
$3/7$ & $5/7$ & $2/3$ & $0$ \\ \hline \end{tabular}};
\node (7) at ($(1)+(4,-4)$) {\begin{tabular}{|cccc|} \hline $C_1 \cup \{\mathbf{r_2}\}$ & $C_2 \cup \{\mathbf{b_1},\mathbf{b_2}\}$ & $C_3$ & $C_4 \cup \{\mathbf{r_1}\}$ \\ \hline 
$3/7$ & $5/7$ & $3/5$ & $1/7$ \\ \hline \end{tabular}};
\draw[arrow] (1) --node[midway,left]{$r_2$} (2); \draw[arrow] (2) --node[midway,left,above]{$r_1$}  (3); \draw[arrow] (3) --node[midway,right,above]{$b_1$}  (4); \draw[arrow] (4) --node[midway,right]{$r_2$}  (5); \draw[arrow] (5) --node[midway,right]{$b_2$}  (6); \draw[arrow] (6) --node[midway,below]{$r_1$}  (7);
\draw[arrow] (7) --node[midway,below]{$b_1$}  (8); \draw[arrow] (8) --node[midway,left]{$b_2$}  (1);
\end{tikzpicture}}
	\caption{Possibility of cycling of IS dynamics in part $(\lnot 3)$ of Theorem~\ref{thm:cycleHDG}. Here, we consider IS dynamics in HDGs under strict and single-peaked preferences, where all deviations satisfy solitary homogeneity.}\label{fig:HDG-cycle-SSPSH}
\end{figure}

Note that all deviations result in non-homogeneous target coalitions, and therefore satisfy solitary homogeneity.

\item[$(\lnot 4)$] Our next goal is to provide an example of cycling under strict and single-peaked preferences while the starting partition is the \singleton. Our example makes use of the previous example. As a first step, we show, how we can create the coalitions $C_i$ for $i\in [4]$ by starting from the singleton coalition. In a second step, we show how to add the deviators $r_1$, $r_2$, $b_1$, and $b_2$ of the previous example to these coalitions to reach a partition from the cycle. As a consequence, cycling can occur by following the cycle from the first part of the proof. For highlighting the relationship of the two examples, we will use bold face for the relevant part of the preferences of the constructed coalitions.

To create the desired coalitions, we use Lemma~\ref{lem:HDG-hom} to create homogeneous auxiliary coalitions. We assume that we take \emph{new} (auxiliary) agents for every step of every construction.

Now, we show how to create the coalitions $C_i$ for $i\in [4]$ one by one.

\paragraph{Creating $C_2$ and $C_4$}
Given the insights gained in Lemma~\ref{lem:HDG-hom}, it is straightforward to manufacture the homogeneous coalitions $C_2$ and $C_4$. Therefore, define the coalitions $C_2 = \{r_{2,i}\colon i\in [5]\}$ and $C_4 = \{b_{4,i}\colon i\in [6]\}$ together with the following strict and naturally single-peaked preferences.

{\centering
	\renewcommand{\arraystretch}{1.25}
	\begin{tabular}{rl}
		$C_2:$ & $\frac 23 \succ \boldsymbol{\frac 57 \succ \frac 56 \succ 1} \succ \frac 12$\\
		$C_4:$ & $\frac 13 \succ \boldsymbol{\frac 14 \succ \frac 17 \succ 0} \succ \frac 12$\\
	\end{tabular}\par}
\medskip

Since the preferences satisfy the assumptions of Lemma~\ref{lem:HDG-hom}, we can apply it to create $C_2$ and $C_4$.

\paragraph{Creating $C_1$}
Creating the coalition~$C_1$ is also not very difficult. We can simply apply Lemma~\ref{lem:HDG-hom} to form a coalition of all the blue agents, and let the red agents join this coalition. More formally, consider the coalition $C_1 = \{b_{1,1}, b_{1,2}, b_{1,3}, b_{1,4}, r_{1,1}, r_{1,2}\}$ with the following strict and naturally single-peaked preferences.

{\centering
	\renewcommand{\arraystretch}{1.25}
	\begin{tabular}{rl}
		$C_1:$ & $\boldsymbol{\frac 38 \succ \frac 37 \succ \frac 13} \succ \frac 15 \succ 0 \succ \frac 12 \succ 1$\\
	\end{tabular}\par}
\medskip

Note that here and in the following constructions, the considered agents can be of any color and we can just ignore the preference for $0$ or $1$ if we consider a red or blue agent from the set, respectively.

To form $C_1$, we can apply Lemma~\ref{lem:HDG-hom} to form $B_a = \{b_{1,1}, b_{1,2}, b_{1,3}, b_{1,4}\}$. Then, $r_{1,1}$ and $r_{1,2}$ can perform IS deviations to join $B_a$ one after another. This results in the coalition $C_1$, as desired.

\paragraph{Creating $C_3$}
The by far most difficult coalition to create is $C_3$, where we have to combine several steps. The central idea is to apply Lemma~\ref{lem:HDG-hom} to create a sufficiently large homogeneous coalition of auxiliary blue agents. Then, the red agents of the future coalition $C_3$ will join. This is followed by having the blue agents of the former homogeneous coalition abandon the so created coalition. An essential step is to create further coalitions to incentivize them to perform the necessary deviations. Finally, the two blue agents from $C_3$ can join.

To this end, consider the coalition $C_3 = \{b_{3,1}, b_{3,2}, r_{3,1}, r_{3,2}, r_{3,3}\}$ with the following strict and naturally single-peaked preferences.

{\centering
	\renewcommand{\arraystretch}{1.25}
	\begin{tabular}{rl}
		$C_3:$ & $\boldsymbol{\frac 47 \succ \frac 12 \succ \frac 35} \succ \frac 34 \succ \frac 3{10} \succ \frac 29 \succ \frac 18 \succ 1 \succ 0$\\
	\end{tabular}\par}
\medskip

For creating $C_3$, we consider a set of auxiliary agents $A$ containing agents  with single-peaked preferences with peak at $\frac{4}{13}$ and satisfying $\frac 13 \succ \frac 3{11} \succ 0 \succ \frac 12 \succ 1$.

Now, let $i\in [7]$ and consider a set of \emph{blue} agents $C_{3,1} = \{b_{a,i}\colon i \in [7]\}\subseteq A$. By Lemma~\ref{lem:HDG-hom}, we can create the coalition $C_{3,1}$. Now, since we move the coalition ratio towards the peak of the blue agents, we can have the red agents from $C_3$ join one by one to form the coalition $C_{3,2} = C_{3,1} \cup \{r_{3,1}, r_{3,2}, r_{3,3}\}$. For this, note that $\frac 3{10} < \frac 4{13}$. The next goal is to get rid of the agents in $C_{3,1}$. To make this happen, we create auxiliary coalitions such that the agents in $C_{3,1}$ can move there and get into a most preferred coalition. Therefore, we create $7$ identical coalitions as follows. By Lemma~\ref{lem:HDG-hom}, we can create a homogeneous coalition consisting of $8$~blue agents from $A$. Then, we let $4$ red agents from $A$ join one after another. Note that this only consists of IS deviations, even though we cross the peak of these agents, because all agents satisfy $\frac 4{12} = \frac 13 \succ \frac 3{11}$. Also, all agents in the resulting coalition would allow another blue agent to join, because this deviation would lead to reaching the peak of $\frac 4{13}$. Hence, we let the agents from $C_{3,1}$, one after another, deviate to distinct auxiliary coalitions. All of these steps are a strict improvement for the deviators leaving $C_{3,2}$. The first of the deviators leaves a coalition of ratio $\frac 3{10}$ and reaches her peak. Afterwards, the ratio of the abandoned coalitions is at least $\frac 39 = \frac 13 > \frac 4{13}$ and therefore all other deviators improve strictly.

We obtain the coalition $\{r_{3,1}, r_{3,2}, r_{3,3}\}$ and the blue agents from $C_3$ can join one after another to form coalition $C_3$.

\paragraph{Starting cycling}
The final step for this example is to show how to start the cycle constructed in the first HDG. Therefore, we have to add the deviator agents $r_1$, $r_2$, $b_1$, and $b_2$ with the following preferences.

{\centering
	\renewcommand{\arraystretch}{1.25}
	\begin{tabular}{rl}
		$b_1:$ & $\frac{3}{8} \succ \frac{5}{7} \succ \frac{5}{6} \succ \frac{2}{7} \succ 0$ \\
		$b_2:$ & $\frac{5}{7} \succ \frac{4}{7} \succ \frac{1}{2} \succ \frac{5}{6} \succ 0$ \\
		$r_1:$ & $\frac{4}{7} \succ \frac{1}{4} \succ \frac{1}{7} \succ \frac{2}{3} \succ 1$ \\
		$r_2:$ & $\frac{1}{4} \succ \frac{3}{8} \succ \frac{3}{7} \succ \frac{1}{7} \succ 1$ \\
	\end{tabular}\par}
\medskip

Given the constructed coalitions $C_i$, $i\in [4]$, we perform the following IS deviations:

\begin{itemize}
	\item Agent $r_2$ joins coalition $C_1$.
	\item Agent $b_1$ joins coalition $C_2$.
	\item Agent $b_2$ joins coalition $C_3$.
	\item Agent $r_1$ joins coalition $C_3\cup \{b_2\}$.
\end{itemize}

This results in a partition that occurs in the cycle of the first example. Hence, the IS dynamics can cycle as shown before.
\end{enumerate}
\end{proof}

The previous examples do not show the impossibility to reach an IS partition since, e.g., in the first case, the IS partition $\{C_1\cup\{b_1,r_2\},C_2,C_3\cup\{r_1,b_2\},C_4\}$ can be reached via IS deviations from some partitions in the cycle. Thus, starting in these partitions, a path to stability may still exist.
Nevertheless, it may be possible that every sequence of IS deviations cycles, even for strict or naturally single-peaked preferences (with indifference), as the next proposition shows. 
An interesting open question is whether strict and naturally single-peaked preferences allow for the existence of a path to stability. 

\begin{proposition}\label{prop:nopath-HDG}
The dynamics of IS deviations may never reach an IS partition in HDGs, whatever the chosen path of deviations, even 
for (1) strict preferences or (2) naturally single-peaked preferences with indifference.
\end{proposition}

\begin{proof}
Let us consider an HDG with 10 agents: 4 red agents and 6 blue agents.
There are two deviating agents: red agent $r$ and blue agent $b$, and three fixed coalitions $C_1$, $C_2$ and $C_3$ such that:
\begin{itemize}
\item $C_1$ contains 2 red agents;
\item $C_2$ contains 1 red agent and 3 blue agents;
\item $C_3$ contains 2 blue agents.
\end{itemize} 
The relevant part of the preferences of the agents is given below, with on the left the preferences for the case of $(1)$ strict preferences, where $[...]$ denotes an arbitrary order over all possible remaining ratios, and on the right the preferences for the case of $(2)$ naturally single-peaked preferences with indifference, where all agents except $r$ and $b$ are indifferent between all possible coalitions.

\renewcommand{\arraystretch}{1.25}
\begin{minipage}[c]{0.47\columnwidth}
\centering\begin{tabular}{rl}
\multicolumn{2}{c}{$(1)$} \\
$r:$ & $\frac{3}{4} \succ \frac{2}{5} \succ \frac{1}{4} \succ \frac{1}{3} \succ 1$ \\
$b:$ & $\frac{1}{4} \succ \frac{1}{5} \succ \frac{3}{4} \succ \frac{2}{3} \succ 0$ \\
$C_1:$ & $\frac{3}{4} \succ \frac{2}{3} \succ 1 \succ [...]$ \\
$C_2:$ & $\frac{2}{5} \succ \frac{1}{5} \succ \frac{1}{4} \succ [...]$ \\
$C_3:$ & $\frac{1}{4} \succ \frac{1}{3} \succ 0 \succ [...]$ \\
\end{tabular}
\end{minipage}
\begin{minipage}[c]{0.47\columnwidth}
\centering\begin{tabular}{rl}
\multicolumn{2}{c}{$(2)$} \\
$r:$ & $\frac{3}{4} \succ \frac{2}{5} \succ \frac{1}{4} \sim \frac{1}{3} \succ 1$ \\
$b:$ & $\frac{1}{4} \succ \frac{1}{5} \succ \frac{3}{4} \sim \frac{2}{3} \succ 0$ \\
$i\in N\setminus\{r,b\}:$ & $C\sim C'$, \quad $\forall C,~C'\in \N_i$ \\
\\
\\
\end{tabular}
\end{minipage}
\smallbreak
Consider the sequence of IS deviations in Figure~\ref{fig:HDG-cycle-eternal} that describes a cycle of the dynamics. 
The two deviating agents of the cycle $r$ and $b$ are marked in bold and the specific deviating agent between two states is indicated next to the arrows.

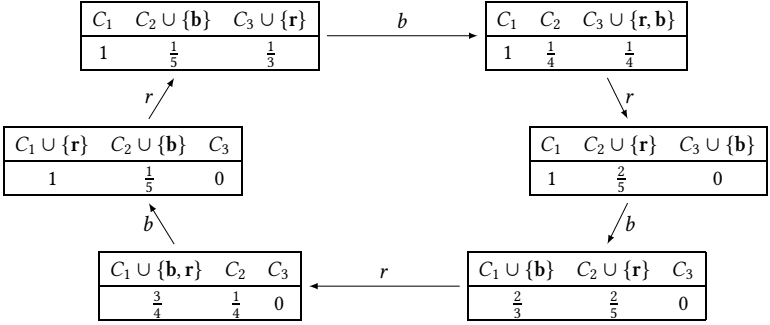
\begin{figure}
{\centering
\resizebox{0.75\columnwidth}{!}{
\begin{tikzpicture}
\node (1) at (0,0) {\begin{tabular}{|ccc|}\hline $C_1\cup\{\mathbf{r}\}$ & $C_2\cup\{\mathbf{b}\}$ & $C_3$ \\ \hline 
$1$ & $\frac{1}{5}$ & $0$ \\ \hline\end{tabular}};
\node (2) at ($(1)+(1.25,2)$) {\begin{tabular}{|ccc|}\hline $C_1$ & $C_2\cup\{\mathbf{b}\}$ & $C_3\cup\{\mathbf{r}\}$ \\ \hline 
$1$ & $\frac{1}{5}$ & $\frac{1}{3}$ \\ \hline\end{tabular}};
\node (3) at ($(2)+(6.25,0)$) {\begin{tabular}{|ccc|}\hline $C_1$ & $C_2$ & $C_3\cup\{\mathbf{r},\mathbf{b}\}$ \\ \hline 
$1$ & $\frac{1}{4}$ & $\frac{1}{4}$ \\ \hline\end{tabular}};
\node (4) at ($(1)+(8.5,0)$) {\begin{tabular}{|ccc|}\hline $C_1$ & $C_2\cup\{\mathbf{r}\}$ & $C_3\cup\{\mathbf{b}\}$ \\ \hline 
$1$ & $\frac{2}{5}$ & $0$ \\ \hline\end{tabular}};
\node (6) at ($(1)+(1.25,-2)$) {\begin{tabular}{|ccc|}\hline $C_1\cup\{\mathbf{b},\mathbf{r}\}$ & $C_2$ & $C_3$ \\ \hline 
$\frac{3}{4}$ & $\frac{1}{4}$ & $0$ \\ \hline\end{tabular}};
\node (5) at ($(6)+(6.25,0)$) {\begin{tabular}{|ccc|}\hline $C_1\cup\{\mathbf{b}\}$ & $C_2\cup\{\mathbf{r}\}$ & $C_3$ \\ \hline 
$\frac{2}{3}$ & $\frac{2}{5}$ & $0$ \\ \hline\end{tabular}};
\draw[arrow] (1) --node[midway,left]{$r$} (2); \draw[arrow] (2) --node[midway,above]{$b$}  (3); \draw[arrow] (3) --node[midway,right]{$r$}  (4); \draw[arrow] (4) --node[midway,right]{$b$}  (5); \draw[arrow] (5) --node[midway,above]{$r$}  (6); \draw[arrow] (6) --node[midway,left]{$b$}  (1);
\end{tikzpicture}}\par}
	\caption{Impossibility of convergence of IS dynamics in the proof of Proposition~\ref{prop:nopath-HDG}.}\label{fig:HDG-cycle-eternal}	
\end{figure}

Note that at each state, the deviation performed by agent $r$ or $b$ is the only possible one that they can do. 
Moreover, by construction of the preferences of the other agents, 
none of them has incentive to deviate at any state.
Therefore, the cycle is the only possible sequence of IS deviations, and the cycle cannot be avoided in this instance.
\end{proof}

It remains open whether there always exists an path to stability when starting from the \singleton.
However, we show now that convergence is guaranteed by combining all restrictions of Theorem~\ref{thm:cycleHDG}. Interestingly, part of the proof is a reduction to the positive result about AHGs, and reveals a close relationship of the two classes of hedonic games if HDGs are sufficiently restricted. However, this reduction requires careful preprocessing of the initial HDG. The first key insight is to show that, for every coalition occurring during the dynamics, there exists a color such that at most one agent of this color is part of the coalition. Consequently, given an instance with $b$ blue and $r$ red agents, the only important ratios (apart from $0$ and $1$) are $\frac{k}{k+1}$ and $\frac 1{k'+1}$, where $1\le k\le r$ and $1\le k' \le b$. The next step is to show how to transition to an HDG where every agent will only end up in coalitions of one of these ratio types. This transition is essentially performed by omitting certain steps in the dynamics. From there, we can observe the structure of an AHG by identifying the ratios $\frac{k}{k+1}$ and $\frac 1{k+1}$ with a coalition size of $k + 1$. This correspondence is reasonable. For instance, the former ratio corresponds to a coalition with one blue and $k$ red agents, i.e., a total number of $k + 1$ agents. We can apply Theorem~\ref{thm:convAHG} to bound the length of the transformed HDG. Interestingly, the identification with an AHG requires some auxiliary agents, and the transformed dynamics is not starting from the \singleton anymore. In this respect, we even need the full power of Theorem~\ref{thm:convAHG}. 

\begin{theorem}\label{thm:converg-HDG-singletonstrictSP}
The dynamics of IS deviations satisfying solitary homogeneity always converges in $\mathcal O(n^5)$ steps when starting from the \singleton in an HDG where agents have strict and naturally single-peaked preferences.
\end{theorem}

\begin{proof}
	\renewcommand\qedsymbol{$\vartriangleleft$}
	Consider an HDG with agent set $N = R\cup B$, where agents have strict and naturally single-peaked preferences. Let $(\partition_k)_{k = 0}^K$ be a sequence of partitions of an execution of the dynamics of IS deviations satisfying solitary homogeneity, where $\partition_0$ is the \singleton and, for every $1\le k \le K$, $\partition_k$ evolves from $\partition_{k-1}$ by an IS deviation of agent $d_k$.

The first step of the proof is to show the specific structure of the attained coalitions.

\begin{claim}\label{lem:converg-HDG-structure}
	For every $k\ge 0$, it holds that every coalition in $\partition_k$ is of the form $\{r_1\}$, $\{b_1\}$, $\{r_1,b_1,\dots,b_m\}$, or $\{b_1,r_1,\dots,r_{m'}\}$, where $1 \le m\leq |B|$ and $1 \le m'\leq |R|$ and $r_i\in R$ and $b_j\in B$ for every $i\in [m']$, $j\in [m]$. Moreover, the following statements hold:
		\begin{enumerate}
		\item If $\{r_1,b_1,\dots,b_m\}\in \partition_k$ for $m\ge 2$, then
		$\frac{1}{m+1} \succ_{r_1} \frac{1}{m} \succ_{r_1} \dots \succ_{r_1} \frac{1}{2} \succ_{r_1} 1$.
		\item If $\{b_1,r_1,\dots,r_{m'}\}\in \partition_k$ for $m'\ge 2$, then 
		$\frac{m'}{m'+1} \succ_{b_1} \frac{m' - 1}{m'} \succ_{b_1} \dots \succ_{b_1} \frac{1}{2} \succ_{b_1} 0$.
	\end{enumerate}
\end{claim}

\begin{proof}
	We will show by induction over $k$ for $0\le k \le K$ that every coalition in $\partition_k$ is of the form $\{r_1\}$, $\{b_1\}$, $\{r_1,b_1,\dots,b_m\}$, or $\{b_1,r_1,\dots,r_{m'}\}$, where $b_i\in B$ and $r_j\in R$ for every $i\in [m]$, $j\in [m']$, $1 \le m\leq |B|$, and $1 \le m'\leq |R|$. Simultaneously to this main claim, we will prove the additional statements as auxiliary claims. 
	
	Clearly, the \singleton satisfies the main and auxiliary claims. Now, assume that the assertion is true for some fixed $0 \le k < K$.
	Assume without loss of generality that $d_{k+1}$ is a red agent (the case for a blue agent is symmetric and uses the second auxiliary claim where we use the first auxiliary claim). We have to consider the two coalitions affected by $d_{k+1}$ to show that $\partition_{k+1}$ satisfies the claims.
	
	First, assume for contradiction that the coalition~$\partition_{k+1}(d_{k+1})$ breaks the main claim. Then, $\partition_{k+1}(d_{k+1}) \setminus \{d_{k+1}\}$ is of the form $\{r_1\}$ or $\{r_1,b_1,\dots,b_m\}$ with $2 \le m \le |B|$. The former case is excluded as the deviation satisfies solitary homogeneity. In the latter case, $\frac 2{2+m} > \frac 1 m > \frac 1{m+1}$, and we know by the first auxiliary claim in step $k$ that $\frac{1}{m+1} \succ_{r_1} \frac{1}{m}$. Hence, single-peakedness implies $\frac 1m \succ_{r_1} \frac 2{2+m}$, and, by transitivity of the preferences, we obtain that $\frac{1}{m+1} \succ_{r_1} \frac 2{2+m}$. This contradicts the fact that the deviation by $d_{k+1}$ was approved by agent~$r_1$. Hence, $\partition_{k+1}(d_{k+1})$ satisfies the main claim and must be of the form $\{b_1,r_1,\dots,r_{m'}\}$, where $m'\ge 1$. 
	
	We proceed with the auxiliary claims for this coalition. As the coalition contains only one blue agent, the first auxiliary claim is vacant. Further, since $b_1$ gave her consent to letting $d_{k+1}$ join, it satisfies the second auxiliary claim (extending the second auxiliary claim for $b_1$ at step $k$) if $m'\ge 3$. If $m' = 2$, then the consent of $b_1$ implies $\frac 23 \succ_{b_1} \frac 12$, and single-peakedness implies $\frac 23 \succ_{b_1} \frac 12 \succ_{b_1} 0$. If $m' = 1$, this claim is also vacant.
	
	Second, assume for contradiction that the coalition abandoned by agent~$d_{k+1}$ violates the main claim. Then, $\partition_k(d_{k+1})$ was of the form $\{r,b_1,\dots,b_m\}$ with $m\ge 2$ and $r = d_{k+1}$. We already know that $\partition_{k+1}(d_{k+1})$ is of the desired form. It cannot be the coalition $\{r\}$, because of the first auxiliary claim for $r$. Also, it cannot be of the form $\{r,\hat b_1,\dots,\hat b_{\hat m}\}$ with $\hat m\ge 2$, because then $\partition_k(\hat b_1)$ violates the main claim in step $k$. Hence, we know that $\partition_{k+1}(d_{k+1})$ is of the form $\{b,r_1,\dots,r_{m'}\}$ with $m' \ge 1$. Using $m\ge 2$, we have that $\frac 1{m+1} < \frac{m'}{m'+1} < 1$ and, since the deviation was performed by $d_{k+1}$, also $\frac{m'}{m'+1} \succ_{d_{k+1}} \frac 1{m+1}$. Hence, single-peakedness implies $1  \succ_{d_{k+1}} \frac{m'}{m'+1}$, and therefore, by transitivity, $1  \succ_{d_{k+1}} \frac 1{m+1}$. However, this contradicts the first auxiliary claim for agent $d_{k+1} = r$ in $\partition_k$.
	
	It remains to prove the auxiliary claims for the abandoned coalition. The first auxiliary claim is vacant. The second auxiliary claim follows directly by induction, whenever it is not vacant.
\end{proof}

	In the sequel, we use the notation $f_i(\partition) = \frac{|R\cap \partition(i)|}{|\partition(i)|}$, which specifies the fraction of red agents in the coalition of agent~$i$ with respect to partition~$\partition$. 
	Also, given an agent~$i$, denote her peak by $p_i$. We distinguish agents according to their peaks. To this end, define the agent sets
	
	\begin{itemize}
		\item $R_S = \{r\in R\colon 0 < p_r < 1/2\}$,
		\item $R_L = \{r\in R\colon 1/2 \le p_r \le 1\}$,
		\item $B_S = \{b\in B\colon 0 \le p_b \le 1/2\}$, and
		\item $B_L = \{b\in B\colon 1/2 < p_b < 1\}$.
	\end{itemize}
	
	The subscripts indicate whether the peak is large ($L$) or small ($S$). We would like to analyze a dynamics where $f_i(\partition)$ is always close to the peak of an agent. This is achieved by agents in $R_L$ and $B_S$.

\begin{claim}\label{lem:HDG-boring-agents}
	Let $k\ge 0$. Then, the following statements hold:
	
	\begin{enumerate}
		\item If $r\in R_L$, then $f_r(\partition_k)\ge \frac 12$.
		\item If $b\in B_S$, then $f_b(\partition_k)\le \frac 12$.
	\end{enumerate}
\end{claim}	

\begin{proof}
	We show the statement by induction over $k$ for $0\le k \le K$. Clearly, the statement is true for $k = 0$, because $\partition_0$ is the \singleton. Now, assume that the assertion is true for some fixed $0 \le k < K$. We assume without loss of generality that agent~$d_{k + 1}$ is red (the case of a blue agent follows from a symmetric argument). Clearly, all agents not affected by the deviation maintain the two invariants claimed in the lemma. Therefore, we have to consider the abandoned and joined coalitions. 
	
	By Claim~\ref{lem:converg-HDG-structure}, $\partition_k(d_{k+1})$ is of the form $\{d_{k+1}\}$ or $\{b_1, r_1, \dots, r_m, d_{k+1}\}$ for some $m\ge 0$, where $b_1 \in B$ and $r_1,\dots, r_m\in R$. The former case is irrelevant because then the abandoned coalition does not exist anymore. In the latter case, $f_{b_1}(\partition_{k+1}) < f_{b_1}(\partition_k)$, and the second invariant follows by induction if $b\in B_S$. 
	Furthermore, if $m\ge 1$, then $f_{b_1}(\partition_{k + 1}) \ge \frac 12$, and the first invariant is true for all red agents the abandoned coalition (in particular if they are in $R_L$).
	
	Applying Claim~\ref{lem:converg-HDG-structure} again, $\partition_{k+1}(d_{k+1})$ is also of the form $\{d_{k+1}\}$ or $\{b_1, r_1, \dots, r_{m'}, d_{k+1}\}$ for some $m'\ge 0$, where $b_1 \in B$ and $r_1,\dots, r_{m'}\in R$. Hence, the first invariant is satisfied for all red agents and for $b_1$ if $m' = 0$. It remains the case $m' \ge 1$. Then, since the deviation was approved by agent~$b_1$, it holds $\frac{m' + 1}{m'+2} \succ_{b_1} \frac{m'}{m'+1}$. Then, single-peakedness implies that $p_{b_1} > \frac{m'}{m'+1} \ge \frac 12$ and therefore $b_1\notin B_S$. Hence, the second invariant is vacant in this case.
	
	Altogether, we have shown that both invariants are satisfied for partition~$\partition_{k+1}$, which completes the induction step.
\end{proof}

\begin{figure}
	\centering
	\begin{tikzpicture}
		\node (d1) at (0,0) {$(\partition_k)_{k = 1}^K$};
		\node (d2) at (3.5,0) {$(\sigma_k)_{k = 1}^K$};
		\node (d3) at (7,0) {$(\tau_l)_{l = 1}^L$};
		\node (d4) at (10.5,0) {$(\rho_p)_{p = 1}^P$};
		
		\draw[->] (d1) edge node[above] {\footnotesize no undesirable}node[below] {\footnotesize coalitions} (d2);
		\draw[->] (d2) edge node[above] {\footnotesize include merging}node[below] {\footnotesize deviations} (d3);
		\draw[->] (d3) edge node[above] {\footnotesize omit duplicate}node[below] {\footnotesize partitions} (d4);
	\end{tikzpicture}
	\caption{Transformation of the dynamics in the proof of Theorem~\ref{thm:converg-HDG-singletonstrictSP}. In the first modification, we keep red agents with a small peak and blue agents with a large peak in a singleton coalition. This yields a sequence of partitions $(\sigma_k)_{k = 1}^K$ where some transitions are no valid IS deviation. We include some merging deviations to resolve this. Finally, we remove duplicate partitions to end up with a valid IS dynamics $(\rho_p)_{p = 1}^P$.\label{fig:dyntrans}}
\end{figure}

	A similar statement is not true for agents in $R_S$ or $B_L$. Therefore, the next step will \emph{modify} the dynamics such that agents are only contained in coalitions close to their peaks, unless they are in a singleton coalition. This is a sophisticated and tedious construction which will be performed in several steps. An outline is depicted in Figure~\ref{fig:dyntrans}. First, we prevent undesirable deviations by red agents with a small peak and blue agents with a large peak by keeping these agents in singleton coalitions. However, this yields a modified dynamics $(\sigma_k)_{k = 1}^K$ where some steps do not correspond to IS deviations. We therefore have to insert and omit certain steps. 
	Eventually, we will end up at a new dynamics  $(\rho_p)_{p = 1}^P$. This dynamics is easier to analyze because it corresponds to an AHG. Moreover, the convergence behavior of the original dynamics only depends on the convergence behavior of the new dynamics because we only omit a linear number of steps (with respect to $n$). 
	Hence, we can complete the proof by showing how to bound $K$, i.e., the length of the original dynamics, with respect to $P$, i.e., the length of the simpler dynamics.
	
	We start with the first transformation. Given a partition~$\partition$, define the subset of agents $F_{\partition} \subseteq N$ as $F_{\partition} = \{r\in R_S\colon \frac 12 < f_r(\partition) < 1\} \cup \{b\in B_L\colon \frac 12 > f_b(\partition) > 0\}$, i.e., the set of agents whose ratio is \emph{far} from their peak. 
	By definition of $F_\partition$ and Claim~\ref{lem:converg-HDG-structure}, a red agent $r\in R\cap F_\partition$ is in a coalition of the form $\{b_1,r_1,\dots,r_m\}$ in partition~$\partition$ and, symmetrically, a blue agent $b\in B\cap F_\partition$ is in a coalition of the form $\{r_1,b_1,\dots,b_m\}$ in partition~$\partition$. Moreover, by Claim~\ref{lem:converg-HDG-structure} and strict single-peakedness of the preferences, an agent in $F_\partition$ is the last agent who entered her coalition in $\partition$.
	
	Now, consider a modified dynamics $(\sigma_k)_{k = 0}^K$, where, for every $0\le k \le K$, $\sigma_k = \bigcup_{C\in \partition_k} \{C\setminus F_{\partition_k}\}\cup \bigcup_{i\in F_{\partition_k}}\{\{i\}\}$. This modification has essentially the following effects: We omit deviations of agents where they land in the set $F_{\partition}$ while keeping them in a singleton coalition. Sometimes, it can happen that the deviator satisfies $d_k\in F_{\partition_{k - 1}}\setminus F_{\partition_k}$. In this case, the modified dynamics \emph{sees} the deviator join her new coalition from a singleton coalition. The second effect that can happen is the case where a non-deviator is in $F_{\partition_{k-1}}\setminus F_{\partition_k}$. This happens exactly if an agent in $F_{\partition_{k-1}}$ is abandoned, and thereby left in a coalition of size~$2$, which consists of one agent of each type.
	Hence, we insert a suitable deviation to obtain a valid modified dynamics of IS deviations. This requires two lemmas. The first lemma gives more structural insight and establishes that every coalition can contain at most one agent from $F_{\partition_k}$.

\begin{claim}\label{lem:HDG-unique-far}
	Let $0\le k \le K$ and $C\subseteq \partition_k$. Then, $|C\cap F_{\partition_k}|\le 1$. Moreover, consider $m>0$ and agents $r_1,\dots, r_m\in R$ and $b_1,\dots, b_m\in B$. Then, the following statements hold:
	\begin{enumerate}
		\item If $\{b_1, r_1, \dots, r_m\}\in \partition_k$, then, for all $1\le i\le m$ with $r_i\notin F_{\partition_k}$, $f_{r_i}(\partition_k) \succeq_{r_i} f_{r_i}(\sigma_k)$.
		\item If $\{r_1, b_1, \dots, b_m\}\in \partition_k$, then, for all $1\le i\le m$ with $b_i\notin F_{\partition_k}$, $f_{b_i}(\partition_k) \succeq_{b_i} f_{b_i}(\sigma_k)$.
	\end{enumerate}
\end{claim}

\begin{proof}
	We show all statements simultaneously by induction over $k$ for $0\le k \le K$. Clearly, all statements are true for $k = 0$. Now, assume that the statements are true for some fixed $0 \le k < K$. Clearly, all coalitions in $\partition_{k+1}$ except possibly $\partition_k(d_{k+1})\setminus\{d_{k+1}\}$ and $\partition_{k + 1}(d_{k+1})$ satisfy the claim. Assume without loss of generality that $d_{k+1}$ is a red agent (the case of a blue agent is symmetric). 
	
	We start with the abandoned coalition. By Claim~\ref{lem:converg-HDG-structure}, $\partition_k(d_{k+1})\setminus\{d_{k+1}\}$ is the empty set (if $d_{k+1}$ was in a singleton coalition) or of the type $\{b_1,r_1,\dots, r_m\}$ where $b_1\in B$ and $r_1,\dots, r_m\in R$ for some $m\ge 0$. 
	In the former case, all claims are vacant, so assume the latter case. If $m = 0$, the abandoned coalition is a singleton coalition, and the assertions are true (the additional statements are then vacant). If $m = 1$, then $f_{b_1}(\partition_{k + 1}) = \frac 12$, and therefore $\partition_{k + 1}(b_1)\cap F_{\partition_{k+1}} = \emptyset$. 
	In particular, $\partition_{k + 1}(b_1) = \sigma_{k+1}(b_1)$, and all claims are true. If $m\ge 2$, it follows from $f_{b_1}(\partition_{k + 1})\ge \frac 23$ that $b_1\notin F_{\partition_{k+1}}$. Moreover, $\{r_1,\dots, r_m\} \cap F_{\partition_{k+1}}\subseteq \{r_1,\dots, r_m\} \cap F_{\partition_k}$. Therefore, $|\partition_{k + 1}(b_1)\cap F_{\partition_{k+1}}|\le 1$ follows from induction for step $k$, which implies that $|\{r_1,\dots, r_m\} \cap F_{\partition_k}| \le |\partition_k(b_1)\cap F_{\partition_k}|\le 1$. Additionally, the additional statement is trivially true unless $\{r_1,\dots, r_m\} \cap F_{\partition_{k+1}}\neq \emptyset$, say $r_1\in F_{\partition_{k+1}}$. Then, $r_1\in F_{\partition_k}$ and we have that $f_{b_1}(\sigma_{k + 1}) < f_{b_1}(\partition_{k + 1}) = f_{b_1}(\sigma_k) < f_{b_1}(\partition_k)$. Let $2\le i\le m$. Then, induction implies 
	$f_{r_i}(\partition_k) \succ_{r_i} f_{r_i}(\sigma_k)$. Hence, single-peakedness implies $f_{r_i}(\partition_{k+1}) \succ_{r_i} f_{r_i}(\sigma_{k+1})$.
	
	The proof for the joined coalition is similar. Using Claim~\ref{lem:converg-HDG-structure} again, we know that $\partition_{k+1}(d_{k+1})$ is a singleton coalition or of the type $\{b_1,r_1,\dots, r_m, d_{k+1}\}$ where $b_1\in B$ and $r_1,\dots, r_m\in R$ for some $m\ge 0$. A singleton coalition fulfills the claim. Therefore, assume the latter case. Since $f_{b_1}(\partition_{k + 1})\ge \frac 12$, it holds that $b_1\notin F_{\partition_{k+1}}$. Moreover, if $m\ge 1$ and $1\le i\le m$, then $\frac {m+1}{m+2}\succ_{r_i}\frac m{m+1}$ and single-peakedness implies that $p_{r_i} > \frac m{m+1}\ge \frac 12$. Hence, $r_i\notin R_S$, and therefore $r_i\notin F_{\partition_{k+1}}$. Together, $\partition_{k+1}(d_{k+1})\cap F_{\partition_{k + 1}}\subseteq \{d_{k+1}\}$, and the first statement is true. The additional statement is clear if $d_{k+1}\notin F_{\partition_{k + 1}}$, in which case $\partition_{k+1}(b_1) = \sigma_{k+1}(b_1)$. If $d_{k+1}\in F_{\partition_{k + 1}}$, it follows for the other red agents, because they approve that $d_{k+1}$ joins.
\end{proof}

We are ready to show how to obtain the valid dynamics.

\begin{claim}\label{lem:shortcutting}
	Let $1\le k \le K$. If $\sigma_k \neq \sigma_{k-1}$, then $\sigma_k$ evolves from $\sigma_{k-1}$ by performing at most $2$ IS deviations. If two deviations have to be performed, then the intermediate partition evolves from $\sigma_{k-1}$ by merging two agents from singleton coalitions.
\end{claim}
	
\begin{proof}
	Let $1\le k \le K$ with $\sigma_k \neq \sigma_{k-1}$. The only agents that matter to us are in $\partition_k(d_k)\cup \partition_{k - 1}(d_k)$. Other agents did not change their coalition in the original dynamics, and therefore, their membership in $F_{\partition_k}$ is also not affected. Without loss of generality, we assume that $d_k$ is a red agent (the case of a blue agent is again symmetric). 
	
	First, we show that $(\partition_k(d_k)\setminus \{d_k\}) \cap F_{\partition_{k-1}} = \emptyset$ and $\partition_k(d_k)\cap F_{\partition_k} = \emptyset$. Assume for contradiction that this is not the case. By Claim~\ref{lem:converg-HDG-structure}, this can only be the case if $\partition_k(d_k)$ is of the form $\{b_1,r_1,\dots, r_m, d_k\}$ where $b_1\in B$ and $r_1,\dots, r_m\in R$ for some $m\ge 1$. As $f_{b_1}(\partition_k)\ge \frac 23$ and $f_{b_1}(\partition_{k - 1})\ge \frac 12$, $b_1\notin F_{\partition_k}$ and $b_1\notin F_{\partition_{k-1}}$, respectively. 
	For $1\le i \le m$, it holds that $\frac{m+1}{m+2}\succ_{r_i} \frac m{m + 1}$, and therefore, by single-peakedness, $p_i > \frac m{m + 1} \ge \frac 12$. Hence, $r_i\in R_L$, and therefore $r_i\notin  F_{\partition_k}$ and $r_i\notin F_{\partition_{k-1}}$. This shows already that $(\partition_k(d_k)\setminus \{d_k\}) \cap F_{\partition_{k-1}} = \emptyset$.
	
	Finally, it remains to exclude that $d_k \in F_{\partition_k}$. Assume for contradiction that $d_k \in F_{\partition_k}$. Then, our considerations about $\partition_k(d_k)$ imply that $\sigma_k(b_1) = \sigma_{k-1}(b_1)$. Moreover, by Claim~\ref{lem:converg-HDG-structure}, $f_{d_k}(\partition_{k-1}) \ge \frac 12$. Since $d_k \in F_{\partition_k}$, it holds that $p_{d_k} < \frac 12$. As we already know that $f_{d_k}(\partition_k) > \frac 12$, single-peakedness and the fact that $d_k$ has improved her ratio imply that $f_{d_k}(\partition_{k-1}) \ge f_{d_k}(\partition_k) > \frac 12$. Therefore, $d_k \in F_{\partition_{k-1}}$ and Claim~\ref{lem:HDG-unique-far} implies that $(\partition_{k-1}(d_k)\setminus\{d_k\}) \cap F_{\partition_{k-1}} = \emptyset$. Additionally, $f_{d_k}(\partition_{k-1}) > \frac 12$ also implies that $(\partition_{k-1}(d_k)\setminus\{d_k\}) \cap F_{\partition_k} = \emptyset$. Hence, the coalition abandoned by $d_k$ has not changed from $\sigma_{k-1}$ to $\sigma_k$. 
	Together, this contradicts that $\sigma_k \neq \sigma_{k-1}$. Hence, our assumption that $d_k \in F_{\partition_k}$ was wrong, and therefore $\partition_k(d_k)\cap F_{\partition_k} = \emptyset$. In particular, we have shown so far that 
	
	\begin{equation}\label{eq:shortcutlemma-welcoming}
		\sigma_{k-1}(b_1) = \partition_{k-1}(b_1)\textnormal{ and }\sigma_k(b_1) = \partition_k(b_1) = \partition_{k-1}\cup\{d_k\}\text.
	\end{equation}

	Next, we consider the abandoned coalition. By Claim~\ref{lem:converg-HDG-structure}, the coalition~$\partition_{k-1}(d_k)$ is of the form $\{d_k\}$ or $\{b'_1,r'_1,\dots, r'_{m'}, d_k\}$ where $b'_1\in B$ and $r'_1,\dots, r'_{m'}\in R$ for some $m'\ge 0$. In the first case, we know that $\{d_k\}\in \sigma_{k-1}$, and, together with Equation~\ref{eq:shortcutlemma-welcoming}, $\sigma_{k-1} = \partition_{k-1}$ and $\sigma_k = \partition_k$. Therefore, $\sigma_k$ evolves from $\sigma_{k-1}$ by an IS deviation of agent~$d_k$.

	Next, we consider the case that $\partition_{k-1}(d_k)$ is of the form $\{b'_1,r'_1,\dots, r'_{m'}, d_k\}$. Note that $f_{b'_1}(\partition_{k-1})\ge \frac 12$ and $f_{b'_1}(\partition_k)\ge \frac 12$ or $f_{b'_1}(\partition_k) = 0$ and therefore $b'_1\notin F_{\partition_{k-1}}$ and $b'_1\notin F_{\partition_k}$. Moreover, it holds that $\{r'_1,\dots, r'_{m'}\} \cap F_{\partition_k}\subseteq \{r'_1,\dots, r'_{m'}\} \cap F_{\partition_{k - 1}}$. We are ready to consider the final cases. 
	
	First assume that $d_k\in F_{\partition_{k - 1}}$. Then, Claim~\ref{lem:HDG-unique-far} and the considerations in the previous paragraph imply that $\{b'_1,r'_1,\dots, r'_{m'}\} \cap F_{\partition_k} = \emptyset$ and $\{b'_1,r'_1,\dots, r'_{m'}\} \cap F_{\partition_{k - 1}} = \emptyset$. Hence, $\{b'_1,r'_1,\dots, r'_{m'}\}\in \sigma_k$ and  $\{b'_1,r'_1,\dots, r'_{m'}\}\in \sigma_{k-1}$. This, together with Equation~\ref{eq:shortcutlemma-welcoming} and the definition of $\sigma_{k-1}$ implies that $\sigma_k$ evolves from $\sigma_{k-1}$ by a unilateral deviation of agent $d_k$ from a singleton coalition to coalition $\sigma_k(d_k)$. Since $\sigma_k \neq \sigma_{k-1}$, we know that $\sigma_k(d_k) \neq \{d_k\}$. Hence, $\frac 12 \le f_{d_k}(\sigma_k) < 1$. This, together with $d_k\in R_S$ implies that $p_{d_k} \le f_{d_k}(\sigma_k) < 1$. Hence, single-peakedness implies that the deviation was a Nash deviation. By Equation~\ref{eq:shortcutlemma-welcoming}, the deviation was also approved by all agents in the joined coalition. Hence, $\sigma_k$ evolves from $\sigma_{k-1}$ through an IS deviation of agent~$d_k$.
	
	It remains the case that $d_k\notin F_{\partition_{k - 1}}$. If $\{r'_1,\dots, r'_{m'}\} \cap F_{\partition_{k - 1}} = \emptyset$, then $\sigma_{k-1} = \partition_{k-1}$ and, together with Equation~\ref{eq:shortcutlemma-welcoming}, $\sigma_k$ evolves from $\sigma_{k-1}$ by an IS deviation of agent~$d_k$.
	
	Assume therefore that there exists $1\le i\le m'$ with $r'_i\in F_{\partition_{k - 1}}$. If $m'\ge 2$, then $\{b'_1,r'_1,\dots, r'_{m'}\} \cap F_{\partition_{k - 1}} = \{b'_1,r'_1,\dots, r'_{m'}\} \cap F_{\partition_k}$. In this case, $\sigma_k$ evolves from $\sigma_{k-1}$ through a unilateral deviation of $d_k$. 
	Since $d_k\notin F_{\partition_{k - 1}}$, the first additional statement of Claim~\ref{lem:HDG-unique-far} implies that $f_{d_k}(\partition_{k-1}) \succeq_{d_k} f_{d_k}(\sigma_{k-1})$. Therefore, $d_k$ performs a Nash deviation because she performed a Nash deviation from $\partition_{k-1}$ to $\partition_k$. The consent of the joined coalition follows again from Equation~\ref{eq:shortcutlemma-welcoming}.
	
	Finally, assume that $\{r'_1,\dots, r'_{m'}\} \cap F_{\partition_{k - 1}} \neq \emptyset$ and $m' = 1$. Then, $\partition_{k-1}(b'_1) = \{b'_1,r'_1,d_k\}$ and $r'_1 \in  F_{\partition_{k - 1}}$. Hence, $\sigma_k$ evolves from $\sigma_{k-1}$ by transforming $\{r'_1\}$, $\{b'_1,d_k\}$, and $\partition_k(d_k)\setminus\{d_k\}$ into $\{b'_1,r'_1\}$ and $\partition_k(d_k)$. These changes can be achieved by two unilateral deviations. First, $d_k$ joins $\partition_k(d_k)\setminus\{d_k\}$ and then $r'_1$ joins $b'_1$. The first deviation is an IS deviation as in the previous case. The second deviation is also an IS deviation. The approval of $b'_1$ follows from the second auxiliary statement in Claim~\ref{lem:converg-HDG-structure} applied to $\partition_{k-1}(b'_1) = \{b'_1,r'_1,d_k\}$. Also, the deviation is improving for $r'_1$, because $r'_1\in F_{\partition_{k - 1}}$. Therefore $r'_1\in R_S$. Hence, $1 > \frac 12 \ge p_{r'_1}$, and therefore $\frac 12 \succ_{r'_1} 1$.
\end{proof}

Using the insights gained in the previous claim, we can define a valid modified dynamics based on $(\sigma_k)_{k=1}^K$. First, we insert the partitions identified in the proof of Claim~\ref{lem:shortcutting} where two agents are merged. This yields a dynamics $(\tau_l)_{l=1}^L$ such that $\tau_l$ evolves from $\tau_{l-1}$ through a deviation of agent $\hat d_l$ whenever $\tau_l \neq \tau_{l-1}$. Then, we remove all steps where $\tau_l = \tau_{l-1}$ to obtain a dynamics $(\rho_p)_{p=1}^P$, which is a dynamics where every step corresponds to an IS deviation.
Define the index set $I = \{1\le k\le K\colon \sigma_k = \sigma_{k-1}\}$, that is, the set of steps where the modified dynamics remains unchanged. Then, $K \le L = |I| + P$.

Hence, we would like to obtain bounds on each of $|I|$ and $P$.
The next claim allows us to bound $|I|$ by replacing it with an appropriate bound with respect to $P$. The key insight for proving the next claim follows from the observation that essentially every $n$-th deviation of an agent has to correspond to a deviation of the modified dynamics.

\begin{claim}\label{lem:HDG-omittance}
	It holds that $L \le n^2 + nP$.
\end{claim}

\begin{proof}
	We first show that if $\sigma_k = \sigma_{k-1}$, then $d_k\in F_{\partition_k}$. We prove this fact by contraposition. Assume that $d_k\notin F_{\partition_k}$. Note that, by the definition of $F_{\partition}$, it holds for every partition $\partition$ and every coalition $C\in \partition$ with $C\cap F_{\partition}\neq \emptyset$ that $|C|\ge 3$. Now, if $\sigma_k(d_k) \neq \{d_k\}$, then there exists an agent~$x\in \sigma_k(d_k)\setminus\{d_k\}$. Since $d_k$ was the deviator, it holds that $\sigma_k(x)\neq \sigma_{k-1}(x)$ and therefore $\sigma_k \neq \sigma_{k-1}$. It remains the case that $\sigma_k(d_k) = \{d_k\}$. Then, as $d_k\notin F_{\partition_k}$, $\partition_k(d_k) = \{d_k\}$. This implies that $d_k\notin F_{\partition_{k - 1}}$. To see this, we assume without loss of generality that $d_k\in R$. Indeed, if $d_k\in F_{\partition_{k - 1}}$, then $d_k\in R_S$ and $\frac 12 \le f_{d_k}(\partition_{k - 1})$. Then, single-peakedness implies that $f_{d_k}(\partition_{k - 1})\succ_{d_k} 1$, contradicting that $d_k$ performed an IS deviation to form a singleton coalition. Hence, $d_k\notin F_{\partition_{k - 1}}$. As in the first case, we find an agent~$x\in \sigma_{k-1}(d_k)\setminus\{d_k\}$, for which it holds that  $\sigma_k(x)\neq \sigma_{k-1}(x)$ and therefore $\sigma_k \neq \sigma_{k-1}$.
	
	The key insight for this claim is that every agent can only perform few successive deviations corresponding to steps in $I$. Indeed, the first part of the proof implies that $d_k\in F_{\partition_k}$ whenever $k\in I$. 
	Consider an arbitrary agent $r\in R_S$.
	We define a potential function 
	
	$$\lambda_r(\partition) = \begin{cases}
		|R| + 1 & f_r(\partition) \le \frac 12 \textnormal{ or } f_r(\partition) = 1\\
		m & f_r(\partition) = \frac m{m+1} \textnormal{ for } 2\le m \le |R|
	\end{cases}\text.$$
	
	Note that $\lambda_r$ is integer-valued and $2\le \lambda_r\le |R| + 1$. We will show that $\lambda_r$ decreases whenever $r$ performs a deviation at step $k$ where she lands in $F_{\partition_{k}}$, and can only increase through a deviation of $r$ in $(\tau_l)_{l=1}^L$. In particular, we will show that the potential does not increase if another agent performs a deviation, unless when $r\in F_{\partition_{k-1}}\setminus F_{\partition_{k}}$ which corresponds to the case of inserting a deviation by $r$, which also corresponds to a deviation in $(\tau_l)_{l=1}^L$.
	
	Consider a step $k$ in the dynamics where $d_k = r$ and $r\in F_{\partition_k}$. Then, $\frac 12 < f_r(\partition_k) < 1$. By Claim~\ref{lem:converg-HDG-structure}, $f_r(\partition_k) = \frac m{m+1}$ for some $2\le m\le |R|$ and therefore $\lambda_r(\partition_k) = m$. Also, by Claim~\ref{lem:converg-HDG-structure}, $r$ is not allowed to perform a deviation if $f_r(\partition_{k-1}) < \frac 12$ (as then an invalid homogeneous coalition of blue agents would remain). Hence, single-peakedness implies that $f_r(\partition_{k-1}) > \frac m{m+1}$, and therefore $\lambda_r(\partition_k) < \lambda_r(\partition_{k-1})$. 
	
	If $d_k = r$ and $r\notin F_{\partition_k}$, then a deviation happens where $\sigma_k \neq \sigma_{k-1}$, and therefore this corresponds to a deviation of $r$ in $(\tau_l)_{l=1}^L$. 
	Next, we want to inspect how $r$ is affected from a deviator if $d_k \neq r$. In this case, $d_k$ cannot join $\partition_{k-1}(r)$ if $\lambda_r(\partition_{k-1})\le |R|$ (that is, in the case where $r$ is in a coalition of `large' ratio). 
	Indeed, since $r\in R_S$, $r$ would block any red agent to join, and a blue agent cannot join due to Claim~\ref{lem:converg-HDG-structure}. Hence, $\partition_{k-1}(r)$ is only affected if $d_k\in \partition_{k-1}(r)$. By Claim~\ref{lem:converg-HDG-structure}, if $|\partition_{k-1}(r)| =2$, then $f_r(\partition) = \frac 12$, and the potential cannot go up. 
	Otherwise, Claim~\ref{lem:converg-HDG-structure} implies that $d_k$ is red. Since $r\in R_s$, it holds in addition that $r\notin F_{\partition_{k-1}}$. Hence, if $|\partition_{k-1}(r)|\ge 4$, then $\lambda_r(\partition_k) = \lambda_r(\partition_{k-1}) - 1$. If $|\partition_{k-1}(r)|= 3$, then $\partition_{k}(r)$ is of size $2$ and $\lambda_r(\partition_k) = |R| + 1$. As then  $r\in F_{\partition_{k-1}}\setminus F_{\partition_{k}}$, this case corresponds exactly to inserting the deviation of $r$ to form a coalition of size $2$ in $(\tau_l)_{l=1}^L$.
	
	Together, there can be at most $|R| - 1 \le n - 1$ successive deviations by $r$ corresponding to steps in $I$ until there is a deviation by $r$ in $(\tau_l)_{l=1}^L$. We obtain a bound for the deviations by $r$ which matter. To make this formal, we consider the following quantities. Given an agent $x$, define 
	$I_x = |\{k\in I \colon d_k = x\}|$ and $L_x = |\{1\le l \le L\colon \hat d_l = x\}|$. Since at least every $n$-th deviation counts towards $L_r$ but not towards $I_r$, we can conclude that $L_r - I_r \ge \lfloor \frac {L_r} n \rfloor\ge \frac {L_r} n - 1$.
	
	By an analogous argument where we consider an analogous potential function for blue agents, we obtain that $L_b - I_b \ge \frac {L_b} n - 1$ for every $b\in B_L$. Additionally, the definition of $F_{\partition}$ implies that $k\notin I$ if $d_k\in R_L$ or $d_k\in B_S$. Hence, for $x\in R_L\cup B_S$, it holds that $L_x - I_x = L_x \ge \frac {L_x} n - 1$ (where the latter inequality is of course a strong estimate, but it is all we need).
	
	Summing up the inequalities for all agents, we obtain
	
	$$P = L - |I| = \sum_{x\in N} L_x - I_x \ge \sum_{x\in N} \frac {L_x} n - 1 = \frac L n - n\text.$$
	
	Solving for $L$ yields the desired inequality.
\end{proof}

It remains to analyze the dynamics $(\rho_p)_{p=1}^P$. To this end, we will show that this dynamics essentially behaves like a specific AHG, where we have to replace some agents by multiple copies. This yields an AHG with at most $n^2$ agents. Hence, Theorem~\ref{thm:fastConvAHG} would provide a running time of $\mathcal O(n^6)$. However, we can do better. By exploiting structural properties of the AHG and a close inspection of the potentials in the proof of Theorem~\ref{thm:fastConvAHG}, we can reduce the running time of our transformed dynamics on the AHG to $\mathcal O(n^4)$.
 
\begin{claim}\label{lem:HDG-AHG-trafo}
	It holds that $P \in \mathcal O(n^4)$.
\end{claim}
\begin{proof}
	 First, let us note that, by construction of $(\rho_p)_{p=0}^P$, $F_{\rho_p} = \emptyset$ for all $0\le p\le P$.
	  Hence, all agents in $R_S$ (or $B_L$) only perform deviations towards singletons or coalitions of ratio at most $\frac 12$ (or at least $\frac 12$). Moreover, we may assume that an agent $r\in R_S$ never forms a coalition of size $2$ with an agent in $b\in B_L$. Due to single-peakedness and their respective peaks, this can only happen if both of them come from singleton coalitions. However, then no further agent can join $\{r,b\}$. Further red agents would be blocked by $r$ and further blue agents by $b$. Hence, this coalition can only be altered if one of these agents leaves. But this deviation can be performed right away from the singleton coalition. Similarly, we can exclude the formation of coalitions of size $2$ by agents in $R_L$ and $B_S$.
	 
	 As this shortcutting can only remove every second step (and an initial $n/2$ steps for forming a first set of pairs), it leaves us with a dynamics $(\rho'_p)_{p=0}^{P'}$ with $P'\ge \frac {P-n/2}2$ such that, for $1\le p\le P'$, $\rho'_p$ evolves from $\rho'_{p-1}$ through an IS deviation of some agent $d'_p$.
	 
	 Even more, we may assume that agents in $R_S$ (or $B_L$) never perform deviations. First, according to Claim~\ref{lem:converg-HDG-structure}, the only coalition which such an agent can leave is a coalition of size $2$ with ratio $\frac 12$. Hence, single-peakedness implies that forming the singleton coalition is not beneficial.
	 Furthermore, Claim~\ref{lem:converg-HDG-structure} implies that they could only form coalitions of size $2$. By the first part of the proof, their partner has to be from $B_S$ (or $R_L$). Since the preferences are strict, we may assume that their partner performs the deviation.
	 
	 Now, we define an AHG $(N^A, (\succ^A_x)_{x\in N^A})$ as follows. In principle, the only part of the preferences are on ratios $\frac 1{m+1}$ or $\frac m{m+1}$, and we want to identify these ratios with coalitions of size $m + 1$, because all coalitions of these ratios have exactly this size (using Claim~\ref{lem:converg-HDG-structure}). This part of the preferences will also inherit single-peakedness from the HDG. However, we have to deal with the preference over the ratio~$1$ for agents in $R_L$ (or over the ratio~$0$ for agents in $B_S$). To maintain single-peakedness, we should identify these ratios with coalition sizes $|R| + 1$ and $|B| + 1$, respectively. To achieve this goal, we introduce some auxiliary agents. Let the agent set of the AHG therefore be $N^A = R_S \cup B_L \cup \{r_0, \dots, r_{|R|}\colon r\in R_L\} \cup  \{b_0, \dots, b_{|B|}\colon b\in B_S\}$ and define strict and single-peaked preferences as follows (where we present only the relevant part of the preferences).
	 
	 \begin{itemize}
	 	\item If $r\in R_S$ and $2\le i,j \le |B| + 1$, then 
		\begin{itemize}
			\item $i\succ^A_r j$ if and only if $\frac 1i \succ_r \frac 1j$, 
			\item $1\succ^A_r i$ if and only if $1\succ_r \frac 1i$, and 
			\item $i\succ^A_r 1$ if and only if $\frac 1i \succ_r 1$.
		\end{itemize}
	 	\item If $b\in B_L$ and $2\le i,j \le |R| + 1$, then 
		\begin{itemize}
			\item $i\succ^A_b j$ if and only if $\frac {i-1}i \succ_b \frac {j-1}j$, 
			\item $1\succ^A_b i$ if and only if $0\succ_b \frac {i-1}i$, and 
			\item $i\succ^A_b 1$ if and only if $\frac {i-1}i \succ_b 0$.
		\end{itemize}
	 	\item If $r\in R_L$ and $2\le i,j \le |R|$, then 
		\begin{itemize} 
			\item $i\succ^A_{r_0} j$ if and only if $\frac {i-1}i \succ_{r} \frac {j-1}j$, 
			\item $|R|+1\succ^A_{r_0} i$ if and only if $1\succ_{r} \frac {i-1}i$,
			\item $i\succ^A_{r_0} |R|+1$ if and only if $\frac {i-1}i \succ_{r} 1$, and
			\item $|R|+1\succ^A_{r_l} |R|$ if $l\in [|R|]$.
		\end{itemize}
	 	\item If $b\in B_S$ and $2\le i,j \le |B|$, then 
		\begin{itemize}
			\item $i\succ^A_{b_0} j$ if and only if $\frac 1i \succ_{b} \frac 1j$, 
			\item $|B|+1\succ^A_{b_0} i$ if and only if $0\succ_{b} \frac 1i$,
			\item $i\succ^A_{b_0} |B|+1$ if and only if $\frac 1i \succ_{b} 0$, and
			\item $|B|+1\succ^A_{b_l} |B|$ if $l\in [|B|]$.
		\end{itemize}
	 \end{itemize}

	 Next, we define the modified dynamics. Therefore, let $0\le p \le P'$ and define the partition $\omega_p = \{\{r_0, \dots, r_{|R|}\}\colon r\in R_L, \{r\}\in \rho'_p\} \cup \{\{r_1, \dots, r_{|R|}\}\colon r\in R_L, \{r\}\notin \rho'_p\} \cup \{\{b_0, \dots, b_{|B|}\}\colon b\in B_S, \{b\}\in \rho'_p\} \cup \{\{b_1, \dots, b_{|B|}\}\colon b\in B_S, \{b\}\notin \rho'_p\} \cup \{\{b,r_0^1,\dots, r_0^m\}\colon b\in B_S, \rho'_p(b) = \{b,r^1,\dots, r^m\} \allowbreak \textnormal{ for } m\ge 0\} \cup \{\{r,b_0^1,\dots, b_0^m\}\colon r\in R_L, \rho'_p(r) = \{r,b^1,\dots, b^m\} \textnormal{ for } m\ge 0\}$.

	 Note that $\omega_p$ is well-defined, because every agent in $R_L$ (or $B_S$), which is not in a singleton coalition is part of a coalition solely consisting of agents in $R_L$ (or $B_S$), and a unique agent in $B_L$  (or $R_S$).
	 
	 Next, let $1\le p\le P'$. Then, $\omega_p$ evolves from $\omega_{p-1}$ through an IS deviation of some agent. This follows directly from the preferences in the AHG, where a fraction of $1$ (or $0$) plays the role of the coalition size $|R|+1$ (or $|B|+1$) for agents in $R_L$ (or $B_S$).
	 Hence, $(\omega_p)_{p=0}^{P'}$ is an execution of an individual dynamics in AHG $(N^A, (\succ^A_x)_{x\in N^A})$. 
	 
	 To bound its running time, we have to inspect the potentials in the proof of Theorem~\ref{thm:fastConvAHG}. First, $v_j^{P'}\le n$ for all agents $j\in N^A$. Second, $v_C^{P'}\le n$ for all $C\in \omega_{P'}$, and $|\omega_{P'}|\le 2n$. The latter bounds hold, because the copies of every original agent are only part of at most $2$ coalitions. Hence, $\Lambda(\omega_{P'})\le n^3 + 2n^2$, and there can be at most that many R-moves. Moreover, $\Lambda(\omega_k)-\Lambda(\omega_{k-1}) \le n - 1$ for every R-move. Hence, as in the proof of Theorem~\ref{thm:fastConvAHG}, we obtain a bound of $n^4 + 2n^3$ L-moves. Hence, the dynamics on the AHG runs for $P' \in \mathcal O(n^4)$ steps. Therefore, as $P\le 2P' + \frac n2$, we obtain $P \in \mathcal O(n^4)$.
\end{proof}
	\renewcommand\qedsymbol{$\square$}
Finally, we can combine all of our insights. Recall that $(\tau_l)_{l = 1}^L$ is longer than $(\sigma_k)_{k=1}^K$ because the only difference is the insertion of certain deviations. We can apply Claim~\ref{lem:HDG-omittance} and Claim~\ref{lem:HDG-AHG-trafo} to obtain

$$K \le L \overset{\mathrm{Claim~\ref{lem:HDG-omittance}}}{\le}  n^2 + n P \overset{\mathrm{Claim~\ref{lem:HDG-AHG-trafo}}}{\in} \mathcal O(n^5)\text.$$
\end{proof}

Under strict preferences, checking the existence of a path to stability and convergence are hard. 

\begin{restatable}{theorem}{existHDG}\label{thm:hardness-existencepath-HDG}
\existpb{IS}{HDG} is \np-hard and \convergpb{IS}{HDG} is \conp-hard, even for strict preferences.
\end{restatable}

\section{Fractional Hedonic Games}

Next, we study fractional hedonic games, which are closely related to hedonic diversity games, but instead of agent types, utilities rely on a cardinal valuation function of the other agents.
In fractional hedonic games, the existence of IS partitions is rare. The only known condition for their existence in previous work is that IS---and even NS---partitions exist if all utilities are non-negative. But in such games, utilities over coalitions are non-negative, and therefore the grand coalition is stable because the only possible deviation leaves agents in singleton coalitions where their utility is $0$. In particular, it was not even known whether symmetry of utilities is helpful \citep{BBS14a}. In the first result of this section, we answer this question negatively.\footnote{Symmetry was a reasonable candidate for the existence of IS partitions because it yields existence of Nash stability in additively separable hedonic games \citep{BoJa02a}.} Moreover, we demonstrate that dynamics offer a more fine-grained view on IS when weights are non-negative.

The first part of this section deals with symmetric games, the second part with simple games, that is, games where the utility function only attains values in $\{0,1\}$.
We start with the non-existence of IS partitions, where we provide a counterexample using $15$ agents. The weights were found with the help of a computer.

\begin{theorem}\label{thm:symFHGnoIS}
	There exists a symmetric FHG without an IS partition.
\end{theorem}

\begin{proof}

Define the sets of agents $N_i = \{a_i,b_i,c_i\}$ for $i \in \{1,\dots, 5\}$ and consider the FHG on the agent set $N = \bigcup_{i=1}^5 N_i$ where symmetric weights are given by

\begin{itemize}
	\item $\util(a_i, b_i) = \util(b_i, c_i) = \util(a_i,c_i) = 228, i\in \{1,\dots, 5\}$,
	\item $\util(a_i,a_{i+1}) = 436, \util(a_i,b_{i+1}) = 228, \util(a_i,c_{i+1}) = 248, i\in \{1,\dots, 5\}$,
	\item $\util(b_i,a_{i+1}) = 223, \util(b_i,b_{i+1}) = 171, \util(b_i,c_{i+1}) = 236, i\in \{1,\dots, 5\}$,
	\item $\util(c_i,a_{i+1}) = 223, \util(c_i,b_{i+1}) = 171, \util(c_i,c_{i+1}) = 188, i\in \{1,\dots, 5\}$, and
	\item $\util(x,y) = - 2251$ for all agents $x,y\in N$ such that the weight is not defined yet.
\end{itemize}

In the above definition, all indices are to be read modulo $5$ (where the modulo function is assumed to map to $\{1,\dots, 5\}$). Note that the large negative weight exceeds the sum of positive weights incident to any agents. Hence, agents linked by a negative weight, can never be in a common coalition in any IS partition. The FHG consists of five triangles that form a cycle. The structure of the game is illustrated in Figure~\ref{fig:counterex_FHG}. While proving that there does not exist an IS partition requires a lengthy case distinction and many computations, 
the global intuition for the proof is to observe that IS dynamics in this instance always cycle. To see this, start with the partition $(N_5\cup N_1, N_2, N_3, N_4)$. First, $a_1$ deviates by joining $N_2$. Then, $b_1$ joins this new coalition, then $c_1$. After this step, we are in an isomorphic state as in the initial partition.

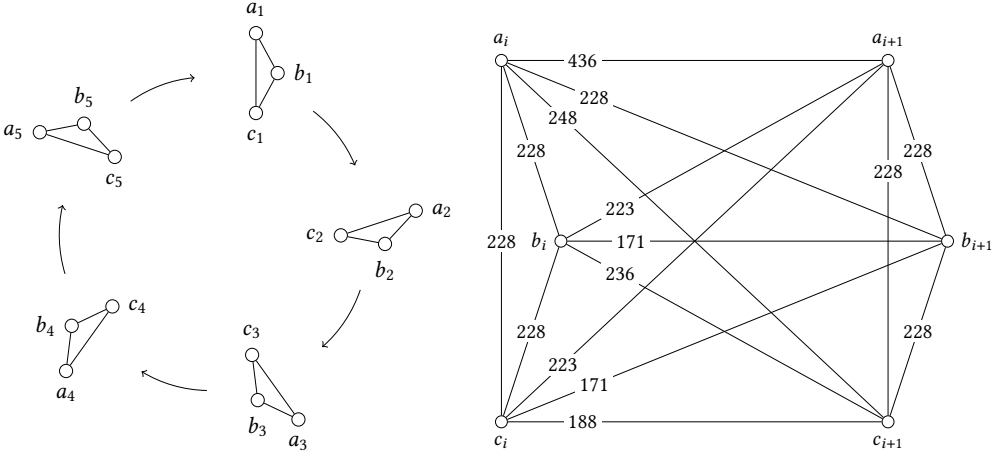
\begin{figure}
\centering
\begin{subfigure}{.45\columnwidth}
  \centering
    \resizebox{\columnwidth}{!}{
	\begin{tikzpicture}[auto,
		complexnode/.pic={
		\node[protovertex] (b) at (0,-0.1){};
		\node[protovertex] (a) at (0,1.1){};
		\node[protovertex] (c) at (0.33,0.5){};
		\path (a) edge (b);
		\path (b) edge (c);
		\path (c) edge (a);},
		complexarc_prev/.pic={
		\node (s) at (-.65,0.3){};
		\node (t) at (.65,0.3){};
		\draw[bend left,->]  (s) to (t);},
		complexarc/.pic={
	\draw[->] (-.6,0.3) arc[radius= 2cm, start angle = 105, end angle= 75];}
	]
		\node[regular polygon,regular polygon sides=10,minimum size=4cm] (p) at (0,0){};
		\draw (p.corner 1) pic {complexnode};
	\node[label = 90:{$a_1$}] at (a) {};
	\node[label = 270:{$c_1$}] at (b) {};
	\node[label = 0:{$b_1$}] at (c) {};
		\draw (p.corner 2) pic[rotate = 20] {complexarc};
		\draw (p.corner 3) pic[rotate = 72] {complexnode};
	\node[label = 180:{$a_5$}] at (a) {};
	\node[label = 270:{$c_5$}] at (b) {};
	\node[label = 90:{$b_5$}] at (c) {};
		\draw (p.corner 4) pic[rotate = 92] {complexarc};
		\draw (p.corner 5) pic[rotate = 144] {complexnode};
	\node[label = 270:{$a_4$}] at (a) {};
	\node[label = 0:{$c_4$}] at (b) {};
	\node[label = 180:{$b_4$}] at (c) {};
		\draw (p.corner 6) pic[rotate = 164] {complexarc};
		\draw (p.corner 7) pic[rotate = 216] {complexnode};
	\node[label = 270:{$a_3$}] at (a) {};
	\node[label = 90:{$c_3$}] at (b) {};
	\node[label = 270:{$b_3$}] at (c) {};
		\draw (p.corner 8) pic[rotate = 236] {complexarc};
		\draw (p.corner 9) pic[rotate = 288] {complexnode};
	\node[label = 0:{$a_2$}] at (a) {};
	\node[label = 180:{$c_2$}] at (b) {};
	\node[label = 270:{$b_2$}] at (c) {};
		\draw (p.corner 10) pic[rotate = 308] {complexarc};
	\end{tikzpicture}}
  \caption{Five triangles are ordered in a cycle such that there is a tendency of agents in $N_i$ to deviate to coalitions in $N_{i+1}$.}
  \label{fig:sub1}
\end{subfigure}%
\hfil
\begin{subfigure}{.5\columnwidth}
  \centering
  \resizebox{\columnwidth}{!}{
\begin{tikzpicture}[auto]
	\node[protovertex,label = 270:$c_i$] (b1) at (0,0){};
	\node[protovertex,label = $a_i$] (a1) at (0,6){};
	\node[protovertex,label = 180:$b_i$] (c1) at (1,3){};
	\node[protovertex,label = 270:$c_{i+1}$] (b2) at (6.5,0){};
	\node[protovertex,label = $a_{i+1}$] (a2) at (6.5,6){};
	\node[protovertex,label = 0:$b_{i+1}$] (c2) at (7.5,3){};
	\draw (a1) edge node[fill=white,anchor=center, pos=0.5, inner sep =2pt] {$228$} (b1);
	\draw (b1) edge node[fill=white,anchor=center, pos=0.5, inner sep =2pt] {$228$} (c1);
	\draw (c1) edge node[fill=white,anchor=center, pos=0.5, inner sep =2pt] {$228$} (a1);
	\draw (a2) edge node[fill=white,anchor=center, pos=0.3, inner sep =2pt] {$228$} (b2);
	\draw (b2) edge node[fill=white,anchor=center, pos=0.5, inner sep =2pt] {$228$} (c2);
	\draw (c2) edge node[fill=white,anchor=center, pos=0.5, inner sep =2pt] {$228$} (a2);
	\draw (a1) edge node[fill=white,anchor=center, pos=0.2, inner sep =2pt] {$436$} (a2);
	\draw (a1) edge node[fill=white,anchor=center, pos=0.15, inner sep =2pt] {$248$} (b2);
	\draw (a1) edge node[fill=white,anchor=center, pos=0.2, inner sep =2pt] {$228$} (c2);
	\draw (b1) edge node[fill=white,anchor=center, pos=0.15, inner sep =2pt] {$223$} (a2);
	\draw (b1) edge node[fill=white,anchor=center, pos=0.2, inner sep =2pt] {$188$} (b2);
	\draw (b1) edge node[fill=white,anchor=center, pos=0.2, inner sep =2pt] {$171$} (c2);
	\draw (c1) edge node[fill=white,anchor=center, pos=0.17, inner sep =2pt] {$223$} (a2);
	\draw (c1) edge node[fill=white,anchor=center, pos=0.17, inner sep =2pt] {$236$} (b2);
	\draw (c1) edge node[fill=white,anchor=center, pos=0.17, inner sep =2pt] {$171$} (c2);
\end{tikzpicture}}
  \caption{The transition weights between the triangles allow for infinite loops of deviations.}
  \label{fig:sub2}
\end{subfigure}
\caption{%
Description of the symmetric FHG without IS partition in the proof of Theorem~\ref{thm:symFHGnoIS}.}
\label{fig:counterex_FHG}
\end{figure}

We are ready for the main proof. Let $\partition$ be any partition of the agents and assume that $\partition$ is IS. In particular, no agent receives negative utility. Therefore there exists an $i\in \{1,\dots, 5\}$ such that $\partition(a_i)\cap \{a_1,\dots,a_5\}= \{a_i\}$. We may assume, without loss of generality, that $a_1$ is such an agent. In the following, we will 
distinguish all the possible cases for the coalition of $a_1$ in $\partition$ and show that none of them can occur in an IS partition, deriving a contradiction.

\begin{itemize}
	\item Goal 1: $b_2\notin \partition(a_1)$.

	First, assume for contradiction that $b_2\in \partition(a_1)$, which implies that $\partition(a_1)\subseteq N_1\cup N_2$. If $\partition(a_2)\subseteq N_1\cup N_2$, then $\partition(a_2)\subseteq \{a_2, c_2, c_1, b_1\}$, and therefore $\util_{a_2}(\partition) \le 168.5$ (the best case is $\partition(a_2) = \{a_2, c_2, c_1, b_1\}$ with $\util_{a_2}(\{a_2, c_2, c_1, b_1\}) = \frac {228 + 223 + 223}{4} = 168.5$), while $\util_{a_2} (\{a_2\}\cup \partition(a_1))\ge 221.3$ (the worst case is $\partition(a_1) = \{a_1,b_2\}$ with $\util_{a_2}(\{a_2, a_1, b_2\}) = \frac {436 + 228}{3} = 221.33{\dots}$). 
	Hence, $a_2$ has an incentive to deviate (making no agent in $\partition(a_1)$ worse). Therefore, $\partition(a_2)\subseteq N_2\cup N_3$ and more precisely $\partition(a_2) \subseteq \{a_2,c_2,a_3,c_3,b_3\}$. In addition, by the same potential deviation, $a_3\in \partition(a_2)$ and $|\partition(a_2)|\ge 3$.

	Next, consider the case that $c_2\in \partition(a_2)$. Then, $\util_{b_2}(\partition)\le 142.5$, while $\util_{b_2}(\{b_2\}\cup\partition(a_2))\ge 169.75$ and $b_2$ would have a beneficial IS deviation. Hence, $c_2\notin \partition(a_2)$. 
	If $c_2\notin \partition(a_1)$, then $\util_{c_2}(\partition)\le \max\{141.34,119.67\} = 141.34$ (this is if $c_2$ forms a coalition with $N_1\setminus \{a_1\}$ or $N_3\setminus \{a_3\}$), while $\util_{c_2}(\{c_2\}\cup\partition(a_1))\ge 158.6$ and it is easily seen that $c_2$ can only improve agents in $\partition(a_1)$. It follows that $c_2\in \partition(a_1)$.

	If $c_3\notin \partition(a_2)$, then $\partition(a_2) = \{a_2,a_3,b_3\}$ and $c_3$ would deviate by joining $\partition(a_2)$. Hence, $\{a_2,a_3,c_3\}\subseteq \partition(a_2)$. But then $\util_{b_2}(\partition)\le 159.6$ (the best case being $\partition(b_2) = N_1\cup\{b_2,c_2\}$), while $\util_{b_2}(\{b_2\}\cup \partition(a_2)) \ge 171.6$ (the worst case here is $\partition(b_2) = N_3\cup\{a_2,b_2\}$ which is worse than the smaller $\{a_2,b_2,a_3,c_3\}$) and joining with $b_2$ makes no agent worse. In conclusion, the initial assumption was wrong and $b_2\notin \partition(a_1)$.

	\item Goal 2: $c_2\notin \partition(a_1)$.

	Second, assume for contradiction that $c_2\in \partition(a_1)$. As in the previous case, it is easily seen that $\partition(a_2)\subseteq \{a_2,b_2\}\cup N_3$, $a_3\in \partition(a_2)$, and $|\partition(a_2)|\ge 3$. If $b_2\notin \partition(a_2)$, then $\util_{b_2}(\partition) \le 118$ (the best coalitions in $N_1\cup N_2$ and $N_2\cup N_3$ are $\{b_2,b_1,c_1\}$ and $\{b_2,c_3\}$, respectively) while $\util_{b_2}(\{b_2\}\cup \partition(a_2))\ge 155.5$. Hence, $b_2\in \partition(a_2)$. But then $\util_{c_2}(\partition) \le 168$ while $\util_{c_2}(\{c_2\}\cup \partition(a_2))\ge 169.75$ and $c_2$ would join $\partition(a_2)$ making no agent worse. We conclude that $c_2\notin \partition(a_1)$ and can therefore assume that $\partition(a_1)\subseteq N_1\cup (N_5\setminus\{a_5\})$.

	\item Goal 3: $c_1\notin \partition(a_1)$.

	Third, assume for contradiction that $c_1\in \partition(a_1)$. Then, $\util_{a_5}(\partition)\le \max\{223,171\} = 223$ (where the first utility in the maximum refers to the coalition $N_4\cup N_5$ and the second utility to $N_5\cup \{b_1\}$). However, $\util_{a_5}(\{a_5\}\cup \partition(a_1))\ge 228$. Since joining $\partition(a_1)$ with $a_5$ makes no agent worse, this is not possible. Hence, $c_1\notin \partition(a_1)$.

	\item Goal 4: $b_1\notin \partition(a_1)$.

	Forth, assume for contradiction that $b_1\in \partition(a_1)$. Then, $\util_{c_1}(\{c_1\}\cup \partition(a_1))\ge 152$ and adding $c_1$ to $\partition(a_1)$ leaves no agent worse off. Since, $\util_{c_1}(\{c_1\}\cup N_2) = 145.5$, it must hold that $\partition(c_1)\subseteq \{c_1\}\cup N_5$ and even $\{a_5,b_5\}\subseteq \partition(c_1)$ since otherwise $\util_{c_1}(\partition)\le 145.4$. But then $\util_{a_1}(\partition)\le 150.4$ while $\util_{a_1}(\{a_1\}\cup\partition(c_1)) \ge 221.75$ and $a_1$ would deviate making no agent worse. It follows that $b_1\notin \partition(a_1)$.

	\item Goal 5: $\partition(a_1)\not\subseteq \{a_1,c_5,b_5\}$.

	It remains the case that $\partition(a_1)\subseteq \{a_1,c_5,b_5\}$. If $|\partition(c_1)|\ge 2$, then $a_1$ would deviate by joining $\partition(c_1)$, making no agent worse. If, however, $c_1$ is in a singleton coalition, then $c_1$ would join $\partition(a_1)$, making no agent worse and improving her utility.
\end{itemize}
It follows that no coalition for agent $a_1$ can be possible in an IS partition $\partition$, implying that the instance admits no IS partition.
\end{proof}

Employing this counterexample, the methods of \citet{BBS14a}, which originate from hardness constructions of \citet{SuDi10a}, can be used to show that it is \np-hard to decide about the existence of IS partitions in symmetric FHGs.

For the following corollary, we omit the full proof because it is analogous to the weaker statement by \citet[Theorem~5]{BBS14a}. The main method will also be applied in the proof of Theorem~\ref{thm:hardness-existencepath-symposFHG}, which considers convergence of the IS dynamics in the case that the FHGs even have symmetric, \emph{non-negative} weights.

\begin{corollary}
	Deciding whether there exists an individually stable partition in symmetric FHGs is \np-hard.
\end{corollary}

\begin{proof}[Proof sketch]
	In the reduction by \citet[Theorem~5]{BBS14a}, we replace the non-symmetric gadget by an agent-minimal symmetric FHG that admits no IS partition. Such an FHG exists according to Theorem~\ref{thm:symFHGnoIS}.\footnote{In fact, we can use a subgame of the game constructed in Theorem~\ref{thm:symFHGnoIS}. Note that we have no proof of agent-minimality of this game but we can simply remove agents until it is agent-minimal.} The weights in the symmetric part of their reduced instances must be large enough to incentivize the agent in the gadget to stay in a coalition outside the gadget.
\end{proof}

If we consider symmetric, non-negative utilities, the grand coalition forms an NS, and therefore IS, 
partition of the agents. However, deciding about the convergence of the IS dynamics starting with the \singleton is \np-hard. The reduction uses similar methods as \citet{SuDi10a} and \citet{BBS14a}. We can avoid negative weights by the fact that, due to symmetry of the weights, in a dynamics starting with the \singleton, all coalitions that can be obtained in the process must have strictly positive mutual utility for all pairs of agents in the coalition.

\begin{restatable}{theorem}{symFHG}\label{thm:hardness-symposFHG}
	\existpb{IS}{FHG} is \np-hard and \convergpb{IS}{FHG} is \conp-hard, even in symmetric FHGs with non-negative weights. The former is even true if the initial partition is the \singleton.
\end{restatable}

From now on, we consider simple FHGs. We start with the additional assumption of symmetry.

\begin{proposition}\label{prop:convergence-simplesymFHG}
	The dynamics of IS deviations starting from the \singleton converges in simple symmetric FHGs in at most $\mathcal O(n^2)$ steps. The dynamics may take $\Omega(n\sqrt{n})$ steps. 
\end{proposition}

\begin{proof}
	We start with the lower bound. Consider the FHG induced by the complete graph on $n = k(k+1)/2$ agents for some non-negative integer $k\ge 1$. We partition the agents arbitrarily into sets $C_1,\dots, C_k$ where $|C_j| = j$ for $j = 1, \dots, k$. Now, we perform two phases of IS deviations. In the first phase, we form the coalitions $C_j$ by having agents join one by one. In the second phase, there are $k-1$ steps. In step $j$, the agents of coalition $C_j$ join coalitions $C_{j+1},\dots, C_k$ one after each other, thereby performing $k-j$ deviations each. The total number of deviations in the second phase is therefore $\sum_{j=1}^{k-1}j\cdot (k-j) = \frac 16 (k-1) k (k+1) = \Theta(k^3) = \Theta(k^2\sqrt{k^2}) = \Theta (n\sqrt{n})$. In particular, there can be $\Omega (n\sqrt{n})$ IS deviation steps starting from the \singleton.

	For the upper bound, let a simple and symmetric FHG be given. Note that all coalitions formed through the deviation dynamics are cliques. Hence, every deviation step will increase the total number of edges in all coalitions. More precisely, the dynamics will increase the potential $\potential(\partition) = \sum_{C\in\partition} |C|(|C|-1)/2$ in every step by at least $1$. Since the total number of edges is bounded by $n(n-1)/2$, this proves the upper bound.
\end{proof}

Note that there is a simple way to converge in a linear number of steps starting with the \singleton by forming largest cliques and removing them from consideration\footnote{The number of steps is linear even though finding such a sequence for an external coordinator would be computationally hard because it would require to solve a maximum clique problem.}. Surprisingly, it seems a lot harder to prove (non-)convergence of the dynamics if we start from an arbitrary partition, and we leave this as an open problem.

If we allow for asymmetries, the dynamics is not guaranteed to converge anymore. For instance, the IS dynamics on an FHG induced by a directed triangle will not converge for any initial partition except for the grand coalition. We can, however, characterize convergence on simple asymmetric FHGs. Tractability depends on the structure of the utilities. First, we consider dynamics 
starting from the \singleton on simple asymmetric FHGs. 

The key insight is that throughout the dynamic process on a simple asymmetric FHG starting from the \singleton, the subgraphs induced by coalitions are always transitive and complete.\footnote{We identify the directed graphs $G = (V,A)$ with relations $R$ where $i\mathrel{R}j$ if and only if $(i,j)\in A$. Hence, we call a directed graph $G = (V,A)$ \emph{transitive} if $(i,j)\in A$ and $(j,k)\in A$ implies $(i,k)$ for every triple $i,j,k\in V$. We call it \emph{complete} if $(i,j)\in A$ or $(j,i)\in A$ for every pair $i,j\in V$.} Convergence is then shown by a potential function argument.

\begin{proposition}\label{thm:convergence-asymmetric+acyclicFHG}
	The dynamics of IS deviations starting from the \singleton converges in simple asymmetric FHGs if and only if the underlying graph is acyclic. Moreover, under acyclicity, it converges in $\mathcal O(n^4)$ steps.
\end{proposition}

\begin{proof}
	Let $G= (V,A)$ be a simple asymmetric graph on $n = |V|$ vertices.
	If the graph contains a cycle, it is easy to find a non-converging series of deviations. There exists a cycle of length at least $3$. We can then let a coalition of size $2$ propagate along the cycle. More formally, assume that $\{c_1,\dots, c_j\}\subseteq V$ induce a directed cycle, where $(c_i,c_{i+1})\in A$ for $1\le i \le j$ (here and in the remainder of the proof, read indices modulo $j$ mapping to the representative in $[j]$). Define $(\partition_k)_{k\ge 0}$ by letting $\spartition$ be the \singleton, and for $p\ge 0$, $1\le i\le j$, let $\partition_{pj + i} = \{\{c_i,c_{i+1}\}\}\cup \{\{x\}\colon x\in V \setminus \{c_i,c_{i+1}\}\}$. Then, $(\partition_k)_{k\ge 0}$ is an IS dynamics of infinite length.
	
	Assume that the graph is acyclic. Our first observation is that, in every step of the dynamics, all subgraphs induced by coalitions are transitive and complete. Indeed, by induction, in a deviation, the coalition that is left still induces a transitive and complete subgraph, and the new coalition induced a transitive and complete subgraph before the deviation. 
	Hence, every agent except one has at least one outgoing edge and will only accept the new agent if she likes her. 
	Since the deviating agent must have non-negative utility after the deviation, she needs to approve the single agent without outgoing edge. 
	Hence, the newly formed coalition still induces a transitive and complete subgraph. 
	
	The last argument also implies that the deviating agent has a utility of $1/k$ if she ends up in a coalition of size $k$ after her deviation. We refer to this fact as $(*)$ in the sequel.

	We will now define two potentials based on the agents that receive $0$ utility in a partition, and based on the coalition sizes. The first potential is monotonically decreasing and bounded. The second potential is strictly increasing whenever the first potential is not strictly decreasing, and bounded. Hence, we establish convergence of the dynamics.

	First, fix a topological order of the agents, i.e., a bijection $\sigma: V \to [n]$ such that for all $(v,w)\in A$, $\sigma(v) < \sigma(w)$.
	For a given partition $\partition$ of the agents, we define the vector $v^\sigma(\partition)$ of length $|\partition|$ that sorts the numbers $\max_{i\in C}\sigma(i)$ for $C\in \partition$ in decreasing order, that is it sorts the coalitions in decreasing topological score of the agent with the highest number due to the topological order. This is exactly the unique agent in every coalition receiving $0$ utility.
	In addition, we define the vector $w(\partition)$ of length $|\partition|$ that sorts the coalition sizes in increasing order. Note that this vector does not depend on the underlying topological order.

	For two vectors $v = (v_i)_{i=1}^k$ and $w = (w_i)_{i=1}^l$, not necessarily of the same length, we say
	\begin{align*}
	v >_{lex} w \iff &\textnormal{there is } i < \max\{k,l\} \textnormal{ with }\\ & v_j = w_j \,\forall 1\le j\le i\textnormal{ and } v_{i+1}>w_{i+1},\textnormal{ or }\\
	& k > l \textnormal{ and } v_j = w_j \,\forall 1\le j \le k\text.
	\end{align*}

	In other words, $v >_{lex} w$ if $v$ is lexicographically greater than $w$.

	The key insight is that, for $\partition'$ formed from $\partition$ by an IS deviation, $v^\sigma(\partition') <_{lex} v^\sigma(\partition)$, or  $v^\sigma(\partition') =_{lex} v^\sigma(\partition)$ and  $w(\partition') >_{lex} w(\partition)$. 
	For a proof, assume that $\partition'$ is formed from $\partition$ by an IS deviation of agent $i$. 
	Note that $\max_{j\in \partition'(i)}\sigma(j) = \max_{j\in \partition'(i)\setminus \{i\}}\sigma(j) $. We distinguish two cases. Either $i = \arg\max_{j\in \partition(i)}\sigma(j)$ and it follows $v^\sigma(\partition') <_{lex} v^\sigma(\partition)$. 
	Otherwise, $\util_i(\partition(i))\ge \frac 1{|\partition(i)|}$, and because $i$ is improving her utility, $\frac 1{|\partition'(i)|} \overset{(*)}{=} \util_i(\partition') > \frac 1{|\partition(i)|}$. It follows that $|\partition(i)| > |\partition'(i)|$. Hence, $|\partition'(i)|-1 < \min\{|\partition'(i)|,|\partition(i)|\}$, and therefore $w(\partition') >_{lex} w(\partition)$.

We estimate the running time in two steps. First, we bound the number of times that the lexicographic score of  $v^\sigma(\partition)$ can decrease. Then, we estimate the number of deviations that can happen while this score does not change. We call the first kind of deviations \emph{primal} and the second type \emph{secondary}. Note that after an deviation, the maximal topological score in the joined coalition remains the same because the deviating agent has to like some agent who therefore has higher topological score. Hence, a primal deviation happens if and only if the agent with highest topological score of the abandoned coalition performs the deviation.

Let us first discuss the idea how to bound the number of the primal deviations. To this end, given a partition $\pi$ and an agent $v\in V$, we define a set $D^{\partition}_v$ that stores a certain amount of deviating agents. This set depends on the agent $v$ and the history of the dynamics until reaching $\pi$. In every step of the dynamics, the sum $\sum_{v\in V} |D^{\partition}_v|$ will be exactly the number of primal deviations so far. We ensure that we can always add the agent $v$ performing a deviation to a set $D^{\partition}_w$ such that $\sigma(v) > \sigma(x)$ for all $x\in D^{\partition}_w$. Hence, at the end of the sequence of deviations, $\sum_{v\in V} |D^{\partition}_v|\le n^2$.

Initially, set $D^{\spartition}_v = \emptyset$ for all $v\in V$ and the starting partition $\spartition$ of the dynamics. Assume first that agent $v$ performs a primary deviation that changes partition $\partition$ into partition $\partition'$. Recall that in this case, $v = \arg\max_{x\in \partition(v)} \sigma(x)$. If $v$ was in a singleton coalition, update $D^{\partition'}_v = \{v\}$ and leave all other sets the same, i.e., $D^{\partition'}_x = D^{\partition}_x$ for all $x\neq v$. Otherwise, let $w = \arg\max_{x\in \partition(v)\setminus \{v\}} \sigma(x)$ be the agent in $\partition(v)$ different from $v$ of highest topological score, i.e., the agent in $\partition(v)$ of second-highest topological score. We update $D^{\partition'}_v = D^{\partition}_w \cup \{v\}$, $D^{\partition'}_w = \emptyset$, and $D^{\partition'}_x = D^{\partition}_x$ for all $x\neq v,w$.
If a secondary deviation is performed from $\partition$ to $\partition'$, leave all sets the same, i.e., $D^{\partition'}_x = D^{\partition}_x$ for all $x\in V$.

Given a set of agents $W\subseteq V$, let $m_W = \arg\max_{x\in W} \sigma(x)$ be the agent in $W$ maximizing the topological score. We have the following invariants for every partition $\partition$ during the dynamics and for every agent $v\in V$:

\begin{itemize}
	\item If $v = m_{\partition(v)}$, then $D^{\partition}_v = \emptyset$.
	\item If $v \neq m_{\partition(v)}$, then $\sigma(x) \le \sigma(v) < \sigma(m_{\partition(v)})$ for all $x\in D^{\partition}_v$.
	\item The number of primal deviations of the dynamics until partition $\partition$ is $\sum_{v\in V} |D^{\partition}_v|$.
\end{itemize}

The first invariant follows directly from the update rules. Indeed, the agent in the newly formed coalition of maximal topological score is the same, and if the agent of highest topological score in the abandoned coalition changes, then we update her set to be the empty set. This proves the first invariant.

The second invariant follows by induction. Assume that $v$ performs a deviation from $\partition$ to $\partition'$. If $v$ performs a primary deviation and $w = \arg\max_{x\in \partition(v)\setminus \{v\}} \sigma(x)$, then $D^{\partition'}_v\setminus\{v\} = D^{\partition}_w$, and therefore $\sigma(x) \le \sigma(w) < \sigma(m_{\partition(w)}) = \sigma(v) < \sigma(m_{\partition'(v)})$ for all $D^{\partition'}_v\setminus\{v\}$ where we apply induction for $w$ and the fact that the agent in $\partition'(v)$ which gives positive utility to $v$ has a higher topological score than $v$. If $v$ performs a secondary deviation, then for all $x\in D^{\partition'}_v$, $\sigma(x)\le \sigma(v) < \sigma(m_{\partition'(v)})$, where the first inequality follows by induction for $v$. 

The third invariant follows from the update rules because the agent newly added to a set has not been in this set due to the second invariant.
The third invariant implies that there can be at most $n^2$ primal deviations, because for the terminal partition $\partition^*$ of the dynamics, $\sum_{v\in V} |D^{\partition^*}_v| \le n^2$.

While the topological score is the same, there can be at most $n^2$ secondary deviations, which follows from the same reasoning as in the proof of Proposition~\ref{prop:convergence-simplesymFHG}. Hence, together there are at most $n^4$ deviations.
\end{proof}

In the previous proposition, it seems that there is still space for improvement of the bound on the running time, in particular due to the interplay of the two nested potentials.

The previous statement shows convergence of the dynamics for simple asymmetric, acyclic FHGs. In addition, it is easy to see that there is always a sequence converging after $n$ steps, starting with the \singleton.
More precisely, one can use a topological order of the agents and allow agents to deviate in decreasing topological order. 
It can be observed that only coalitions of at most size two would form and that no agent would deviate from her current non-singleton coalition. 
The idea is that when an agent $i$ wants to join a coalition of at least two agents, there must exist an agent $j$ in this coalition, with a greater index in the topological order and non-null current utility, therefore agent $j$ does not want agent $i$ to join. 

There are two interesting further directions. One can weaken either the restriction on the initial partition or on asymmetry. If we allow for general initial partitions, we immediately obtain hardness results for simple asymmetric FHGs which are in particular a subclass of simple FHGs.

\begin{restatable}{theorem}{hardExPathAsymFHG}\label{thm:hardness-existencepath-asymFHG}
	\existpb{IS}{FHG} is \np-hard and \convergpb{IS}{FHG} is \conp-hard, even in simple asymmetric FHGs.
\end{restatable}

On the other hand, if we transition to simple FHGs while starting the dynamics from the \singleton, the problem of deciding whether a path to stability exists becomes hard. We leave the complexity of \convergpb{IS}{FHG} for simple FHG as an open problem.

\begin{restatable}{theorem}{hardnessSimpleFHG}\label{thm:hardness-existencepath-simpleFHG}
	\existpb{IS}{FHG} is \np-hard even in simple FHGs when starting from the \singleton. 
\end{restatable}

\section{Dichotomous Hedonic Games}

By taking into account the identity of other agents in the preferences of agents over coalitions, it can be more complicated to get positive results regarding individual stability (see, e.g., Theorem~\ref{thm:symFHGnoIS}). 
However, by restricting the evaluation of coalitions to dichotomous preferences, the existence of an IS partition is guaranteed~\citep{Pete16a}, as well as convergence of the dynamics of IS deviations, when starting from the grand coalition~\citep{BoEl20a}.
Nevertheless, the convergence of the dynamics is not guaranteed for an arbitrary initial partition and no sequence of IS deviations may ever reach an IS partition.

\begin{proposition}\label{prop:noconv-dicho}
The dynamics of IS deviations may never reach an 
IS partition in 
DHGs, whatever the chosen path of deviations, even when starting from the \singleton.
\end{proposition}

\begin{proof}
Let us consider an instance of a DHG 
with three agents.
Their preferences are described in the table below.

\medskip
{\centering
\begin{tabular}{*{4}{c}}
\toprule
Agent & 1 & 2 & 3 \\
\midrule
Approvals & $\{1,2\}$ & $\{2,3\}$ & $\{1,3\}$ \\
Disapprovals & $\{1\},\{1,3\},\{1,2,3\}$ & $\{2\},\{1,2\},\{1,2,3\}$ & $\{3\},\{2,3\},\{1,2,3\}$ \\ 
\bottomrule
\end{tabular}\par}
\medskip

There is a unique IS partition which consists of the grand coalition $\{1,2,3\}$.
We represent below all possible IS deviations between all the other possible partitions.
An IS deviation between two partitions is indicated by an arrow mentioning the name of the deviating agent.

\medskip
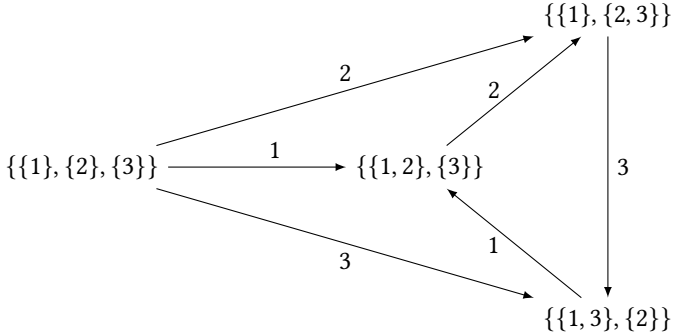
{\centering
\begin{tikzpicture}
\node (1) at (0,0) {$\{\{1\},\{2\},\{3\}\}$};
\node (2) at ($(1)+(4.5,0)$) {$\{\{1,2\},\{3\}\}$};
\node (3) at ($(2)+(2.5,2)$) {$\{\{1\},\{2,3\}\}$};
\node (4) at ($(2)+(2.5,-2)$) {$\{\{1,3\},\{2\}\}$};
\draw[arrow] (1) --node[pos=0.6,above]{$1$} (2); \draw[arrow] (2) --node[pos=0.35,above]{$2$}  (3); \draw[arrow] (3) --node[midway,right]{$3$}  (4); \draw[arrow] (4) --node[pos=0.65,below]{$1$}  (2);
\draw[arrow] (1) --node[midway,above]{$2$}  (3); \draw[arrow] (1) --node[midway,below]{$3$}  (4);
\end{tikzpicture}\par}
\medskip

One can check that the described deviations are IS deviations. 
A cycle is necessarily reached when starting from a partition different from the unique IS partition, which can be reached only if it is the initial partition. 
\end{proof}

Moreover, it is hard to decide on the existence of a sequence of IS deviations ending in an IS partition, even when starting from the \singleton, as well as checking convergence. 

\begin{restatable}{theorem}{existDHG}\label{thm:hardness-existencepath-DHG}
\existpb{IS}{DHG} is \np-hard even when starting from the \singleton, and \convergpb{IS}{DHG} is \conp-hard.
\end{restatable}

Note that the counterexample provided in the proof of Proposition~\ref{prop:noconv-dicho} exhibits a global cycle in the preferences of the agents: $\{1,2\}\succ_1 \{1,3\} \succ_3 \{2,3\} \succ_2 \{1,2\}$.
However, by considering dichotomous preferences with \emph{common ranking property}, that is, each agent has a threshold for acceptance in a given global order of coalitions, we obtain convergence thanks to the same potential function argument used by~\citet{CaKi19a}, for proving the existence of a core-stable partition in hedonic games with common ranking property. 

Also note that when assuming that if a coalition is approved by one agent, then it must be approved by all the members of the coalition (so-called \emph{symmetric dichotomous preferences}), we obtain a special case of preferences with common ranking property where all the approved coalitions are at the top of the global order.
Therefore, convergence is also guaranteed under symmetric dichotomous preferences.

\section{Conclusion}

We have investigated dynamics of deviations based on individual stability in hedonic games. The two main questions we considered were whether there exists \emph{some} sequence of deviations terminating in an IS partition, and whether \emph{all} sequences of deviations terminate in an IS partition, i.e., the dynamics converges. 
Many of our results are negative, that is, examples of cycles in dynamics or even non-existence of IS partitions under strong preference restrictions.
In particular, we have answered a number of open problems proposed in the literature leading to boundaries of dynamics. For all hedonic games under study, it turned out that the existence of cycles for IS deviations is sufficient to prove the hardness of recognizing instances for which there exists a finite sequence of deviations or whether all sequences of deviations are finite, i.e., the dynamics converges. 
On the other hand, we have identified natural conditions for convergence that are based on $i)$ initial conditions, that is, the starting partition, $ii)$ selection rules for the performed deviation, and $iii)$ preference restrictions such as a common scale for the agents (e.g., the common ranking property), single-peakedness, or symmetry. 

An overview of our results can be found in Table~\ref{tab:results}. While our hardness results show boundaries for both the possible and guaranteed convergence of dynamics, our positive results mostly focus on guaranteed convergence. In particular, we have made sophisticated use of potential functions to show the polynomial running time of the dynamics for anonymous hedonic games and hedonic diversity games with restrictions at the boundary of convergence. Our convergence result for HDGs features are rare case in the literature, presenting a highly non-trivial reduction to AHGs and therefore revealing a deep relationship of these two classes under natural single-peakedness. 

\newcommand{\positiveResult}{{\Large\checkmark\xspace}}
\newcommand{\negativeResult}{\makebox[\widthof{\positiveResult}]{\LARGE$\circ$\xspace}}

\newcommand{\existNP}{{\Large$\exists$\xspace}}
\newcommand{\forallcoNP}{{\Large$\forall$\xspace}}

\begin{table}[tb]
\caption{%
Convergence and hardness results for the dynamics of IS deviations in various classes of hedonic games. Symbol $\checkmark$ marks guaranteed convergence under the given preference restrictions and initial partition (if applicable)
while $\circ$ marks non-convergence, i.e., cycling dynamics. For all of our positive results except Theorem~\ref{thm:convAHG}, we can even show that the dynamics necessarily terminate after a polynomial number of steps. Symbol $\exists$ (or $\forall$) denotes that problem \existpb{IS}{HG} (or \convergpb{IS}{HG}) is \np-hard (or \conp-hard).}
\label{tab:results}
\centering\footnotesize
\resizebox{\textwidth}{!}{
\begin{tabular}{L{\widthof{HDGs}}L{\widthof{spa~common ranking property or symmetric  (Caskurlu and Kizilkaya
[15])}}L{\widthof{hi~ simple; singletons (Theorem~6.3)}}}
\toprule
Class & Guaranteed convergence & Hardness \\
\midrule
AHGs & \positiveResult ~ natural SP (single-peaked) (Theorem~\ref{thm:convAHG}) \newline
\positiveResult ~ neutral (\citet{Suks15a}) \newline
\negativeResult ~ strict \& general SP; singletons / grand coalition (Proposition~\ref{prop:cycleAHG-singleton}) &
\existNP ~ strict (Theorem~\ref{thm:hardness-existencepath-AHG}) \newline \forallcoNP ~ strict (Theorem~\ref{thm:hardness-existencepath-AHG}) \\
\midrule

HDGs & \positiveResult ~ strict \& natural SP; singletons; solitary homogeneity (Theorem~\ref{thm:converg-HDG-singletonstrictSP}) \newline
\negativeResult ~ \parbox[c]{\widthof{any three of: strict, natural SP, singletons, }}{any three of: strict, natural SP, singletons, and solitary homogeneity (Theorem~\ref{thm:cycleHDG})} &
\existNP ~  strict (Theorem~\ref{thm:hardness-existencepath-HDG}) \newline
\forallcoNP ~ strict (Theorem~\ref{thm:hardness-existencepath-HDG}) \\
\midrule

FHGs & \positiveResult ~ simple \& sym.; singletons (Proposition~\ref{prop:convergence-simplesymFHG}) \newline
\positiveResult ~ acyclic digraph (Theorem~\ref{thm:convergence-asymmetric+acyclicFHG}) \newline
\negativeResult ~ symmetric (Theorem~\ref{thm:symFHGnoIS}) &
\existNP ~ symmetric (Theorem~\ref{thm:hardness-symposFHG})\newline
\existNP ~ simple; singletons (Theorem~\ref{thm:hardness-existencepath-simpleFHG})\newline
\existNP ~ simple asym. (Theorem~\ref{thm:hardness-existencepath-asymFHG})\newline
\forallcoNP ~ symmetric (Theorem~\ref{thm:hardness-symposFHG}) \newline
\forallcoNP ~ simple asym. (Theorem~\ref{thm:hardness-existencepath-asymFHG}) \\
\midrule

DHGs & \positiveResult ~ grand coalition (\citet{BoEl20a}) \newline
\positiveResult ~ common ranking property or symmetric (\citet{CaKi19a})\newline
\negativeResult ~ singletons (Proposition~\ref{prop:noconv-dicho}) &
\existNP ~   singletons (Theorem~\ref{thm:hardness-existencepath-DHG}) \newline
\forallcoNP ~ general (Theorem~\ref{thm:hardness-existencepath-DHG}) \\
\bottomrule
\end{tabular}
} 
\end{table}

An important message of our results is that the consideration of dynamics can offer important novel insights regarding the reachability of stable states, even if the static picture drawn by asking for the mere existence of stability seems clear.
For instance, FHGs with non-negative utilities and HDGs always contain stable states, while it is hard to decide if we can reach a stable state from some initial partition, even under severe restrictions. On the other hand, dynamics always converge in naturally single-peaked AGHs, which is in accordance with the existence of stable partitions observed by \citet{BoJa02a}. In other words, the existence and the distributed attainability of stable states do not necessarily coincide.

While our results cover a broad range of hedonic games considered in the literature, there are still promising directions for further research.
First, even though our hardness results hold under strong restrictions, 
the complexity of these questions remains open for some interesting preference restrictions, some of which do not guarantee convergence. Following our work, the most intriguing case for guaranteed convergence are simple symmetric FHGs with arbitrary initial partitions. 
On the other hand, our computational hardness of simple FHGs only covers the \existpb{IS}{FHG} while we leave the complexity of \convergpb{IS}{FHG} open. We conjecture that this problem is hard as well.

Since our positive results mainly concern guaranteed convergence, there are also interesting open problems concerning the existence of a path to stability. In general, there is hope that less of the restrictions necessary for guaranteed convergence suffice for a path to stability. This especially concerns HDGs. Since solitary homogeneity is just a selection rule among possible deviations, Theorem~\ref{thm:converg-HDG-singletonstrictSP} implies that there always exists a path to stability from the \singleton in HDGs where agents have strict and single-peaked preferences. On the other hand, Proposition~\ref{prop:nopath-HDG}, our result about conditions under which cycling can necessarily occur leaves space for possibilities. Two intriguing questions are whether there always exists a path to stability in HDGs where agents' preferences are strict and single-peaked, or when the dynamics start with the \singleton.

Possible convergence of the dynamics is closely related to the investigation of specific selection rules for the performed deviations. With the exception of Theorem~\ref{thm:converg-HDG-singletonstrictSP}, we do not have to pose any assumptions on the performed deviations to obtain our results for guaranteed convergence. However, apart from possible convergence, selecting appropriate deviations may also lead to a fast termination of the dynamics in IS partitions, even in classes of hedonic games that allow for cyclic IS deviations. 
For instance, for simple symmetric FHGs, there is the possibility of convergence such that each agent deviates at most once, but the selection of the deviating agents in this approach requires to solve a maximum clique problem (cf. the discussion after Proposition~\ref{prop:convergence-simplesymFHG}).

An open problem concerning the convergence speed of dynamics is to bound the number of steps until convergence in AHGs under \emph{weak} and single-peaked preferences. Note that considering fastest instead of fast convergence is a problem that is likely to lead to further intractabilities. In this respect, \citet{BBK22a} provide first results by proving hardness of finding the shortest path to stability for several dynamics in additively separable hedonic games.

In principal, one can define dynamics also based on other stability concepts like Nash stability or contractual individual stability. For the latter, cycling is not possible, and therefore an analysis within the complexity class \pls as local search algorithms is a natural approach that measures the complexity of convergence. Such an analysis is also appropriate for dynamics guaranteed to converge based on a potential function argument as it was already done for Nash stability in additively separable hedonic games \citep{BoJa02a,GaSa19a}.

Finally, the final states reached in the dynamics we consider do not provide information beyond individual stability. One could therefore additionally aim to reach efficient outcomes, potentially measured amongst stable outcomes. The notion of Pareto optimality seems natural here because it gives also rise to a natural improvement dynamics. Indeed, \citet{BoJa02a} consider this problem for naturally single-peaked AHGs where their algorithm constructs an IS partition that is weakly Pareto-optimal.

\begin{acks}
This work was supported by the Deutsche Forschungsgemeinschaft under grant BR 2312/12-1. We would like to thank Abheek Ghosh for sharing insights on how to extend our convergence result for AHGs to weak preferences.
We also thank Leo Tappe and the anonymous reviewers from AAAI and TEAC for helpful discussions. We were impressed by the depth of the reviews from TEAC.
A preliminary version of this article appeared in the Proceedings of the 35th AAAI Conference on Artificial Intelligence (February, 2021).
Results from this article were presented at the 3rd Games, Agents, and Incentives Workshop (May, 2021) and the 15th Journ\'ees d'Intelligence Artificielle Fondamentale (July, 2021).
\end{acks}

\newpage

\appendix

\section*{APPENDIX: Omitted Proofs}

In the appendix, we provide the proofs omitted in the main part of the paper.

\section{Anonymous Hedonic Games}

\existAHG*

We prove the two hardness results by providing separate reductions for each problem in the next two lemmas.

\begin{lemma}\label{lem:hard-exist-AHG}
\existpb{IS}{AHG} is \np-hard even for strict preferences.
\end{lemma}

\begin{proof}
Let us perform a reduction from (3,B2)-SAT, a variant of the \textsc{Satisfiability} problem known to be \np-complete~\citep{BKS03a}. 
In an instance of (3,B2)-SAT, we are given a CNF propositional formula $\varphi$ where every clause $C_j$, for $1\leq  j \leq m$, contains exactly three literals and every variable $x_i$, for $1\leq i \leq p$, appears exactly twice as a positive literal and twice as a negative literal. 
From such an instance, we construct an instance of an anonymous game with initial partition as follows.

\newcommand{\clauseof}[1]{cl(#1)}

For each $\ell^{\text{th}}$ occurrence ($\ell\in\{1,2\}$) of a positive literal $x_i$ (or negative literal $\overline{x}_i$), we create a literal-agent $y_i^\ell$ (or $\overline{y}_i^\ell$).
All literal-agents are singletons in the initial partition $\spartition$.
Let us consider four integers $\alpha$, $\beta^+$, $\beta^-$ and $\gamma$ such that $(1)$ $q\cdot\alpha +x \neq r\cdot \beta^+ + y \neq s\cdot \beta^- + z \neq t \gamma + w$ for every $r,s,t\in [p]$, $q\in [m]$, $x,y,z\in\{0,1,2\}$ and $w\in [7]$ and, without loss of generality, $\alpha>\beta^+>\beta^->\gamma>1$. 
For instance, we can set the following values: $\alpha=m^5$, $\beta^+=m^4$, $\beta^-=m^3$ and $\gamma=m^2$ (condition $(1)$ is satisfied since in a (3,B2)-SAT instance, it holds that $m\geq 4$ and $p=3/4 m$).
For each clause $C_j$, we create $j\alpha$ dummy clause-agents who are all grouped within the same coalition $K_j$ in the initial partition $\spartition$.
For each literal $x_i$ (or $\overline{x}_i$), we create one variable-agent $z_i$ (or $\overline{z}_i$) and $i \beta^+-1$ (or $i\beta^- -1$) dummy variable agents who are all grouped within the same coalition $Z_i$ (or $\overline{Z_i}$) in the initial partition $\spartition$.
Finally, for each variable $x_i$, we create $i\gamma$ dummy agents who are all grouped within the same coalition $G_i^1$ in the initial partition $\spartition$, $i\gamma+3$ dummy agents who are all grouped within the same coalition $G_i^2$ in the initial partition $\spartition$ and $i\gamma+5$ dummy agents who are all grouped within the same coalition $G_i^3$ in the initial partition $\spartition$.
These dummy agents are used as a gadget for a cycle.
Although we have created many agents, the construction remains polynomial by considering reasonable values of $\alpha$, $\beta^+$, $\beta^-$ and $\gamma$, as previously described. 

The preferences of the agents over sizes of coalitions are given in Table~\ref{tab:pref-reduc-exist-AHG}.
By the design of the preferences of the members of the initial coalitions in $\spartition$, i.e., the members of the initial non-singleton coalitions accept at most two additional agents in their coalition and otherwise are happy with their coalition, and by condition~(1), all the sizes of non-singleton coalitions explicitly given in the preferences can be reached only in one way, which is the one described in the preferences, i.e., by the addition of at most two agents in a specific initial non-singleton coalition.
It follows that the preferences of the agents can be expressed in terms of preferences for joining, or that one or two agents join, a specific non-singleton coalition from the initial partition $\spartition$.

\begin{table}
\caption{Preferences of the agents in the reduced instance of Lemma~\ref{lem:hard-exist-AHG}, for every $1\leq i\leq p$, $1\leq j\leq m$, $\ell\in\{1,2\}$. Notation $\clauseof{x_i^\ell}$ (or $\clauseof{\overline{x}_i^\ell}$) stands for the index of the clause to which the $\ell^{\text{th}}$ occurrence of literal $x_i$ (or $\overline{x}_i$) belongs, the framed value is the size of the initial coalition in partition $\spartition$, and $[\dots]$ denotes an arbitrary order over the rest of the coalition sizes.}
\label{tab:pref-reduc-exist-AHG}
\begin{center}
\resizebox{\columnwidth}{!}{
\begin{tabular}{rl}
$z_i:$ & $|Z_i|+2 \succ |G_i^1|+2 \succ |G_i^2|+1 \succ |G_i^3|+2 \succ |G_i^3|+1 \succ |G_i^1|+1 \succ |Z_i|+1 \succ \boxed{|Z_i|} \succ [\dots]$ \\
$\overline{z}_i:$ & $|\overline{Z_i}|+2 \succ |G_i^3|+2 \succ |G_i^2|+1 \succ |G_i^1|+2 \succ |G_i^1|+1 \succ |G_i^3|+1 \succ |\overline{Z_i}|+1 \succ \boxed{|\overline{Z_i}|} \succ [\dots]$ \\
$y_i^\ell:$ & $|K_{\clauseof{x_i^\ell}}|+1 \succ |Z_i|+2 \succ |Z_i|+1 \succ \boxed{1} \succ [\dots]$ \\
$\overline{y}_i^\ell:$ & $|K_{\clauseof{\overline{x}_i^\ell}}|+1 \succ |\overline{Z_i}|+2 \succ |\overline{Z_i}|+1 \succ \boxed{1} \succ [\dots]$ \\
\midrule
$K_j:$ & $|K_j|+1 \succ \boxed{|K_j|} \succ [\dots]$ \\
$Z_i\setminus\{z_i\}:$ & $|Z_i|+2 \succ |Z_i|+1 \succ \boxed{|Z_i|} \succ |Z_i|-1 \succ [\dots]$ \\
$\overline{Z_i}\setminus\{\overline{z}_i\}:$ & $|\overline{Z_i}|+2 \succ |\overline{Z_i}|+1 \succ \boxed{|\overline{Z_i}|} \succ |\overline{Z_i}|-1 \succ [\dots]$ \\
$G_i^1:$ & $|G_i^1|+2 \succ |G_i^1|+1 \succ \boxed{|G_i^1|} \succ [\dots]$ \\
$G_i^2:$ & $|G_i^2|+1 \succ \boxed{|G_i^2|} \succ [\dots]$ \\
$G_i^3:$ & $|G_i^3|+2 \succ |G_i^3|+1 \succ \boxed{|G_i^3|} \succ [\dots]$ \\
\end{tabular}}
\end{center}
\end{table}

We claim that there exists a sequence of IS deviations which leads to an IS partition iff formula $\varphi$ is satisfiable.

Suppose first that there exists a truth assignment of the variables $\phi$ that satisfies all the clauses.
Let us denote by $\ell_j$ a chosen literal-agent associated with an occurrence of a literal true in $\phi$ which belongs to clause $C_j$.
Since all the clauses of $\varphi$ are satisfied by $\phi$, there exists such a literal-agent $\ell_j$ for each clause $C_j$.
For every clause $C_j$, let literal-agent $\ell_j$ join coalition $K_j$. 
These IS deviations make all the dummy clause-agents and the chosen literal-agents the most happy as possible, therefore none of them will deviate afterwards or let other agents enter their coalition.
Then, let all remaining literal-agents $y_i^\ell$ (or $\overline{y}_i^\ell$) deviate by joining coalition $Z_i$ (or $\overline{Z_i}$).
Since $\phi$ is a truth assignment of the variables, for each variable $x_i$, the two literal-agents corresponding to the literal of variable $x_i$ that is false in $\phi$ both deviate in this second round of deviations.
Therefore, there exists a coalition $Z_i$ or $\overline{Z_i}$ that is joined by two literal-agents and thus whose members all reach their most preferred size $|Z_i|+2$ or $|\overline{Z_i}|+2$.
It follows that no member of such a newly formed coalition would move afterwards or let other agents enter the coalition: all members of $Z_i$ or $\overline{Z_i}$ get their most preferred size while the two joining literal-agents get their second most preferred size and their most preferred size is not accessible anymore (their associated clause coalition has already been joined by another literal-agent). 
Consequently, for each variable $x_i$, at most one coalition between $Z_i$ and $\overline{Z_i}$ may not be joined by two literal-agents and, if there is one, it must be the coalition that corresponds to the literal of variable $x_i$ that is true in $\phi$.
In such a case, we let the associated variable-agent $z_i$ or $\overline{z}_i$ deviate for joining coalition $G_i^2$, and if one literal-agent previously joined the corresponding variable-coalition, she deviates to be alone.
Such a literal-agent then gets her fourth most preferred size while her most preferred ones are not accessible anymore (because the variable-agent has left the coalition and her associated clause coalition has already been joined by another literal-agent).
Moreover, such a variable-agent $z_i$ or $\overline{z}_i$, by joining coalition $G_i^2$, gets her third most preferred size while her most preferred ones are not accessible (no two additional agents want to enter the initial coalition $Z_i$ or $\overline{Z_i}$ and only one additional agent, herself, is present in the gadget associated with variable $x_i$).
Also, note that, by the design of the preferences, no dummy agent in the gadget has an incentive to move to another coalition.
All in all, no agent can then move in an IS deviation, and thus the reached partition is IS.

Suppose now that there exists no truth assignment of the variables that satisfies all the clauses.
That means that it is not possible that each clause coalition is joined by a literal-agent associated with this clause 
while two other literal-agents $y_i^1$ and $y_i^2$ join coalition $Z_i$ (or $\overline{y}_i^1$ and $\overline{y}_i^2$ join coalition $\overline{Z_i}$), for each variable $x_i$.
Moreover since, by design of the preferences, each literal-agent prefers to join clause coalitions than variable coalitions, it means that in a maximal sequence of IS deviations, all dummy clause-agents in each coalition $K_j$ will be completely satisfied with a coalition size equal to $|K_j|+1$ (if a clause coalition is not joined by a literal-agent, then a literal-agent associated with this clause has an incentive to join this coalition, no matter her current coalition).
It means that after a maximal sequence of IS deviations, each clause coalition is joined by a literal-agent associated with this clause.
Therefore, by the previous argument derived from the assumption that there exists no truth assignment of the variables that satisfies all the clauses, 
there must exist a variable $x_i$ such that at most one literal-agent joins coalition $Z_i$ and at most one literal-agent joins coalition $\overline{Z_i}$.
It follows that both variable-agents $z_i$ and $\overline{z_i}$ have an incentive to deviate to the gadget associated with variable $x_i$ (their respective most preferred coalition sizes $|Z_i|+2$ and $|\overline{Z_i}|+2$ can never be reached, while they prefer to join some coalitions in the gadget than staying in their current coalition of size $|Z_i|$ or $|Z_i|+1$ for $z_i$, and $|\overline{Z_i}|$ or $|\overline{Z_i}|+1$ for $\overline{z}_i$).
Within the gadget associated with variable $x_i$, variable-agents $z_i$ and $\overline{z_i}$ are the only agents who can deviate and we necessarily reach the cycle illustrated in Figure~\ref{fig:cycle-reduc-AHG}.

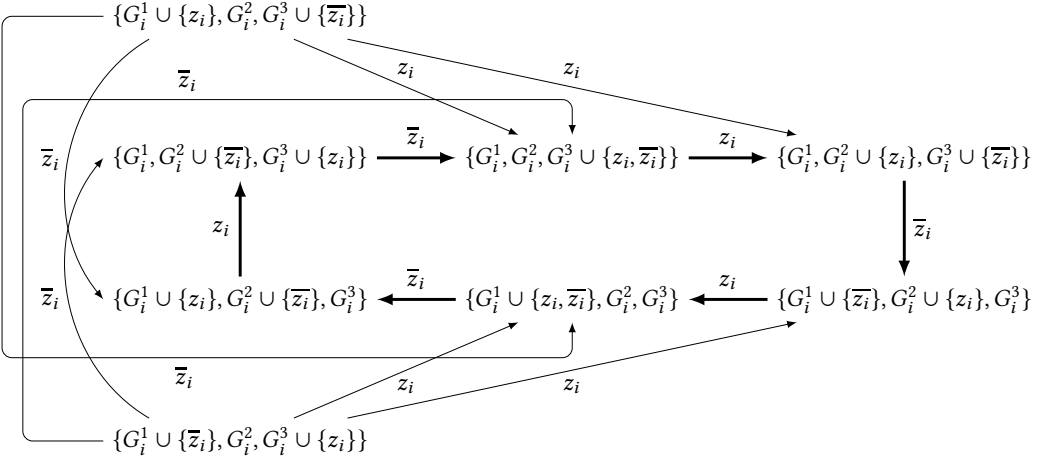
\begin{figure}
\begin{center}
\resizebox{\columnwidth}{!}{
\begin{tikzpicture}
\node (p1) at (-4.75,0) {$\{G_i^1,G_i^2\cup\{\overline{z_i}\},G_i^3\cup \{z_i\}\}$};
\node (p2) at (0,0) {$\{G_i^1,G_i^2,G_i^3\cup\{z_i,\overline{z_i}\}\}$};
\node (p3) at (4.75,0) {$\{G_i^1,G_i^2\cup\{z_i\},G_i^3\cup\{\overline{z_i}\}\}$};
\node (p4) at (4.75,-2) {$\{G_i^1\cup\{\overline{z_i}\},G_i^2\cup\{z_i\},G_i^3\}$};
\node (p5) at (0,-2) {$\{G_i^1\cup\{z_i,\overline{z_i}\},G_i^2,G_i^3\}$};
\node (p6) at (-4.75,-2) {$\{G_i^1\cup\{z_i\},G_i^2\cup\{\overline{z_i}\},G_i^3\}$};
\node (p7) at (-4.75,2) {$\{G_i^1\cup\{z_i\},G_i^2,G_i^3\cup\{\overline{z_i}\}\}$};
\node (p8) at (-4.75,-4) {$\{G_i^1\cup\{\overline{z}_i\},G_i^2,G_i^3\cup\{z_i\}\}$};
\draw[arrow,very thick] (p1) -- node[midway,above] {$\overline{z}_i$}  (p2);
\draw[arrow,very thick] (p2) -- node[midway,above] {$z_i$}  (p3);
\draw[arrow,very thick] (p3) -- node[midway,right] {$\overline{z}_i$}  (p4);
\draw[arrow,very thick] (p4) -- node[midway,above] {$z_i$}  (p5);
\draw[arrow,very thick] (p5) -- node[midway,above] {$\overline{z}_i$}  (p6);
\draw[arrow,very thick] (p6) -- node[midway,left] {$z_i$}  (p1);
\draw[arrow] (p7) -- node[midway,above] {$z_i$}  (p2);
\draw[arrow] (p7) -- node[midway,above] {$z_i$}  (p3);
\draw[arrow,rounded corners] (p7.west) -- ($(p7.west)+(-1.45,0)$) |- node[pos=0.66,below] {$\overline{z}_i$} ($(p5.south)+(0,-0.5)$) -- (p5.south);
\draw[arrow] (p7) to[bend right=50] node[midway,left] {$\overline{z}_i$}  (p6.west);
\draw[arrow] (p8) to[bend left=50] node[midway,left] {$\overline{z}_i$}  (p1.west);
\draw[arrow,rounded corners] (p8.west) -- ($(p8.west)+(-1.15,0)$) |- node[pos=0.65,above] {$\overline{z}_i$} ($(p2.north)+(0,0.5)$) -- (p2.north);
\draw[arrow] (p8) -- node[midway,below] {$z_i$}  (p4);
\draw[arrow] (p8) -- node[midway,below] {$z_i$}  (p5);
\end{tikzpicture}}
\end{center}
\caption{Necessary cycle of IS deviations within the gadget associated with variable $x_i$ in the reduced instance of Lemma~\ref{lem:hard-exist-AHG}.}
\label{fig:cycle-reduc-AHG}
\end{figure} 

It follows that no sequence of IS deviations can reach an IS partition.
\end{proof}

\begin{lemma}\label{lem:hard-converg-AHG}
\convergpb{IS}{AHG} is \conp-hard even for strict preferences.
\end{lemma}

\begin{proof}
For this purpose, we prove the \np-hardness of the complement problem, which asks whether there exists a cycle of IS deviations. 
Let us perform a reduction from (3,B2)-SAT~\citep{BKS03a}. 
In an instance of (3,B2)-SAT, we are given a CNF propositional formula $\varphi$ where every clause $C_j$, for $1\leq  j \leq m$, contains exactly three literals and every variable $x_i$, for $1\leq i \leq p$, appears exactly twice as a positive literal and twice as a negative literal. 
From such an instance, we construct an instance of an anonymous game with initial partition as follows.

For each $\ell^{\text{th}}$ occurrence ($\ell\in\{1,2\}$) of a positive literal $x_i$ (or negative literal $\overline{x}_i$), we create a literal-agent $y_i^\ell$ (or $\overline{y}_i^\ell$).
We create another agent $t$.
All these agents are singletons in the initial partition $\spartition$.
Let us consider five integers $\alpha$, $\beta^+_1$, $\beta^-_1$, $\beta^+_2$ and $\beta^-_2$ such that $(1)$ $q\cdot\alpha +x \neq r\cdot \beta^+_1 + y \neq s\cdot \beta^-_1 + z\neq t\cdot \beta^+_2 +v\neq u\cdot\beta^-_2 +w$ for every $r,s,t,u\in [p]$, $q\in [m]$ and $x,y,z,v,w\in\{0,1,2\}$ and, without loss of generality, $\alpha>\beta^+_1>\beta^-_1>\beta^+_2>\beta^-_2>1$. 
For instance, we can set the following values: $\alpha=m^5$, $\beta^+_1=m^4$, $\beta^-_1=m^3$, $\beta^+_2=m^2$, $\beta^-_2=m$ (condition $(1)$ is satisfied since in a (3,B2)-SAT instance, it holds that $m\geq 4$ and $p=3/4 m$).
For each clause $C_j$, we then create $j\cdot\alpha$ dummy clause agents grouped within the same coalition $K_j$ in the initial partition $\spartition$.
We also create $(m+1)\cdot\alpha$ dummy agents grouped within the same coalition $K_{m+1}$ in initial partition $\spartition$.
Finally, for each literal $x_i$ (or $\overline{x}_i$) and each $\ell\in\{1,2\}$, we create $i \cdot \beta^+_\ell$ (or $i \cdot \beta^-_\ell$) dummy variable agents grouped within the same coalition $Y_i^\ell$ (or $\overline{Y_i^\ell}$) in the initial partition $\spartition$.
Although we have created many agents, the construction remains polynomial by considering reasonable values of $\alpha$, $\beta^+_1$, $\beta^-_1$, $\beta^+_2$ and $\beta^-_2$, as previously described. 

\newcommand{\clauseof}[1]{cl(#1)}

The preferences of the agents over sizes of coalitions are given in Table~\ref{tab:pref-reduc-converg-AHG}.
By the design of the preferences of the members of the initial coalitions in $\spartition$, i.e., the members of the initial non-singleton coalitions accept at most two additional agents in their coalition and otherwise are happy with their coalition, and by condition~(1), all the sizes of non-singleton coalitions explicitly given in the preferences can be reached only in one way, which is the one described in the preferences, i.e., by the addition of at most two agents in a specific initial non-singleton coalition.
It follows that the preferences of the agents can be expressed in terms of preferences for joining, or that one or two agents join, a specific non-singleton coalition from the initial partition $\spartition$.

\begin{table}
\caption{Preferences of the agents in the reduced instance of Lemma~\ref{lem:hard-converg-AHG}, for every $1\leq i\leq p$, $1\leq i' < p$, $1\leq j\leq m+1$, $\ell\in\{1,2\}$. Notation $\clauseof{x_i^\ell}$ (or $\clauseof{\overline{x}_i^\ell}$) stands for the index of the clause to which the $\ell^{\text{th}}$ occurrence of literal $x_i$ (or $\overline{x}_i$) belongs, the framed value is the size of the initial coalition in partition $\spartition$, and $[\dots]$ denotes an arbitrary order over the rest of the coalition sizes.}
\label{tab:pref-reduc-converg-AHG}
\begin{center}
\resizebox{\columnwidth}{!}{
\begin{tabular}{rl}
$y_i^1:$ & $|K_{\clauseof{x_i^1}}|+2 \succ |K_{\clauseof{x_i^1}+1}|+2 \succ |K_{\clauseof{x_i^1}+1}|+1 \succ |K_{\clauseof{x_i^1}}|+1 \succ |Y_i^1|+2 \succ |Y_i^2|+2 \succ |Y_i^2|+1 \succ |Y_i^1|+1 \succ  \boxed{1} \succ [\dots]$ \\
$y_{i'}^2:$ & $|K_{\clauseof{x_{i'}^2}}|+2 \succ |K_{\clauseof{x_{i'}^2}+1}|+2 \succ |K_{\clauseof{x_{i'}^2}+1}|+1 \succ |K_{\clauseof{x_{i'}^2}}|+1 \succ |Y_{i'}^2|+2 \succ |Y_{i'+1}^1|+2 \succ |Y_{i'+1}^1|+1 \succ |\overline{Y_{i'+1}^1}|+2 \succ $ \\ & \hfill $|\overline{Y_{i'+1}^1}|+1 \succ |Y_{i'}^2|+1 \succ  \boxed{1} \succ [\dots]$ \\
$y_p^2:$ & $|K_{\clauseof{x_p^2}}|+2 \succ |K_{\clauseof{x_p^2}+1}|+2 \succ |K_{\clauseof{x_p^2}+1}|+1 \succ |K_{\clauseof{x_p^2}}|+1 \succ |Y_p^2|+2 \succ |K_1|+2 \succ |K_1|+1 \succ |Y_p^2|+1 \succ  \boxed{1} \succ [\dots]$ \\
$\overline{y}_i^1:$ & $|K_{\clauseof{\overline{x}_i^1}}|+2 \succ |K_{\clauseof{\overline{x}_i^1}+1}|+2 \succ  |K_{\clauseof{\overline{x}_i^1}+1}|+1 \succ |K_{\clauseof{\overline{x}_i^1}}|+1 \succ |\overline{Y_i^1}|+2 \succ |\overline{Y_i^2}|+2 \succ |\overline{Y_i^2}|+1 \succ |\overline{Y_i^1}|+1 \succ  \boxed{1} \succ [\dots]$ \\
$\overline{y}_{i'}^2:$ & $|K_{\clauseof{\overline{x}_{i'}^2}}|+2 \succ |K_{\clauseof{\overline{x}_{i'}^2}+1}|+2 \succ |K_{\clauseof{\overline{x}_{i'}^2}+1}|+1 \succ |K_{\clauseof{\overline{x}_{i'}^2}}|+1 \succ |\overline{Y_{i'}^2}|+2 \succ |Y_{i'+1}^1|+2 \succ |Y_{i'+1}^1|+1 \succ |\overline{Y_{i'+1}^1}|+2 \succ$ \\ & \hfill $|\overline{Y_{i'+1}^1}|+1 \succ |\overline{Y_{i'}^2}|+1 \succ  \boxed{1} \succ [\dots]$ \\
$\overline{y}_p^2:$ & $|K_{\clauseof{\overline{x}_p^2}}|+2 \succ |K_{\clauseof{\overline{x}_p^2}+1}|+2 \succ |K_{\clauseof{\overline{x}_p^2}+1}|+1 \succ |K_{\clauseof{\overline{x}_p^2}}|+1 \succ |\overline{Y_p^2}|+2 \succ |K_1|+2 \succ |K_1|+1 \succ |\overline{Y_p^2}|+1 \succ  \boxed{1} \succ [\dots]$ \\
$t:$ & $|K_{m+1}|+2 \succ |Y_{1}^1|+2 \succ |Y_{1}^1|+1 \succ |\overline{Y_{1}^1}|+2 \succ |\overline{Y_{1}^1}|+1 \succ |K_{m+1}|+1 \succ \boxed{1} \succ [\dots]$ \\
\midrule
$K_j:$ & $|K_j|+2 \succ |K_j|+1 \succ \boxed{|K_j|} \succ 1 \succ [\dots]$ \\
$Y_i^\ell:$ & $|Y_i^\ell|+2 \succ |Y_i^\ell|+1 \succ \boxed{|Y_i^\ell|} \succ 1 \succ [\dots]$ \\
$\overline{Y_i^\ell}:$ & $|\overline{Y_i^\ell}|+2 \succ |\overline{Y_i^\ell}|+1 \succ \boxed{|\overline{Y_i^\ell}|} \succ 1 \succ [\dots]$ \\
\end{tabular}}
\end{center}
\end{table}

We claim that there exists a cycle of IS deviations iff formula $\varphi$ is satisfiable\footnote{Note that the \singleton is nevertheless always IS.}.

The global idea of the proof is that a cycle of IS deviations necessarily involves, as deviating agents, agent $t$ and $(i)$ one literal-agent for each clause (the associated literal occurrence of the literal-agent belongs to the clause), as well as $(ii)$ the two literal-agents associated with a same literal for each variable.
All these agents must be distinct, implying the existence of a truth assignment of the variables that satisfies all the clauses.
The cycle of IS deviations is such that the literal-agents corresponding to case $(i)$ alternate between joining the coalition of dummy clause agents associated with their clause and the one associated with the next clause (w.r.t. the indices of clauses),
and the literal-agents corresponding to case $(ii)$ alternate between joining the coalition of dummy variable agents associated with their literal occurrence and the one associated with the other occurrence of the same literal, if the literal-agent corresponds to the first occurrence of the literal, or the first occurrence of the chosen literal of the next variable (w.r.t. the indices of variables), if the literal-agent corresponds to the second occurrence of the literal.
The example of such a cycle can be found in Figure~\ref{fig:cycle-reduc-conv-AHG}. 

Suppose first that formula $\varphi$ is satisfiable by a truth assignment of the variables denoted by $\phi$.
Let us denote by $\ell_j$ a chosen literal-agent associated with an occurrence of a literal true in $\phi$ which belongs to clause $C_j$.
Since all the clauses of $\varphi$ are satisfied by $\phi$, there exists such a literal-agent $\ell_j$ for each clause $C_j$.
Further, let us denote by $z_i^1$ and $z_i^2$ the literal-agents associated with the two occurrences of the literal of variable $x_i$ which is false in $\phi$.
In the same vein, let us denote by $Z_i^1$ and $Z_i^2$ the coalitions of dummy variable agents associated with $z_i^1$ and $z_i^2$, respectively. 
Since $\phi$ is a truth assignment of the variables, $z_i^1$, $z_i^2$, $Z_i^1$ and $Z_i^2$ all correspond to the same literal (either $x_i$ or $\overline{x}_i$) and it holds that $\bigcup_{1\leq j\leq m}\ell_j \cap \bigcup_{1\leq i\leq p} \{z_i^1,z_i^2\}=\emptyset$.
We will construct a cycle of IS deviations involving, as deviating agents, the literal-agents $\ell_j$, for every $1\leq j\leq m$, the literal-agents $z_i^1$ and $z_i^2$, for every $1\leq i\leq p$, and agent $t$.
Since $m$ is even in a (3,B2)-SAT (recall that $m=4/3p$), there is an odd number of deviating agents in total.
The main steps of the cycle are illustrated in Figure~\ref{fig:cycle-reduc-conv-AHG}.

\begin{figure}
\begin{center}
\resizebox{\columnwidth}{!}{
\begin{tikzpicture}
\node (K1) at (0,-4) {$K_1\cup\{y_1^1,y_3^2\}$};
\node (K2) at ($(K1)+(3,0)$) {$K_2$};
\node (K3) at ($(K2)+(3,0)$) {$K_3\cup\{y_2^2\}$};
\node (K4) at ($(K3)+(3,0)$) {$K_4\cup\{y_1^2,\overline{y}_3^2\}$};
\node (K5) at ($(K4)+(3,0)$) {$K_5$};
\node (Y11) at ($(K1)+(0,-0.75)$) {$Y_1^1$};
\node (Y12) at ($(Y11)+(3,0)$) {$Y_1^2$};
\node (nY11) at ($(Y12)+(3,0)$) {$\overline{Y_1^1}\cup\{t,\overline{y}_1^1\}$};
\node (nY12) at ($(nY11)+(3,0)$) {$\overline{Y_1^2}$};
\node (Y21) at ($(Y11)+(0,-0.75)$) {$Y_2^1$};
\node (Y22) at ($(Y21)+(3,0)$) {$Y_2^2$};
\node (nY21) at ($(Y22)+(3,0)$) {$\overline{Y_2^1}\cup\{\overline{y}_1^2,\overline{y}_2^1\}$};
\node (nY22) at ($(nY21)+(3,0)$) {$\overline{Y_2^2}$};
\node (Y31) at ($(Y21)+(0,-0.75)$) {$Y_3^1\cup\{\overline{y}_2^2,y_3^1\}$};
\node (Y32) at ($(Y31)+(3,0)$) {$Y_3^2$};
\node (nY31) at ($(Y32)+(3,0)$) {$\overline{Y_3^1}$};
\node (nY32) at ($(nY31)+(3,0)$) {$\overline{Y_3^2}$};
\node (y11) at ($(K5)+(-1,-0.325)$) {};
\node (y21) at ($(y11)+(0,-0.75)$) {$y_2^1$};
\node (ny31) at ($(y21)+(2,0)$) {$\overline{y}_3^1$};
\node (pi) at ($(K1)+(-2,-1)$) {$\partition$:};
\draw[arrow,line width=0.001pt] ($(K1)+(0.65,-0.25)$) -- (Y32);
\draw[arrow,line width=0.001pt] ($(Y31)+(0.65,-0.25)$) to[bend right=10] (Y32);
\draw[arrow,line width=0.001pt] ($(Y31)+(0.25,-0.25)$) to[bend right=20] (nY22);
\draw[arrow,line width=0.001pt] ($(nY11)+(0.65,0.15)$) to[bend left=15] (nY12);
\draw[arrow,line width=0.001pt] ($(nY21)+(0.25,0.15)$) -- (nY12);
\draw[arrow,line width=0.001pt] ($(nY21)+(0.65,0.15)$) to[bend left=15] (nY22);
\draw[arrow,line width=0.001pt] ($(K4)+(0.15,0.15)$) to[bend right=20] (K3);
\draw[arrow,line width=0.001pt] ($(K4)+(0.65,0.25)$) to[bend left=20] (K5);
\draw[arrow,rounded corners,line width=0.001pt] ($(nY11)+(0.25,0.1)$) |- ($(K5)+(0,-0.25)$) -- ($(K5)+(0,-0.1)$);
\draw[arrow,line width=3pt] ($(pi)+(8,-1.75)$) -- ($(pi)+(8,-2.75)$);
\node (K1) at (0,-8.25) {$K_1\cup\{y_1^1\}$};
\node (K2) at ($(K1)+(3,0)$) {$K_2$};
\node (K3) at ($(K2)+(3,0)$) {$K_3\cup\{y_2^2,y_1^2\}$};
\node (K4) at ($(K3)+(3,0)$) {$K_4$};
\node (K5) at ($(K4)+(3,0)$) {$K_5\cup\{\overline{y}_3^2,t\}$};
\node (Y11) at ($(K1)+(0,-0.75)$) {$Y_1^1$};
\node (Y12) at ($(Y11)+(3,0)$) {$Y_1^2$};
\node (nY11) at ($(Y12)+(3,0)$) {$\overline{Y_1^1}$};
\node (nY12) at ($(nY11)+(3,0)$) {$\overline{Y_1^2}\cup\{\overline{y}_1^1,\overline{y}_1^2\}$};
\node (Y21) at ($(Y11)+(0,-0.75)$) {$Y_2^1$};
\node (Y22) at ($(Y21)+(3,0)$) {$Y_2^2$};
\node (nY21) at ($(Y22)+(3,0)$) {$\overline{Y_2^1}$};
\node (nY22) at ($(nY21)+(3,0)$) {$\overline{Y_2^2}\cup\{\overline{y}_2^1,\overline{y}_2^2\}$};
\node (Y31) at ($(Y21)+(0,-0.75)$) {$Y_3^1$};
\node (Y32) at ($(Y31)+(3,0)$) {$Y_3^2\cup\{y_3^1,y_3^2\}$};
\node (nY31) at ($(Y32)+(3,0)$) {$\overline{Y_3^1}$};
\node (nY32) at ($(nY31)+(3,0)$) {$\overline{Y_3^2}$};
\node (y11) at ($(K5)+(-1,-0.325)$) {};
\node (y21) at ($(y11)+(0,-0.5)$) {$y_2^1$};
\node (ny31) at ($(y21)+(2,0)$) {$\overline{y}_3^1$};
\node (pi) at ($(K1)+(-2,-1)$) {$\partition_1$:};
\draw[arrow,line width=0.001pt] ($(K1)+(0.35,0.2)$) to[bend left] (K2);
\draw[arrow,line width=0.001pt] ($(K3)+(0.75,0.2)$) to[bend left] (K4);
\draw[arrow,line width=0.001pt] ($(K3)+(0.2,0.2)$) to[bend right] (K2);
\draw[arrow,line width=0.001pt] ($(K5)+(0.25,0.2)$) to[bend right] (K4);
\draw[arrow,line width=0.001pt] ($(Y32)+(0.15,-0.25)$) to[bend left=15] (Y31);
\draw[arrow,line width=0.001pt] ($(nY12)+(0.15,0.2)$) to[bend right=15] (nY11);
\draw[arrow,line width=0.001pt] ($(nY22)+(0.15,0.2)$) to[bend right=10] (nY21);
\draw[arrow,rounded corners,line width=0.001pt] ($(Y32)+(0.65,-0.25)$) |- ($(Y31)+(-0.5,-0.5)$) |- ($(K1)+(-0.45,-0.25)$);
\draw[arrow,line width=0.001pt] ($(nY12)+(0.6,-0.25)$) -- ($(nY21)+(0.25,0.2)$);
\draw[arrow,rounded corners,line width=0.001pt] ($(nY22)+(0.7,-0.2)$) -- ($(nY32)+(0.7,-0.4)$) -| (Y31);
\draw[arrow,line width=3pt] ($(pi)+(8,-1.75)$) -- ($(pi)+(8,-2.75)$);
\node (K1) at (0,-12.5) {$K_1\cup\{y_3^2\}$};
\node (K2) at ($(K1)+(3,0)$) {$K_2\cup\{y_1^1,y_2^2\}$};
\node (K3) at ($(K2)+(3,0)$) {$K_3$};
\node (K4) at ($(K3)+(3,0)$) {$K_4\cup\{y_1^2,\overline{y}_3^2\}$};
\node (K5) at ($(K4)+(3,0)$) {$K_5$};
\node (Y11) at ($(K1)+(0,-0.75)$) {$Y_1^1$};
\node (Y12) at ($(Y11)+(3,0)$) {$Y_1^2$};
\node (nY11) at ($(Y12)+(3,0)$) {$\overline{Y_1^1}\cup\{t,\overline{y}_1^1\}$};
\node (nY12) at ($(nY11)+(3,0)$) {$\overline{Y_1^2}$};
\node (Y21) at ($(Y11)+(0,-0.75)$) {$Y_2^1$};
\node (Y22) at ($(Y21)+(3,0)$) {$Y_2^2$};
\node (nY21) at ($(Y22)+(3,0)$) {$\overline{Y_2^1}\cup\{\overline{y}_1^2,\overline{y}_2^1\}$};
\node (nY22) at ($(nY21)+(3,0)$) {$\overline{Y_2^2}$};
\node (Y31) at ($(Y21)+(0,-0.75)$) {$Y_3^1\cup\{\overline{y}_2^2,y_3^1\}$};
\node (Y32) at ($(Y31)+(3,0)$) {$Y_3^2$};
\node (nY31) at ($(Y32)+(3,0)$) {$\overline{Y_3^1}$};
\node (nY32) at ($(nY31)+(3,0)$) {$\overline{Y_3^2}$};
\node (y11) at ($(K5)+(-1,-0.325)$) {};
\node (y21) at ($(y11)+(0,-0.5)$) {$y_2^1$};
\node (ny31) at ($(y21)+(2,0)$) {$\overline{y}_3^1$};
\node (pi) at ($(K1)+(-2,-1)$) {$\partition_2$:};
\draw[arrow,line width=0.001pt] ($(K2)+(0.1,0.2)$) to[bend right] (K1);
\draw[arrow,line width=0.001pt] ($(K2)+(0.6,0.2)$) to[bend left] (K3);
\draw[arrow,line width=3pt,rounded corners] ($(pi)+(8,-1.75)$) -- ($(pi)+(-0.5,-1.75)$) |- ($(pi)+(1,9)$);
\end{tikzpicture}}
\end{center}
\caption{Example for the cycle of IS deviations described in the reduced instance of Lemma~\ref{lem:hard-converg-AHG}. The illustration is for a yes-instance of (3,B2)-SAT with four clauses and three variables where $\varphi\equiv (x_1 \vee x_2 \vee x_3) \wedge (\overline{x}_1 \vee x_2 \vee \overline{x}_3) \wedge (x_1 \vee \overline{x}_2 \vee x_3) \wedge (\overline{x}_1 \vee \overline{x}_2 \vee \overline{x}_3)$. We consider the truth assignment of the variables $\phi$ where literals $x_1$, $x_2$ and $\overline{x}_3$ are true. Assignment $\phi$ satisfies formula $\varphi$ with, e.g., literal $x_1$ which makes clauses $C_1$ and $C_3$ true, literal $x_2$ which makes clause $C_2$ true, and literal $\overline{x}_3$ which makes clause $C_4$ true.
By using the notations of the proof, we thus have $\ell_1=y_1^1$, $\ell_2=y_2^2$, $\ell_3=y_1^2$, $\ell_4=\overline{y}_3^2$ and, for every $r\in\{1,2\}$, $Z_1^r=\overline{Y_1^r}$, $Z_2^r=\overline{Y_2^r}$, and $Z_3^r=Y_3^r$.}
\label{fig:cycle-reduc-conv-AHG}
\end{figure}
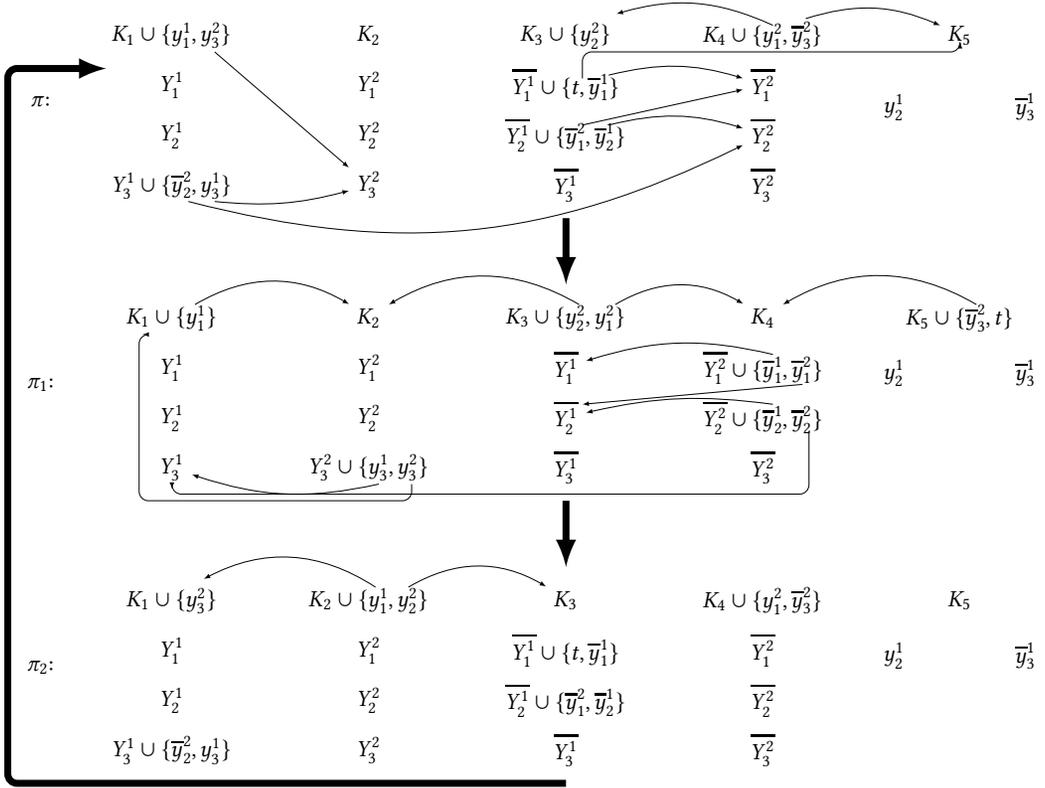

First of all, let agent $z_p^2$ and then agent $\ell_1$ join coalition $K_1$.
For each $1<i\leq p$, let agent $z_{i-1}^2$ and then agent $z_i^1$ join coalition $Z_i^1$.
Let agent $t$ and then agent $z_1^1$ join coalition $Z_1^1$.
For each even $j$ such that $3<j\leq m$, let agent $\ell_{j-1}$ and then agent $\ell_j$ join coalition $K_j$.
Finally, let agent $\ell_2$ join coalition $K_3$. 
The reached partition is 
$\partition:=
\{K_1\cup\{\ell_1,z_p^2\},K_2,K_3\cup\{\ell_2\},K_4\cup\{\ell_3,\ell_4\},K_5,K_6\cup\{\ell_5,\ell_6\},K_7,\dots,K_m\cup\{\ell_{m-1},\allowbreak\ell_m\},K_{m+1},Z_1^1\cup\{t,z_1^1\},Z_1^2,\allowbreak Z_2^1\cup\{z_1^2,z_2^1\},Z_2^2,\dots,Z_p^1\cup\{z_{p-1}^2,\allowbreak z_p^1\},Z_p^2,\overline{Z_1^1},\overline{Z_1^2},\dots,\overline{Z_p^1},\overline{Z_p^2}\}$, where coalition $\overline{Z_i^\ell}$ refers to $\overline{Y_i^\ell}$ if $Z_i^\ell=Y_i^\ell$ and to $Y_i^\ell$ if $Z_i^\ell=\overline{Y_i^\ell}$. 
Partition $\partition$ is the first step of the cycle (see Figure~\ref{fig:cycle-reduc-conv-AHG}).

From partition $\partition$, let literal-agent $\ell_3$ deviate from current coalition $K_4\cup\{\ell_3,\ell_4\}$ to join coalition $K_3\cup\{\ell_2\}$.
This deviation makes literal-agent $\ell_4$ worse off, who then deviates to join coalition $K_5$.
Then, the same deviations occur for literal-agents $\ell_5$ and $\ell_6$, and so on. 
More generally, for every odd $j$ such that $3\leq j\leq m$ by increasing order of indices, literal-agent $\ell_j$ leaves coalition $K_{j+1}\cup\{\ell_j,\ell_{j+1}\}$ to join coalition $K_j\cup\{\ell_{j-1}\}$ and then literal-agent $\ell_{j+1}$, who is worse off by this deviation, deviates to join coalition $K_{j+2}$.
After that, agent $t$ deviates from coalition $Z_1^1\cup\{t,z_1^1\}$ to join coalition $K_{m+1}\cup\{\ell_m\}$, which makes literal-agent $z_1^1$ worse off.
Then, for each $1\leq i\leq p$ by increasing order of indices, let literal-agent $z_i^1$ deviate from coalition $Z_i^1\cup\{z_i^1\}$ to join coalition $Z_i^2$ and then, if $i<p$, literal-agent $z_i^2$ deviates from coalition $Z_{i+1}^1\cup\{z_i^2,z_{i+1}^1\}$ 
to join coalition $Z_i^2\cup\{z_i^1\}$, which makes literal-agent $z_{i+1}^1$ 
worse off.
Afterwards, literal-agent $z_p^2$ deviates from coalition $K_1\cup\{z_i^2,\ell_1\}$ to join coalition $Z_p^2\cup\{z_p^1\}$, which makes literal-agent $\ell_1$ worse off.
We thus reach partition $\partition_1:=\{K_1\cup\{\ell_1\},K_2,K_3\cup\{\ell_2,\ell_3\},K_4,K_5\cup\{\ell_4,\ell_5\},K_6,K_7\cup\{\ell_6,\ell_7\},\dots,K_{m-1}\cup\{\ell_{m-2},\allowbreak\ell_{m-1}\},K_m,K_{m+1}\cup\{\ell_m,t\},Z_1^1,Z_1^2\cup\{z_1^1,z_1^2\},\allowbreak Z_2^1,Z_2^2\cup\{z_2^1,z_2^2\},\dots,Z_p^1,Z_p^2\cup\{z_p^1,\allowbreak z_p^2\},\overline{Z_1^1},\overline{Z_1^2},\dots,\overline{Z_p^1},\overline{Z_p^2}\}$ (see Figure~\ref{fig:cycle-reduc-conv-AHG}).

Then, from partition $\partition_1$, for every odd $j$ such that $1\leq j < m$ by increasing order of indices, let literal-agent $\ell_j$ deviate from coalition $K_j\cup\{\ell_j\}$ to join coalition $K_{j+1}$ and then, if $j<m-1$, literal-agent $\ell_{j+1}$ deviates from $K_{j+2}\cup\{\ell_{j+1},\ell_{j+2}\}$ 
to join coalition $K_{j+1}\cup\{\ell_j\}$, which makes literal-agent $\ell_{j+2}$ 
worse off.
Afterwards, literal-agent $\ell_m$ deviates from $K_{m+1}\cup\{\ell_m,t\}$ to join coalition $K_m\cup\{\ell_{m-1}\}$, which makes literal-agent $t$ worse off.
Let agent $t$ then deviate from coalition $K_{m+1}\cup\{t\}$ to join coalition $Z_1^1$ and literal-agent $z_1^1$ deviate from coalition $Z_1^2\cup\{z_1^1,z_1^2\}$ to join coalition $Z_1^1\cup\{t\}$, which makes literal-agent $z_1^2$ worse off.
Then, for each $1\leq i < p$  by increasing order of indices, let literal-agent $z_i^2$ deviate from coalition $Z_i^2\cup\{z_i^2\}$ to join coalition $Z_{i+1}^1$, and literal-agent $z_{i+1}^1$ deviate from coalition $Z_{i+1}^2\cup\{z_{i+1}^1,z_{i+1}^2\}$ to join coalition $Z_{i+1}^1\cup\{z_i^2\}$, which makes literal-agent $z_{i+1}^2$ worse off.
And then, let literal-agent $z_p^2$ deviate from coalition $Z_p^2\cup\{z_p^2\}$ to join coalition $K_1$, leading to partition $\partition_2$ (see Figure~\ref{fig:cycle-reduc-conv-AHG}). 
Afterwards, literal-agent $\ell_1$ deviates from coalition $K_2\cup\{\ell_1,\ell_2\}$ to join coalition $K_1\cup\{z_p^2\}$, which makes literal-agent $\ell_2$ worse off.
Finally, let literal-agent $\ell_2$ deviate from coalition $K_2\cup\{\ell_2\}$ to join coalition $K_3$, and we have finally reached partition $\partition$, leading to a cycle.

Suppose now that there exists a cycle of IS deviations.
Observe first that no dummy agent can deviate.
Indeed, the only coalition sizes that are preferred by a dummy agent to the size of her initial coalition are the size of the current coalition plus one and the size of the current coalition plus two. 
These sizes cannot be reached by joining other coalitions by condition~$(1)$, and the fact that the other coalitions do not want to integrate more than two additional agents in their coalition.
Therefore, the only possible deviations are when the dummy agents let at most two agents join their coalition.
It follows, by construction of the preferences, that no agent can belong to a coalition whose size is ranked after size 1 in her preferences, i.e., we do not have to care about the preferences within $[\dots]$ in the preference ranking of the agents.
Indeed, a literal-agent or agent $t$ (for the sake of simplicity we also talk about agent $t$ when referring to literal-agents since the behavior is similar), whose initial coalition is of size one, can join some coalitions of dummy agents and sometimes accepts one additional literal-agent in such coalitions. 
The worst thing that can happen to these deviating literal-agents is that the other literal-agent, who has joined the same coalition of dummy agents as her, leaves the coalition. 
However, in such a case, by construction of the preferences, both literal-agents are still in a coalition whose size is ranked before one in their preferences.

Observe that literal-agents can only join coalitions of dummy agents with at most one other literal-agent in this coalition and that no literal-agent can join another literal-agent outside a coalition of dummy agents, because size two is not ranked before size one in the preferences of the agents (recall that a coalition of dummy agents cannot be of size smaller than 2). 
Moreover, since the literal-agents can never be in a coalition of size less preferred than one, once a literal-agent deviates from her initial coalition where she is alone, she has no incentive to come back to the coalition where she is alone.
Hence, the deviations in the cycle must be performed by literal-agents who join different coalitions of dummy agents.

Because only literal-agents can deviate and all non-singleton coalitions of any reached partition must of the form of a coalition of dummy agents (from the initial partition) potentially joined by one or two literal-agents, each deviating literal-agent $ag$ in the cycle must be left at some step in order to come back to a less preferred coalition.
Since literal-agent $ag$ can be left only by one other literal-agent, it follows that the current coalition of agent $ag$ was of size $|K|+2$ for a given coalition $K$ of dummy agents and becomes of size $|K|+1$. 
To be able to come back to a previous less preferred coalition, agent $i$ must prefer $|K|+2$ over $|K|+1$. 
Moreover, since literal-agent $i$ is a deviating agent, there must be intermediate sizes in the preference ranking of agent $ag$ between $|K|+2$ and $|K|+1$.
Say that literal-agent $ag$ corresponds to a literal-agent $z_i^\ell$ where $z_i$ refers either to literal $y_i$ or to literal $\overline{y}_i$, and $Z_i^\ell$ refers to the associated clause of dummy variable agents, i.e., $Z_i^\ell=Y_i^\ell$ if $z_i^\ell=y_i^\ell$, and $Z_i^\ell=\overline{Y_i^\ell}$ if $z_i^\ell=\overline{y}_i^\ell$.
Moreover, denote by $j$ the index of the clause to which the literal associated with $z_i^\ell$ belongs, i.e., $j=\clauseof{x_i^\ell}$ if $z_i^\ell=y_i^\ell$ and $j=\clauseof{\overline{x}_i^\ell}$ if $z_i^\ell=\overline{y}_i^\ell$.
Then, 
by construction of the preferences, coalition $K$ of dummy agents can only be: $(a)$ 
$K_j$, i.e., the coalition associated with the clause to which the occurrence of the literal of the literal-agent belongs, or $(b)$ 
$Z_i^\ell$, i.e., the coalition associated with the occurrence of the literal of the literal-agent.
We detail below the possible cases for coalition $K$.

\begin{enumerate}[label=(\alph*)]
\item If $K$ is the coalition of dummy clause agents $K_j$, 
then the other literal-agent who leaves the coalition cannot be associated with an occurrence of a literal belonging to the corresponding clause. 
Indeed, otherwise, by construction of the preferences, the size of this coalition would be the most preferred one for this leaving literal-agent, contradicting her IS deviation from this coalition.
Therefore, according to the preferences of the literal-agents, the only possibility is that this other literal-agent who leaves coalition $K_j$ 
is associated with an occurrence of a literal belonging to $C_{j-1}$ if $j>1$ or is agent $z_p^2\in\{y_p^2,\overline{y}_p^2\}$ 
if $j=1$.
Due to the preferences of the literal-agents, if a literal-agent leaves such a coalition, it is necessarily for joining the coalition of dummy agents $K_{j-1}$ (which has an additional literal-agent) if $j>1$ or 
$Z_p^2$ (with an additional literal-agent) if $j=1$. 
In the latter case ($j=1$), this is the only possibility even if the associated deviating agent 
$z_p^2$ prefers several other coalition sizes over $|K_1|+2$, because these other choices would prevent her to come back to size $|K_1|+2$: the worst thing that can occur after some steps if $z_p^2$ chooses preferred coalitions other than $Z_p^2$ is that she would be in a coalition of size $|K_{\clauseof{x_p^2}}|+1$ or $|K_{\clauseof{x_p^2}+1}|+1$ if $z_p^2=y_p^2$, or $|K_{\clauseof{\overline{x}_p^2}}|+1$ or $|K_{\clauseof{\overline{x}_p^2}+1}|+1$) if $z_p^2=\overline{y}_p^2$, and both sizes are preferred to $|K_1|+2$, which contradicts the cycle.
\item \begin{enumerate}[label=\arabic*.]
\item If $K$ is the coalition of dummy variable agents $Z_i^1$, 
i.e., $\ell=1$, then the only possibility is that the other literal-agent who leaves the coalition is literal-agent $z_{i-1}^2\in\{y_{i-1}^2,\overline{y}_{i-1}^2\}$ 
if $i>1$, or agent $t$ if $i=1$.

If $i>1$, following the same argument as in case (a) for agent $z_p^2$, literal-agent $z_{i-1}^2$ 
cannot deviate to coalitions of dummy clause agents, otherwise she would never come back to the current coalition size. 
If 
$Z_i^1=Y_i^1$, then the only possibility is that literal-agent $z_{i-1}^2$ 
deviates for joining coalition $Z_{i-1}^2$ 
(and an additional literal-agent).
Otherwise, i.e., if 
$Z_i^1=\overline{Y_i^1}$, then literal-agent $z_{i-1}^2$ 
can deviate to join $Z_{i-1}^2$ 
(and an additional literal-agent), but she could also deviate for joining the coalition of dummy variable agents $Y_i^1$.
In the latter case, by construction of the preferences, the only other agent who can also join coalition $Y_i^1$ is literal-agent $y_i^1$.
If $y_i^1$ never joins this coalition, then agent $z_{i-1}^2$ would impact no deviating agent by staying in this coalition, and it would not enable her to come back to a previous less preferred coalition since no agent can leave the coalition.
Therefore, because she still prefers coalition $Z_{i-1}^2$, she would deviate anyway to coalition $Z_{i-1}^2$.
Otherwise, i.e., if $y_i^1$ joins the coalition (before or after $z_{i-1}^2$), it means that agent $y_i^1$ reaches the best possible coalition of dummy variable agents.
By following the same arguments as in case (a) for agent $z_p^2$, agent $y_i^1$ cannot deviate to coalitions of dummy clause agents, otherwise she would never come back to the current coalition size.
It follows that agent $y_i^1$ has no reason to deviate from the coalition she forms with $z_{i-1}^2$ and $Y_i^1$.
Therefore, in order to be left at some point of the cycle in order to come back to a previous less preferred coalition, agent $z_{i-1}^2$ still needs to deviate and the only possibility is to join coalition $Z_{i-1}^2$ (and an additional literal-agent).

If $i=1$, the same arguments can be applied and then agent $t$ deviates to join the coalition of dummy clause agents $K_{m+1}$.
\item If $K$ is the coalition of dummy variable agents $Z_i^2$, 
i.e., $\ell=2$, then the only possibility is that the other literal-agent who leaves the coalition is literal-agent $z_i^1$ (i.e., the literal-agent that is associated with the first occurrence of the same literal as $z_i^\ell$). 
Since, literal-agent $z_i^1$ 
cannot deviate to join a coalition of dummy clause agents (by following the same arguments as in case (a) for agent $z_p^2$), she will necessarily join the coalition of dummy variable agents $Z_i^1$ 
(with an additional literal-agent).
\end{enumerate}
\end{enumerate}

To summarize, if there is a cycle, only the following can occur: 
\begin{enumerate}
\item agent $t$ in coalition $K_{m+1}$ can only be left by a literal-agent corresponding to an occurrence of a literal belonging to clause $C_m$ who deviates to join the coalition of dummy clause agents $K_m$;
\item a literal-agent $y_i^\ell$ (or $\overline{y}_i^\ell$) in a coalition of dummy clause agents $K_j$ (for $1<j\leq m$), corresponding to the clause $C_j$ to which the $\ell^\text{th}$ occurrence of $x_i$ (or $\overline{x}_i$) belongs, can only be left by a literal-agent corresponding to an occurrence of a literal belonging to clause $C_{j-1}$ who deviates to join the coalition of dummy clause agents $K_{j-1}$;
\item a literal-agent $y_i^\ell$ (or $\overline{y}_i^\ell$) in the coalition of dummy clause agents $K_1$, corresponding to the clause $C_1$ to which the $\ell^\text{th}$ occurrence of $x_i$ (or $\overline{x}_i$) belongs, can only be left by literal-agent $y_p^2$ or $\overline{y}_p^2$ who joins the coalition of dummy variable agents $Y_p^2$ or $\overline{Y_p^2}$, respectively;
\item a literal-agent $y_i^2$ (or $\overline{y}_i^2$), for $1\leq i\leq p$, in a coalition of dummy variable agents $Y_i^2$ (or $\overline{Y_i^2}$) can only be left by literal-agent $y_{i}^1$ (or $\overline{y}_{i}^1$) who joins the coalition of dummy variable agents $Y_{i}^1$ (or $\overline{Y}_{i}^1$);
\item a literal-agent $y_i^1$ (or $\overline{y}_i^1$), for $1<i\leq p$, in a coalition of dummy variable agents $Y_i^1$ (or $\overline{Y_i^1}$) can only be left by a literal-agent $y_{i-1}^2$ or $\overline{y}_{i-1}^2$ who eventually joins the coalition of dummy variable agents $Y_{i-1}^2$ or $\overline{Y}_{i-1}^2$, respectively;
\item literal-agent $y_1^1$ (or $\overline{y}_1^1$) in the coalition of dummy variable agents $Y_1^1$ (or $\overline{Y_1^1}$) can only be left by agent $t$ who joins the coalition of dummy clause agents $K_{m+1}$.
\end{enumerate}

Therefore, as soon as there is one deviating literal-agent in the cycle, it implies that there exists a whole chain of deviating agents in the cycle, who alternate between joining coalitions of dummy agents which are consecutive in the following cycle over initial coalitions: $K_1 < K_2 < \dots < K_m < K_{m+1} < Z_1^1 < Z_1^2 < Z_2^1 < Z_2^2 < \dots < Z_p^1 < Z_p^2 < K_1$, where $Z_i^1$ and $Z_i^2$ refer either to $Y_i^1$ and $Y_i^2$, or to $\overline{Y_i^1}$ and $\overline{Y_i^2}$.
More precisely, for the cycle to occur, we need, as deviating agents, $(i)$ for each clause, a literal-agent corresponding to an occurrence of a literal belonging to this clause who alternates between the coalition of dummy clause agents associated with this clause and the next coalition in the previously mentioned cycle over initial coalitions (cases (1)-(3)) 
and, $(ii)$ for each variable, two literal-agents corresponding to the same literal, who alternate between the coalition of dummy variable agents associated with their literal occurrence and the next coalition in the previously mentioned cycle over initial coalitions (cases (3)-(6)).
As described in the previous case distinctions, these two groups of deviating literal-agents ($(i)$ and $(ii)$) are distinct because,
by construction of the preferences, 
once a literal-agent reaches the coalition of dummy clause agents associated with the clause to which her corresponding literal belongs, she cannot be in a coalition of less preferred size, and in particular the coalition of dummy variable agents associated with her literal 
(recall that only literal-agents can deviate and that at most two literal-agents can join a coalition of dummy agents).
Moreover, as summarized in case (4), we need that the two deviating literal-agents associated with each variable (group $(ii)$) correspond to the same literal of the variable.
Hence, by setting to true the literals associated with the deviating literal-agents of group $(i)$ and to false the literals associated with the deviating literal-agents of group $(ii)$ (and arbitrarily the rest of literals), we get a valid truth assignment of the variables which satisfies all the clauses.
\end{proof}

\section{Hedonic Diversity Games}

In this section, we provide the missing proofs for hedonic diversity games. First, we consider the remaining two examples of cycling under strong assumptions.

\cycleHDG* 

\begin{proof}
\begin{enumerate}
\item[$(\lnot 1)$] Our next example satisfies all properties except single-peakedness of preferences. Let us consider an HDG with 12 agents: 3 red agents and 9 blue agents.
There are two deviating agents in the cycle: blue agents $b_1$ and $b_2$.
In the cycle, there are three fixed coalitions $C_1$, $C_2$ and $C_3$ such that:
\begin{itemize}
\item $C_1$ contains 1 red agent and 2 blue agents;
\item $C_2$ contains 1 red agent and 1 blue agent;
\item $C_3$ contains 1 red agent and 4 blue agents.
\end{itemize}

The relevant part of the preferences of the agents is given below. Recall Footnote~\ref{fot:relevantprefs} where we describe how to complete preferences.

{\centering
\renewcommand{\arraystretch}{1.25}
\begin{tabular}{rl}
$b_1:$ & $\frac{1}{5} \succ \frac{1}{3} \succ \frac{1}{7} \succ \frac{1}{6} \succ \frac{1}{4} \succ 0$ \\
$b_2:$ & $\frac{1}{7} \succ \frac{1}{3} \succ \frac{1}{5} \succ \frac{1}{4} \succ \frac{1}{6} \succ 0$ \\
$C_1:$ & $\frac{1}{5} \succ \frac{1}{4} \succ \frac{1}{3} \succ \frac{1}{2} \succ [1 \text{ if red, } 0 \text{ otherwise}]$  \\ 
$C_2:$ & $\frac{1}{3} \succ \frac{1}{2} \succ [1 \text{ if red, } 0 \text{ otherwise}]$ \\ 
$C_3:$ & $\frac{1}{7} \succ \frac{1}{6} \succ \frac{1}{5} \succ \frac{1}{4} \succ \frac{1}{3} \succ \frac{1}{2} \succ [1 \text{ if red, } 0 \text{ otherwise}]$ \\
\end{tabular}\par}
\medskip
Consider the sequence of IS deviations in Figure~\ref{fig:HDG-cycle-excSP} that describes a cycle in the dynamics. 
The two deviating agents of the cycle $b_1$ and $b_2$ are marked in bold and the specific deviating agent between two states is indicated next to the arrows.

\begin{figure}
{\centering
\resizebox{0.75\textwidth}{!}{
\begin{tikzpicture}
\node (1) at (0,0) {\begin{tabular}{|ccc|}\hline $C_1\cup\{\mathbf{b_1}\}$ & $C_2\cup\{\mathbf{b_2}\}$ & $C_3$ \\ \hline 
$\frac{1}{4}$ & $\frac{1}{3}$ & $\frac{1}{5}$ \\ \hline\end{tabular}};
\node (2) at ($(1)+(1.25,2)$) {\begin{tabular}{|ccc|}\hline $C_1$ & $C_2\cup\{\mathbf{b_2}\}$ & $C_3\cup\{\mathbf{b_1}\}$ \\ \hline 
$\frac{1}{3}$ & $\frac{1}{3}$ & $\frac{1}{6}$ \\ \hline\end{tabular}};
\node (3) at ($(2)+(6.25,0)$) {\begin{tabular}{|ccc|}\hline $C_1$ & $C_2$ & $C_3\cup\{\mathbf{b_1},\mathbf{b_2}\}$ \\ \hline 
$\frac{1}{3}$ & $\frac{1}{2}$ & $\frac{1}{7}$ \\ \hline\end{tabular}};
\node (4) at ($(1)+(8.5,0)$) {\begin{tabular}{|ccc|}\hline $C_1$ & $C_2\cup\{\mathbf{b_1}\}$ & $C_3\cup\{\mathbf{b_2}\}$ \\ \hline 
$\frac{1}{3}$ & $\frac{1}{3}$ & $\frac{1}{6}$ \\ \hline\end{tabular}};
\node (6) at ($(1)+(1.25,-2)$) {\begin{tabular}{|ccc|}\hline $C_1\cup\{\mathbf{b_1},\mathbf{b_2}\}$ & $C_2$ & $C_3$ \\ \hline 
$\frac{1}{5}$ & $\frac{1}{2}$ & $\frac{1}{5}$ \\ \hline\end{tabular}};
\node (5) at ($(6)+(6.25,0)$) {\begin{tabular}{|ccc|}\hline $C_1\cup\{\mathbf{b_2}\}$ & $C_2\cup\{\mathbf{b_1}\}$ & $C_3$ \\ \hline 
$\frac{1}{4}$ & $\frac{1}{3}$ & $\frac{1}{5}$ \\ \hline\end{tabular}};
\draw[arrow] (1) --node[midway,left]{$b_1$} (2); \draw[arrow] (2) --node[midway,above]{$b_2$}  (3); \draw[arrow] (3) --node[midway,right]{$b_1$}  (4); \draw[arrow] (4) --node[midway,right]{$b_2$}  (5); \draw[arrow] (5) --node[midway,above]{$b_1$}  (6); \draw[arrow] (6) --node[midway,left]{$b_2$}  (1);
\end{tikzpicture}}\par}
	\caption{Possibility of cycling of IS dynamics in part $(\lnot 1)$ of Theorem~\ref{thm:cycleHDG}. Here, we consider IS dynamics starting from the \singleton in HDGs under strict preferences, where all deviations satisfy solitary homogeneity.}\label{fig:HDG-cycle-excSP}
\end{figure}

To show that this cycle can be reached from the \singleton, it suffices to observe that the two deviating agents $b_1$ and $b_2$ prefer to join the fixed coalitions than being alone and that each fixed coalition can be formed from the \singleton: the red agent of each future fixed coalition joins first a blue agent and then all the other blue agents of the future fixed coalition successively join. Note that all deviations result in non-homogeneous target coalitions, and therefore satisfy solitary homogeneity.

\item[$(\lnot 2)$] Our final example satisfies all properties except strictness of preferences. Let us consider an HDG with 10 agents: 4 red agents and 6 blue agents.
There are two deviating agents: red agent $r$ and blue agent $b$, and three fixed coalitions $C_1$, $C_2$ and $C_3$ such that:
\begin{itemize}
\item $C_1$ contains 2 red agents;
\item $C_2$ contains 1 red agent and 3 blue agents;
\item $C_3$ contains 2 blue agents.
\end{itemize} 

The relevant part of the preferences of the agents is given below.

\renewcommand{\arraystretch}{1.25}
{\centering
\begin{tabular}{rl}
$r:$ & $\frac{3}{4} \succ \frac{2}{5} \succ \frac{1}{4} \sim \frac{1}{3} \succ 1$ \\
$b:$ & $\frac{1}{4} \succ \frac{1}{5} \succ \frac{1}{2} \sim \frac{2}{3} \sim \frac{3}{4} \succ 0$ \\
$C_1:$ & $\frac{3}{4} \succ \frac{2}{3} \succ \frac{1}{2} \sim \frac{1}{3} \succ 1$ \\
$C_2:$ & $\frac{2}{5} \succ \frac{1}{3} \sim \frac{1}{4} \sim \frac{1}{5} \succ \frac{1}{2} \succ [1 \text{ if red, } 0 \text{ otherwise}]$ \\
$C_3:$ & $\frac{1}{4} \succ \frac{1}{3} \succ \frac{1}{2} \succ 0$ \\
\end{tabular}\par}
\smallbreak
Consider the sequence of IS deviations in Figure~\ref{fig:HDG-cycle-excSt} that describes a cycle in the dynamics. 
The two deviating agents of the cycle $r$ and $b$ are marked in bold and the specific deviating agent between two states is indicated next to the arrows.

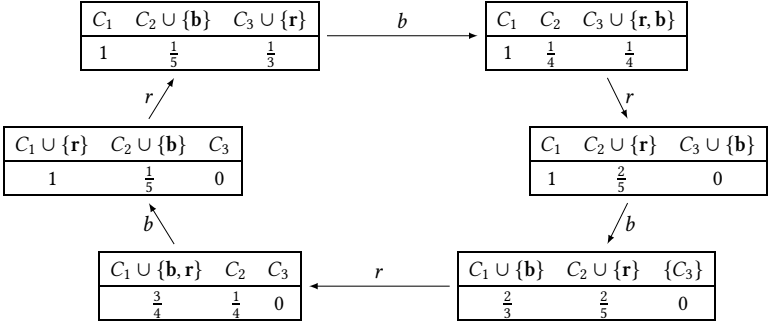
\begin{figure}
{\centering
\resizebox{0.75\textwidth}{!}{
\begin{tikzpicture}
\node (1) at (0,0) {\begin{tabular}{|ccc|}\hline $C_1\cup\{\mathbf{r}\}$ & $C_2\cup\{\mathbf{b}\}$ & $C_3$ \\ \hline 
$1$ & $\frac{1}{5}$ & $0$ \\ \hline\end{tabular}};
\node (2) at ($(1)+(1.25,2)$) {\begin{tabular}{|ccc|}\hline $C_1$ & $C_2\cup\{\mathbf{b}\}$ & $C_3\cup\{\mathbf{r}\}$ \\ \hline 
$1$ & $\frac{1}{5}$ & $\frac{1}{3}$ \\ \hline\end{tabular}};
\node (3) at ($(2)+(6.25,0)$) {\begin{tabular}{|ccc|}\hline $C_1$ & $C_2$ & $C_3\cup\{\mathbf{r},\mathbf{b}\}$ \\ \hline 
$1$ & $\frac{1}{4}$ & $\frac{1}{4}$ \\ \hline\end{tabular}};
\node (4) at ($(1)+(8.5,0)$) {\begin{tabular}{|ccc|}\hline $C_1$ & $C_2\cup\{\mathbf{r}\}$ & $C_3\cup\{\mathbf{b}\}$ \\ \hline 
$1$ & $\frac{2}{5}$ & $0$ \\ \hline\end{tabular}};
\node (6) at ($(1)+(1.25,-2)$) {\begin{tabular}{|ccc|}\hline $C_1\cup\{\mathbf{b},\mathbf{r}\}$ & $C_2$ & $C_3$ \\ \hline 
$\frac{3}{4}$ & $\frac{1}{4}$ & $0$ \\ \hline\end{tabular}};
\node (5) at ($(6)+(6.25,0)$) {\begin{tabular}{|ccc|}\hline $C_1\cup\{\mathbf{b}\}$ & $C_2\cup\{\mathbf{r}\}$ & $\{C_3\}$ \\ \hline 
$\frac{2}{3}$ & $\frac{2}{5}$ & $0$ \\ \hline\end{tabular}};
\draw[arrow] (1) --node[midway,left]{$r$} (2); \draw[arrow] (2) --node[midway,above]{$b$}  (3); \draw[arrow] (3) --node[midway,right]{$r$}  (4); \draw[arrow] (4) --node[midway,right]{$b$}  (5); \draw[arrow] (5) --node[midway,above]{$r$}  (6); \draw[arrow] (6) --node[midway,left]{$b$}  (1);
\end{tikzpicture}}\par}
\caption{Possibility of cycling of IS dynamics in part $(\lnot 2)$ of Theorem~\ref{thm:cycleHDG}. Here, we consider IS dynamics starting from the \singleton in HDGs under naturally single-peaked preferences, where all deviations satisfy solitary homogeneity.}\label{fig:HDG-cycle-excSt}
\end{figure}

To show that this cycle can be reached from the \singleton, it suffices to observe that partition $\{C_1\cup\{b\},C_2\cup\{r\},C_3\}$ belonging to the cycle can be reached from the \singleton.
Indeed, agent $b$ can join a red agent from the future fixed coalition $C_1$ while the other red agent of the future fixed coalition $C_1$ can join a blue agent from the future fixed coalition $C_3$.
The second blue agent of the future fixed coalition $C_3$ then joins them and afterwards, the red agent leaves them to join $b$ and the other red agent of $C_1$.
For forming coalition $C_2$, the red agent joins one of the blue agents, and then the two remaining blue agents join them.
Agent $r$ can then join coalition $C_2$. Note that all deviations result in non-homogeneous target coalitions, and therefore satisfy solitary homogeneity.
\end{enumerate}
\end{proof}

Next, we provide a proof for the lemma used for the part of the proof of the previous theorem given in the body of the paper.

\HDGhom*

\begin{proof}
	In the following proof, we consider various sets of auxiliary agents. We assume that we take \emph{new} agents in every step of the constructions.
	
	We show the statement for homogeneous coalitions of blue agents. The statement for red agents is completely symmetric, by reversing the respective roles. We use a few types of auxiliary agents with extreme preferences that have their peaks at the largest or smallest ratios, except for liking homogeneous coalitions the worst. Specifically, we consider four sets $R_x, B_x, R_y$, and $B_y$ of agents with the following preferences. Note that the sets $R_y$ and $B_y$ are only needed for the statement about red agents, but we state them for completeness.
	
	{\centering
		\renewcommand{\arraystretch}{1.25}
		\begin{tabular}{rl}
			$R_x, B_x:$ & $f\succ g$ if and only if $0 < f < g$ or $f > g = 0$\\
			$R_y, B_y:$ & $f\succ g$ if and only if $1 > f > g$ or $f < g = 1$\\
		\end{tabular}\par}
	\medskip
	
	Now, let $B_a$ be a set of blue agents such that every agent in $B_a$ has preferences satisfying $\frac 13 \succ 0 \succ \frac 12$. Suppose that $|B_a| = k$ and $B_a = \{b_{a,i}\colon i \in [k]\}$. Let $r_{x,i}\in R_x$ for $i\in [2]$ and $b_{x,i}\in B_x$ for $i \in [k + 2]$. We create a \emph{trash coalition} (used to get rid of agents not needed anymore) by creating $C_a = \{r_{x,2}, b_{x,k}, b_{x,k+1}\}$. To create $C_a$, we let the blue agents of this coalition join the red agent~$r_{x,2}$.
	
	Now, we perform for any $i\in [k-1]$ the following steps (in increasing order of indices): $b_{x,i}$ joins $r_{x,1}$, then $b_{a,i}$ joins this coalition. Then, $b_{x,i}$ leaves this coalition to join $C_a \cup \{b_{x,j}\colon j\in [i-1]\}$, and finally $b_{a,i}$ joins $b_{a,k}\cup \{b_{a,j}\colon j\in [i-1]\}$. Note that these are all IS deviations. In particular, $b_{a,i}$'s deviation to join $b_{a,k}$'s coalition is feasible, because she leaves coalition $\{r_{x,1}, b_{a,i}\}$ which has a ratio of $1/2$, which is strictly worse for her than being in a homogeneous coalition. After this procedure, we have obtained the coalition $B_a$.
\end{proof}

\existHDG*

We prove the two hardness results by providing separate reductions for each problem in the next two lemmas.
The proofs for these two lemmas (Lemmas~\ref{lem:exist-HDG-NPH} and \ref{lem:hard-converg-HDG}) work in the same way as the proofs of Lemmas~\ref{lem:hard-exist-AHG} and \ref{lem:hard-converg-AHG}, 
except that we have to ensure appropriate ratios of red agents in each constructed initial coalition in order to guarantee, similarly as in the proofs of Lemmas~\ref{lem:hard-exist-AHG} and \ref{lem:hard-converg-AHG}, that the preferences of the agents over ratios can be expressed in terms of preferences over initial non-singleton coalitions augmented by one or two agents of a given type (red or blue). 
Instead of playing only with the size of coalitions for the design of the reductions, like in the proofs of Lemmas~\ref{lem:hard-exist-AHG} and \ref{lem:hard-converg-AHG}, we need here to play with both the size of the coalitions and the proportion of red and blue agents within them.

\begin{lemma}\label{lem:exist-HDG-NPH}
\existpb{IS}{HDG} is \np-hard even for strict preferences.
\end{lemma}

\begin{proof}
Let us perform a reduction from (3,B2)-SAT~\citep{BKS03a}. 
In an instance of (3,B2)-SAT, we are given a CNF propositional formula $\varphi$ where every clause $C_j$, for $1\leq  j \leq m$, contains exactly three literals and every variable $x_i$, for $1\leq i \leq p$, appears exactly twice as a positive literal and twice as a negative literal. 
From such an instance, we construct an instance of a hedonic diversity game with initial partition as follows.
The proof works in the same way as the proof of Lemma~\ref{lem:hard-exist-AHG}. 

For each $\ell^{\text{th}}$ occurrence ($\ell\in\{1,2\}$) of a positive literal $x_i$ (or negative literal $\overline{x}_i$), we create a red literal-agent $y_i^\ell$ (or a blue literal-agent $\overline{y}_i^\ell$).
All literal-agents are singletons in the initial partition $\spartition$.
Let us consider three integers $\alpha$, $\beta$ and $\gamma$ such that $(1)$ $\alpha>2m-1$, $\beta>\max\{3p-2;3p\alpha+3p-2\}$, $\gamma>\max\{12p-2;6p\beta+12p-1\}$. 
For instance, we can set the following values: $\alpha=m^2$, $\beta=m^4$ and $\gamma=m^7$ (one can verify that condition $(1)$ is satisfied, especially because in a (3,B2)-SAT instance, it holds that $m\geq 4$ and $p=3/4 m$).
For each clause $C_j$, we then create $\alpha$ dummy clause-agents with among them $2j-1$ red agents. 
They are all grouped within the same coalition $K_j$ in the initial partition $\spartition$.
For each literal $x_i$ (or $\overline{x}_i$), we create a red variable-agent $z_i$ (or a blue variable-agent $\overline{z}_i$) and $\beta-1$ (or $\beta -1$) dummy variable-agents with among them $3i-2$ (or $3i-2$) red agents ($z_i$ included).
They are all grouped within the same coalition $Z_i$ (or $\overline{Z_i}$) in the initial partition $\spartition$.
Finally, for each variable $x_i$, we create three coalitions in partition $\spartition$ of dummy agents $G_i^1$, $G_i^2$ and $G_i^3$ of size $i\gamma$ with among them, $6(p-i)+1$, $6(p-i)+3$, or $6(p-i)+5$ red agents, for each coalition respectively.
These dummy agents are used as a gadget for a cycle.
Although we have created many agents, the construction remains polynomial by considering reasonable values of $\alpha$, $\beta$ and $\gamma$, as previously described. 

\newcommand{\clauseof}[1]{cl(#1)}

The preferences of the agents over ratios of red agents are given in Table~\ref{tab:prefs-exist-HDG-NPH}.

\begin{table}
	\caption{Preferences of the agents in the reduced instance of Lemma~\ref{lem:exist-HDG-NPH}, for every $1\leq i\leq p$, $1\leq j\leq m$, $\ell\in\{1,2\}$. Notation $\clauseof{x_i^\ell}$ (or $\clauseof{\overline{x}_i^\ell}$) stands for the index of the clause to which the $\ell^{\text{th}}$ occurrence of literal $x_i$ (or $\overline{x}_i$) belongs, the framed value corresponds to the ratio of the initial coalition in partition $\spartition$, and $[\dots]$ denotes an arbitrary order over the rest of the coalition ratios.\label{tab:prefs-exist-HDG-NPH}}
\begin{center}
\begin{tabular}{rl}
$z_i:$ & $\frac{3i}{\beta+2} \succ \frac{6(p-i)+2}{i\gamma+2} \succ \frac{6(p-i)+4}{i\gamma+1} \succ \frac{6(p-i)+6}{i\gamma+2} \succ \frac{6(p-i)+6}{i\gamma+1} \succ \frac{6(p-i)+2}{i\gamma+1} \succ \frac{3i-1}{\beta+1} \succ \boxed{\frac{3i-2}{\beta}} \succ [\dots]$ \\
$\overline{z}_i:$ & $\frac{3i-2}{\beta+2} \succ \frac{6(p-i)+6}{i\gamma+2} \succ \frac{6(p-i)+3}{i\gamma+1} \succ \frac{6(p-i)+2}{i\gamma+2} \succ \frac{6(p-i)+1}{i\gamma+1} \succ \frac{6(p-i)+5}{i\gamma+1} \succ \frac{3i-2}{\beta+1} \succ \boxed{\frac{3i-2}{\beta}}  \succ [\dots]$ \\
$y_i^\ell:$ & $\frac{2\clauseof{x_i^\ell}}{\alpha+1} \succ \frac{3i}{\beta+2} \succ \frac{3i-1}{\beta+1} \succ \boxed{1} \succ [\dots]$ \\
$\overline{y}_i^\ell:$ & $\frac{2\clauseof{\overline{x}_i^\ell}-1}{\alpha+1} \succ \frac{3i-2}{\beta+2} \succ \frac{3i-2}{\beta+1} \succ \boxed{0} \succ [\dots]$ \\
\midrule
$K_j:$ & $\frac{2j}{\alpha+1} \succ \frac{2j-1}{\alpha+1} \succ \boxed{\frac{2j-1}{\alpha}} \succ [\dots]$ \\
$Z_i\setminus\{z_i\}:$ & $\frac{3i}{\beta+2} \succ \frac{3i-1}{\beta+1} \succ \boxed{\frac{3i-2}{\beta}} \succ \frac{3i-3}{\beta-1} \succ [\dots]$ \\
$\overline{Z_i}\setminus\{\overline{z}_i\}:$ & $\frac{3i-2}{\beta+2} \succ \frac{3i-2}{\beta+1} \succ \boxed{\frac{3i-2}{\beta}} \succ \frac{3i-2}{\beta-1} \succ [\dots]$  \\
$G_i^1:$ & $\frac{6(p-i)+2}{i\gamma+2} \succ \frac{6(p-i)+2}{i\gamma+1} \succ \frac{6(p-i)+1}{i\gamma+1} \succ \boxed{\frac{6(p-i)+1}{i\gamma}} \succ [\dots]$ \\
$G_i^2:$ & $\frac{6(p-i)+4}{i\gamma+1} \succ \frac{6(p-i)+3}{i\gamma+1} \succ \boxed{\frac{6(p-i)+3}{i\gamma}} \succ [\dots]$ \\
$G_i^3:$ & $\frac{6(p-i)+6}{i\gamma+2} \succ \frac{6(p-i)+6}{i\gamma+1} \succ \frac{6(p-i)+5}{i\gamma+1} \succ \boxed{\frac{6(p-i)+5}{i\gamma}} \succ [\dots]$ \\
\end{tabular}
\end{center}
\end{table}

We claim that there exists a sequence of IS deviations which leads to an IS partition iff formula $\varphi$ is satisfiable.

Suppose first that there exists a truth assignment of the variables $\phi$ that satisfies all the clauses.
Let us denote by $\ell_j$ a chosen literal-agent associated with an occurrence of a literal true in $\phi$ which belongs to clause $C_j$.
Since all the clauses of $\varphi$ are satisfied by $\phi$, there exists such a literal-agent $\ell_j$ for each clause $C_j$.
For every clause $C_j$, let literal-agent $\ell_j$ join coalition $K_j$. 
These IS deviations make the chosen literal-agents reach their most preferred ratio so none of them will deviate afterwards. 
For the clause-agents, they all reach either their first or second most preferred ratio but have no possibility to improve their satisfaction in the latter case so none of them will deviate afterwards neither.
Then, let all remaining literal-agents $y_i^\ell$ (or $\overline{y}_i^\ell$) deviate by joining coalition $Z_i$ (or $\overline{Z_i}$).
Since $\phi$ is a truth assignment of the variables, for each variable $x_i$, there exists a coalition $Z_i$ or $\overline{Z_i}$ that is joined by two literal-agents and thus reaches the most preferred ratio $\frac{3i}{\beta+2}$ or $\frac{3i-2}{\beta+2}$.
It follows that no member of such a newly formed coalition would move afterwards or let other agents enter the coalition: all members of $Z_i$ or $\overline{Z_i}$ get their most preferred size while the two joining literal-agents get their second most preferred size and their most preferred size is not accessible anymore (their associated clause coalition has already been joined by another literal-agent). 
For each variable $x_i$, at most one coalition between $Z_i$ and $\overline{Z_i}$ may not be joined by two literal-agents and, if there is one, it must be the coalition that corresponds to the literal of variable $x_i$ that is true in $\phi$. 
In such a case, we let the associated variable-agent $z_i$ or $\overline{z}_i$ deviate for joining coalition $G_i^2$, and if one literal-agent previously joined the corresponding variable-coalition $Z_i$ or $\overline{Z_i}$, she deviates to be alone.
Such a literal-agent then gets her fourth most preferred size while her most preferred ones are not accessible anymore (because the variable-agent has left the coalition and her associated clause coalition has already been joined by another literal-agent).
Moreover, such a variable-agent $z_i$ (or $\overline{z}_i$), by joining coalition $G_i^2$, gets her third most preferred size while her most preferred ones are not accessible: no two additional red (or blue) agents want to enter the initial coalition $Z_i$ (or $\overline{Z_i}$) and only one additional agent, herself, is present in the gadget associated with variable $x_i$, whereas variable-agent $z_i$ (or $\overline{z}_i$) in the gadget prefers ratios which differ by one blue agent (or red agent) from the ratio of the current coalitions.
Also, note that, by the design of the preferences, no dummy agent in the gadget has an incentive to move to another coalition.
All in all, no agent can then move in an IS deviation, and thus the reached partition is IS.

Suppose now that there exists no truth assignment of the variables that satisfies all the clauses.
That means that it is not possible that one literal-agent joins each clause coalition while two literal-agents $y_i^1$ and $y_i^2$ join coalition $Z_i$ or $\overline{y}_i^1$ and $\overline{y}_i^2$ join coalition $\overline{Z_i}$ for each variable $x_i$.
By construction of the preferences, the only agents who want to join a coalition $K_j$ are literal-agents associated with a literal belonging to clause $C_j$ and the only agents who want to join a coalition $Z_i$ (or $\overline{Z_i}$) are literal-agents $y_i^1$ and $y_i^2$ (or $\overline{y}_i^1$ and $\overline{y}_i^2$).
Moreover, since each literal-agent prefers to join clause coalitions than variable coalitions, it means that in a maximal sequence of IS deviations, all clause-agents in $K_j$ will reach one of the two most preferred ratios, $\frac{2j}{\alpha+1}$ in case a red literal-agent joined or $\frac{2j-1}{\alpha+1}$ in case a blue literal-agent joined.
In both cases, they have no incentive to deviate afterwards.
However, in such a case, there exists a variable $x_i$ such that at most one literal-agent joins coalition $Z_i$ and $\overline{Z_i}$.
It follows that both variable-agents $z_i$ and $\overline{z_i}$ have an incentive to deviate to the gadget associated with variable $x_i$ (their respective most preferred ratios $\frac{3i}{\beta+2}$ and $\frac{3i-2}{\beta+2}$ can never be reached).
Within the gadget associated with variable $x_i$, variable-agents $z_i$ and $\overline{z_i}$ are the only agents who can deviate and we necessarily reach a cycle, which is the same as described in Figure~\ref{fig:cycle-reduc-AHG} for the proof of Lemma~\ref{lem:hard-exist-AHG}. 

Finally, we must verify that all the fractions described in the preferences with different variables are indeed different.
\begin{itemize}
\item For gadget coalitions, since $\gamma>12p-2$, it holds that $\frac{6(p-i)+6}{i\gamma+1}>\frac{6(p-i)+6}{i\gamma+2}>\frac{6(p-i)+5}{i\gamma}>\frac{6(p-i)+5}{i\gamma+1}>\frac{6(p-i)+4}{i\gamma+1}>\frac{6(p-i)+3}{i\gamma}>\frac{6(p-i)+3}{i\gamma+1}>\frac{6(p-i)+2}{i\gamma+1}>\frac{6(p-i)+2}{i\gamma+2}>\frac{6(p-i)+1}{i\gamma}>\frac{6(p-i)+1}{i\gamma+1}$ for every $i\in [p]$.
Moreover, it holds that $\frac{6(p-i)+6}{i\gamma+1}>\frac{6(p-(i+1)+1}{(i+1)\gamma+1}$ for every $i\in[p-1]$ so all the values associated with ratios preferred to the initial ones are different for all gadget coalitions.
\item For variable coalitions, since $\beta>3p-2$, it holds that $\frac{3i}{\beta+2}>\frac{3i-1}{\beta+1}>\frac{3i-2}{\beta}>\frac{3i-2}{\beta+1}>\frac{3i-2}{\beta+2}$ for every $i\in[p]$.
Moreover, it holds that $\frac{3i}{\beta+2}<\frac{3(i+1)-2}{\beta+2}$ for every $i\in[p-1]$ so all the values associated with ratios preferred to the initial ones are different for all variable coalitions.
\item For clause coalitions, since $\alpha>2m-1$, it holds that $\frac{2j-1}{\alpha+1}<\frac{2j-1}{\alpha}<\frac{2j}{\alpha+1}$ for every $j\in[m]$.
Moreover, it holds that $\frac{2j}{\alpha+1}<\frac{2(j+1)-1}{\alpha+1}$ for every $j\in [m-1]$ so all the values associated with ratios preferred to the initial ones are different for all clause coalitions.
\end{itemize}
It remains to check that the ratios associated with clause, variable or gadget coalitions do not interfere with each other.
Since $\gamma>6p\beta+12p-1$, it holds that the highest reachable ratio associated with a gadget coalition is smaller than the smallest reachable ratio associated with a variable coalition, i.e., $\frac{6p}{\gamma+1}<\frac{1}{\beta+2}$. 
Since $\beta>3p\alpha+3p-2$, it holds that the highest reachable ratio associated with a variable coalition is smaller than the smallest reachable ratio associated with a clause coalition, i.e., $\frac{3p}{\beta+2}<\frac{1}{\alpha+1}$.
Therefore, all the reachable ratios are indeed different for clause, variable and gadget coalitions.
It follows that the previously described deviations are indeed the only possible ones and hence no sequence of IS deviations can reach an IS partition.
\end{proof}

\begin{lemma}\label{lem:hard-converg-HDG}
\convergpb{IS}{HDG} is \conp-hard even for strict preferences.
\end{lemma}

\begin{proof}
For this purpose, we prove the \np-hardness of the complement problem, which asks whether there exists a cycle of IS deviations. 
Let us perform a reduction from (3,B2)-SAT~\citep{BKS03a}. 
In an instance of (3,B2)-SAT, we are given a CNF propositional formula $\varphi$ where every clause $C_j$, for $1\leq  j \leq m$, contains exactly three literals and every variable $x_i$, for $1\leq i \leq p$, appears exactly twice as a positive literal and twice as a negative literal. 
From such an instance, we construct an instance of a hedonic diversity game with initial partition as follows.
The proof works in the same way as the proof of Lemma~\ref{lem:hard-converg-AHG}. 

For each $\ell^{\text{th}}$ occurrence ($\ell\in\{1,2\}$) of a positive literal $x_i$ (or negative literal $\overline{x}_i$), we create a red literal-agent $y_i^\ell$ (or a blue literal-agent $\overline{y}_i^\ell$).
We create another red agent $t$.
All these agents are singletons in the initial partition $\spartition$.
Let us consider four integers $\alpha$, $\beta_1^+$, $\beta_1^-$ and $\beta_2$ such that $(1)$ $\alpha>6m+2$, $\beta_1^+>\max\{4p-2;(2p+1)\alpha+4p\}$, $\beta_1^->\max\{4p-2;2p\beta_1^++2p-2\}$ and $\beta_2>\max\{3p-2;3p\beta_1^-+4p\}$. 
For instance, we can set the following values: $\alpha=m^3$, $\beta_1^+=m^5$, $\beta_1^-=m^7$ and $\beta_2=m^9$ (one can verify that condition $(1)$ is satisfied, especially because in a (3,B2)-SAT instance, it holds that $m\geq 4$ and $p=3/4 m$).
For each clause $C_j$, we then create $\alpha$ dummy clause-agents with among them $3j-2$ red agents. 
They are all grouped within the same coalition $K_j$ in the initial partition $\spartition$.
We also create $\alpha$ dummy agents with among them $3m+1$ red agents, they are all grouped within the same coalition $K_{m+1}$ in initial partition $\spartition$.
For each first occurrence of literal $x_i$ (or $\overline{x}_i$), we create $\beta_1^+$ (or $\beta_1^-$) dummy variable agents with among them $2i-1$ red agents, they are all grouped within the same coalition $Y_i^1$ (or $\overline{Y_i^1}$) in the initial partition $\spartition$.
Finally, for each second occurrence of literal $x_i$ (or $\overline{x}_i$), we create $\beta_2$ dummy variable agents with among them $3i-2$ red agents, they are all grouped within the same coalition $Y_i^2$ (or $\overline{Y_i^2}$) in the initial partition $\spartition$.
Although we have created many agents, the construction remains polynomial by considering reasonable values of $\alpha$, $\beta^+_1$, $\beta^-_1$ and $\beta_2$, as previously described. 

\newcommand{\clauseof}[1]{cl(#1)}

The preferences of the agents over coalition ratios are given in Table~\ref{tab:pref-reduc-converg-HDG}.

\begin{table}
\caption{Preferences of the agents in the reduced instance of Lemma~\ref{lem:hard-converg-HDG}, for every $1\leq i\leq p$, $1\leq i' < p$, $1\leq j\leq m+1$, $\ell\in\{1,2\}$. Notation $\clauseof{x_i^\ell}$ (or $\clauseof{\overline{x}_i^\ell}$) stands for the index of the clause to which the $\ell^{\text{th}}$ occurrence of literal $x_i$ (or $\overline{x}_i$) belongs, the framed value corresponds to the ratio of the initial coalition in partition $\spartition$, and $[\dots]$ denotes an arbitrary order over the rest of the coalition ratios.}
\label{tab:pref-reduc-converg-HDG}
\begin{center}
\resizebox{\columnwidth}{!}{
\begin{tabular}{rl}
$y_i^1:$ & $\frac{3\clauseof{x_i^1}}{\alpha+2} \succ \frac{3\clauseof{x_i^1}-1}{\alpha+2} \succ \frac{3(\clauseof{x_i^1}+1)}{\alpha+2} \succ \frac{3(\clauseof{x_i^1}+1)-1}{\alpha+2} \succ \frac{3(\clauseof{x_i^1}+1)-1}{\alpha+1} \succ \frac{3\clauseof{x_i^1}-1}{\alpha+1} \succ \frac{2i+1}{\beta_1^++2} \succ \frac{2i}{\beta_1^++2} \succ \frac{3i}{\beta_2+2} \succ \frac{3i-1}{\beta_2+1} \succ \frac{2i}{\beta_1^++1} \succ  \boxed{1} \succ [\dots]$ \\

$y_{i'}^2:$ & $\frac{3\clauseof{x_{i'}^1}}{\alpha+2}  \succ \frac{3\clauseof{x_{i'}^1}-1}{\alpha+2} \succ \frac{3(\clauseof{x_{i'}^1}+1)}{\alpha+2} \succ \frac{3(\clauseof{x_{i'}^1}+1)-1}{\alpha+2} \succ \frac{3(\clauseof{x_{i'}^1}+1)-1}{\alpha+1} \succ \frac{3\clauseof{x_{i'}^1}-1}{\alpha+1} \succ  \frac{3i'}{\beta_2+2} \succ \frac{2(i'+1)+1}{\beta_1^++2} \succ \frac{2(i'+1)}{\beta_1^++2} \succ \frac{2(i'+1)}{\beta_1^++1}  \succ \frac{2(i'+1)+1}{\beta_1^-+2} \succ$ \\ & \hfill $\frac{2(i'+1)}{\beta_1^-+2}  \succ \frac{2(i'+1)}{\beta_1^-+1} \succ \frac{3i'-1}{\beta_2+1} \succ \boxed{1} \succ [\dots]$ \\

$y_p^2:$ & $\frac{3\clauseof{x_p^1}}{\alpha+2} \succ \frac{3\clauseof{x_p^1}-1}{\alpha+2} \succ \frac{3(\clauseof{x_p^1}+1)}{\alpha+2} \succ \frac{3(\clauseof{x_p^1}+1)-1}{\alpha+2} \succ \frac{3(\clauseof{x_p^1}+1)-1}{\alpha+1} \succ \frac{3\clauseof{x_p^1}-1}{\alpha+1} \succ \frac{3p}{\beta_2+2} \succ \frac{3}{\alpha+2} \succ \frac{2}{\alpha+2}  \succ \frac{2}{\alpha+1} \succ \frac{3p-1}{\beta_2+1} \succ  \boxed{1} \succ [\dots]$ \\

$\overline{y}_i^1:$ & $\frac{3\clauseof{\overline{x}_i^1}-1}{\alpha+2} \succ \frac{3\clauseof{\overline{x}_i^1}-2}{\alpha+2} \succ \frac{3(\clauseof{\overline{x}_i^1}+1)-1}{\alpha+2} \succ \frac{3(\clauseof{\overline{x}_i^1}+1)-2}{\alpha+2} \succ
 \frac{3(\clauseof{\overline{x}_i^1}+1)-2}{\alpha+1} \succ \frac{3\clauseof{\overline{x}_i^1}-2}{\alpha+1} \succ \frac{2i}{\beta_1^-+2} \succ \frac{2i-1}{\beta_1^-+2} \succ \frac{3i-2}{\beta_2+2} \succ \frac{3i-2}{\beta_2+1} \succ \frac{2i-1}{\beta_1^-+1} \succ  \boxed{0} \succ [\dots]$ \\
 
$\overline{y}_{i'}^2:$ & $\frac{3\clauseof{\overline{x}_{i'}^1}-1}{\alpha+2} \succ \frac{3\clauseof{\overline{x}_{i'}^1}-2}{\alpha+2} \succ \frac{3(\clauseof{\overline{x}_{i'}^1}+1)-1}{\alpha+2} \succ \frac{3(\clauseof{\overline{x}_{i'}^1}+1)-2}{\alpha+2} \succ \frac{3(\clauseof{\overline{x}_{i'}^1}+1)-2}{\alpha+1} \succ \frac{3\clauseof{\overline{x}_{i'}^1}-2}{\alpha+1} \succ \frac{3i'-2}{\beta_2+2} \succ \frac{2(i'+1)}{\beta_1^++2}  \succ \frac{2(i'+1)-1}{\beta_1^++2}  \succ \frac{2(i'+1)-1}{\beta_1^++1} \succ \frac{2(i'+1)}{\beta_1^-+2} \succ$ \\ & \hfill $ \frac{2(i'+1)-1}{\beta_1^-+2} \succ \frac{2(i'+1)-1}{\beta_1^-+1} \succ \frac{3i'-2}{\beta_2+1} \succ  \boxed{0} \succ [\dots]$ \\

$\overline{y}_p^2:$ & $\frac{3\clauseof{\overline{x}_p^1}-1}{\alpha+2} \succ \frac{3\clauseof{\overline{x}_p^1}-2}{\alpha+2} \succ \frac{3(\clauseof{\overline{x}_p^1}+1)-1}{\alpha+2} \succ \frac{3(\clauseof{\overline{x}_p^1}+1)-2}{\alpha+2} \succ \frac{3(\clauseof{\overline{x}_p^1}+1)-2}{\alpha+1} \succ \frac{3\clauseof{\overline{x}_p^1}-2}{\alpha+1} \succ \frac{3p-2}{\beta_2+2} \succ \frac{2}{\alpha+2} \succ \frac{1}{\alpha+2} \succ \frac{1}{\alpha+1} \succ \frac{3p-2}{\beta_2+1} \succ  \boxed{0} \succ [\dots]$ \\

$t:$ & $\frac{3m+3}{\alpha+2} \succ \frac{3m+2}{\alpha+2} \succ \frac{3}{\beta_1^++2} \succ \frac{2}{\beta_1^-+2} \succ \frac{2}{\beta_1^++1} \succ \frac{2}{\beta_1^-+1} \succ \frac{3m+2}{\alpha+1} \succ \boxed{1} \succ [\dots]$ \\

\midrule

$K_j:$ & $\frac{3j}{\alpha+2} \succ \frac{3j-1}{\alpha+2} \succ \frac{3j-2}{\alpha+2} \succ \frac{3j-1}{\alpha+1} \succ \frac{3j-2}{\alpha+1} \succ \boxed{\frac{3j-2}{\alpha}} \succ [\dots]$ \\

$Y_i^1:$ & $\frac{2i+1}{\beta_1^++2} \succ \frac{2i}{\beta_1^++2} \succ \frac{2i}{\beta_1^++1} \succ \frac{2i-1}{\beta_1^++1} \succ \boxed{\frac{2i-1}{\beta_1^+}} \succ [\dots]$ \\

$\overline{Y_i^1}:$ & $\frac{2i}{\beta_1^-+2} \succ \frac{2i-1}{\beta_1^-+2} \succ \frac{2i}{\beta_1^-+1} \succ \frac{2i-1}{\beta_1^-+1} \succ \boxed{\frac{2i-1}{\beta_1^-}} \succ [\dots]$ \\

$Y_i^2:$ & $\frac{3i}{\beta_2+2} \succ \frac{3i-1}{\beta_2+1} \succ \boxed{\frac{3i-2}{\beta_2}} \succ [\dots]$ \\

$\overline{Y_i^2}:$ & $\frac{3i-2}{\beta_2+2} \succ \frac{3i-2}{\beta_2+1} \succ \boxed{\frac{3i-2}{\beta_2}} \succ [\dots]$ \\
\end{tabular}}
\end{center}
\end{table}

We claim that there exists a cycle of IS deviations iff formula $\varphi$ is satisfiable.
We omit the formal proof of equivalence which follows exactly the same arguments as the proof of Lemma~\ref{lem:hard-converg-AHG} with even the same name of agents and fixed coalitions.
When given a truth assignment of the variables which satisfies formula $\varphi$, it is easy to see that the cycle described in the first part of the proof of Lemma~\ref{lem:hard-converg-AHG} can also occur in this instance (see Figure~\ref{fig:cycle-reduc-conv-AHG} for an example of such a cycle), proving the if part.
For the only if part, the same arguments as the ones given in the second part of the proof of Lemma~\ref{lem:hard-converg-AHG} also hold, except that we need to adapt to the context of evaluations of coalitions based on red agent ratios.
Instead of speaking about agents who prefer one or two additional agents, here we need to speak about agents who prefer one or two red or blue additional agents where, in our construction, literal-agents related to positive literals $x_i$ are red agents and literal-agents related to negative literals $\overline{x}_i$ are blue agents.  
This difference is already reflected in the construction of the preferences of the agents.
Indeed, when in our cycle we want deviating literal-agents that correspond to the same literal, the red/blue type is fixed in the preferences of the agents (e.g., we want that only literal-agents $y_i^1$ and $y_i^2$, who are red agents, can join the coalition of dummy variable agents $Y_i^2$, therefore the preferences of the members of $Y_i^2$ are constructed in such a way that they can only accept up to two additional red agents).
Alternatively, when the type of the deviating literal-agents is not known a priori in the cycle, we specify all possible new ratios in the preferences, in order to take into account every possible combination of types for deviating agents.
It is the case, e.g., for the preferences of the coalitions of dummy clause agents $K_j$ because we do not know a priori which literal will satisfy the corresponding clause. 
Therefore, for the preferences of $K_j$, we specify preferences for two additional literal-agents which can be of any type, by nevertheless ensuring that two additional agents is more preferred than only one (in order to come back to the same idea as in the proof of Lemma~\ref{lem:hard-converg-AHG}).
Even if an arbitrary preference ranking is chosen between coalition ratios that correspond to the same number of additional agents, this does not impact the proof because only one of these ratios can be reached, depending on the type of agent who enters the coalition.

The only point that must be additionally checked is that all the fractions described in the preferences with different variables are indeed different.
\begin{itemize}
\item For clause coalitions, since $\alpha>6m+2$, it holds that $\frac{3j}{\alpha+2} >\frac{3j-1}{\alpha+1} > \frac{3j-1}{\alpha+2} > \frac{3j-2}{\alpha} > \frac{3j-2}{\alpha+1} > \frac{3j-2}{\alpha+2}$ for every $j\in [m+1]$.
Moreover, it holds that $\frac{3j}{\alpha+2} <\frac{3(j+1)-2}{\alpha+2}$ for every $j\in[m]$ so all the values associated with ratios preferred to the initial ones are different for all clause coalitions.
\item For variable coalitions associated with the first positive occurrence of a variable, since $\beta_1^+>4p-2$, it holds that $\frac{2i+1}{\beta_1^++2}>\frac{2i}{\beta_1^++1}>\frac{2i}{\beta_1^++2}>\frac{2i-1}{\beta_1^+}>\frac{2i-1}{\beta_1^++1}$ for every $i\in[p]$.
Moreover, it holds that $\frac{2i+1}{\beta_1^++2}<\frac{2(i+1)-1}{\beta_1^++1}$ for every $i\in[p-1]$ so all the values associated with ratios preferred to the initial ones are different for all variable coalitions associated with the first positive occurrence of a variable.
\item For variable coalitions associated with the first negative occurrence of a variable, since $\beta_1^->4p-2$, it holds that $\frac{2i}{\beta_1^-+2}>\frac{2i}{\beta_1^-+2}>\frac{2i-1}{\beta_1^-}>\frac{2i-1}{\beta_1^-+1}>\frac{2i-1}{\beta_1^-+2}$ for every $i\in[p]$.
Moreover, it holds that $\frac{2i}{\beta_1^-+2}<\frac{2(i+1)-1}{\beta_1^-+2}$ for every $i\in[p-1]$ so all the values associated with ratios preferred to the initial ones are different for all variable coalitions associated with the first negative occurrence of a variable.
\item For variable coalitions associated with the second occurrence of a literal, since $\beta_2>3p-2$, it holds that $\frac{3i}{\beta_2+2}>\frac{3i-1}{\beta_2+1}>\frac{3i-2}{\beta_2}>\frac{3i-2}{\beta_2+1}>\frac{3i-2}{\beta_2+2}$ for every $i\in[p]$.
Moreover, it holds that $\frac{3i}{\beta_2+2}<\frac{3(i+1)-2}{\beta_2+2}$ for every $i\in[p-1]$ so all the values associated with ratios preferred to the initial ones are different for all variable coalitions associated with the second occurrence of a literal.
\end{itemize}
It remains to check that the ratios associated with clause or variable coalitions do not interfere with each other.
Since $\beta_2>3p\beta_1^-+4p$, it holds that the highest reachable ratio associated with a variable coalition related to the second occurrence of a literal is smaller than the smallest reachable ratio associated with a variable coalition related to the first negative occurrence of a variable, i.e., $\frac{3p}{\beta_2+2}<\frac{1}{\beta_1^-+2}$. 
Since $\beta_1^->2p\beta_1^++2p-2$, it holds that the highest reachable ratio associated with a variable coalition related to the first negative occurrence of a variable is smaller than the smallest reachable ratio associated with a variable coalition related to the first positive occurrence of a variable, i.e., $\frac{2p}{\beta_1^-+2}<\frac{1}{\beta_1^++1}$.
Since $\beta_1^+>(2p+1)\alpha+4p$, it holds that the highest reachable ratio associated with a variable coalition related to the first positive occurrence of a variable is smaller than the smallest reachable ratio associated with a clause coalition, i.e., $\frac{2p+1}{\beta_1^++2}<\frac{1}{\alpha+2}$.
Therefore, all the reachable ratios are indeed different for all clause and variable coalitions.
It follows that the deviations described in the second part of the proof of Lemma~\ref{lem:hard-converg-AHG} are the only possible ones. 
Hence the described cycle is actually the only possible one.
\end{proof}

\section{Fractional Hedonic Games}

The hardness reductions in this section are from the \np-complete problem \textsc{Exact Cover by $3$-Sets} \citep{Karp72a}. An instance of \textsc{Exact Cover by $3$-Sets} consists of a tuple $(R,S)$, where $R$ is a ground set together with a set $S$ of $3$-element subsets of $R$. A Yes-instance is an instance so that there exists a subset $S'\subseteq S$ that partitions $R$.

\symFHG*

We provide separate reductions for the two hardness results in the next lemmas.

\begin{restatable}{lemma}{existsymFHG}\label{thm:hardness-existencepath-symposFHG}
	\existpb{IS}{FHG} is \np-hard even in symmetric FHGs with non-negative weights where the initial partition is the \singleton.
\end{restatable}

\begin{proof}%
	We provide a reduction from \textsc{Exact Cover by $3$-Sets}. 
	Let $(R,S)$ be an instance of \textsc{Exact Cover by $3$-Sets}. The intuition of the proof is as follows. Elements from the ground set~$R$ are represented by gadgets corresponding to a non-negative version of the game constructed in Theorem~\ref{thm:symFHGnoIS}. 
	The sets in $S$ are represented by cliques of size $4$ where one agent is irrelevant to agents outside the clique, and the other three agents represent the three elements $r\in s$ and are linked to the respective gadgets representing $r$. 
	The correspondence occurs because all cliques corresponding to a set $s$ can simultaneously prevent cycling in all gadgets corresponding to agents in~$s$.
	
	Let us specify the construction. We may assume that every $r\in R$ occurs in at least one set of $S$. Let $m_r : = |\{s\in S\colon r\in s\}|-1\ge 0$, for $r\in R$. We define the symmetric FHG on agent set $N$, where the underlying graph consists of a $4$-clique for every set in $S$, and $m_r$ copies of a non-negative version of the example from Theorem~\ref{thm:symFHGnoIS}. Formally, $N = \bigcup_{s\in S}(\{t_s\}\cup\{s^i\colon i\in s\})\cup \bigcup_{r\in R}\bigcup_{v =1 }^{m_r} \{a_w^{r,v},b_w^{r,v},c_w^{r,v}\colon w = 1,\dots, 5\}$, and non-negative, symmetric weights are given by

	\begin{itemize}
			\item For all $r\in R$, $v\in \{1,\dots, m_r\}$, and $w\in \{1,\dots, 5\}$,
		\begin{itemize}
		\item $\util(a_w^{r,v}, b_w^{r,v}) = \util(b_w^{r,v}, c_w^{r,v}) = \util(a_w^{r,v},c_w^{r,v}) = 228$,
		\item $\util(a_w^{r,v},a_{w+1}^{r,v}) = 436$, $\util(a_w^{r,v},b_{w+1}^{r,v}) = 228$, $\util(a_w^{r,v},c_{w+1}^{r,v}) = 248$,
		\item $\util(b_w^{r,v},a_{w+1}^{r,v}) = 223$, $\util(b_w^{r,v},b_{w+1}^{r,v}) = 171$, $\util(b_w^{r,v},c_{w+1}^{r,v}) = 236$, and
		\item $\util(c_w^{r,v},a_{w+1}^{r,v}) = 223$, $\util(c_w^{r,v},b_{w+1}^{r,v}) = 171$, $\util(c_w^{r,v},c_{w+1}^{r,v}) = 188$.
		\end{itemize}
		\item $\util(t_s, s^i) = 304$, $s\in S, i\in s$,
		\item $\util(s^j, s^i) = 304$, $s\in S, i,j\in s$,
		\item $\util(s^i, a_1^{i,v}) = 304$, $s\in S, i\in s$, $v\in \{1,\dots, m_r\}$, and
		\item $\util(x,y) = 0$ for all agents $x,y\in N$ such that the weight is not defined, yet.
	\end{itemize}

In the above definition, all indices are to be read modulo $5$ (where the modulo function is assumed to map to $\{1,\dots, 5\}$). For $s\in S$, define $N^s = \{t_s\}\cup\{s^i\colon i\in s\}$.

Assume first that $(R,S)$ is a Yes-instance and let $S'\subseteq S$ be a partition of $R$. For $r\in R$, let $\sigma_r: \{s\in S\setminus S'\colon r\in s\}\to \{1,\dots, m_r\}$ be a bijection. Note that the domain and image of $\sigma_r$ have the same cardinality for every $r\in R$, because $S'$ is a partition of $R$. Consider the partition $\partition = \bigcup_{r\in R} \bigcup_{v=1}^{m_r} \{\{a_2^{r,v},b_2^{r,v},c_2^{r,v},a_3^{r,v},b_3^{r,v},c_3^{r,v}\}$, $\{a_4^{r,v},b_4^{r,v},c_4^{r,v},a_5^{r,v},b_5^{r,v},c_5^{r,v}\}, \{b_1^{r,v},c_1^{r,v}\}\}\cup \bigcup_{s\in S'} \{N^s\} \cup \bigcup_{s\in S\setminus S'}\{\{t_s\}\}\cup\{\{s^i,a_1^{i,\sigma_i(s)}\}\colon i\in s\}$. It is quickly checked that $\partition$ is IS. Moreover, $\partition$ can be reached by deviations starting from the \singleton, by forming the coalitions one by one. In particular, coalitions of the type $\{a_2^{r,v},b_2^{r,v},c_2^{r,v},a_3^{r,v},b_3^{r,v},c_3^{r,v}\}$ can be formed by having $a_3^{r,v}$ join $b_3^{r,v}$, forming a coalition that is subsequently joined by $c_3^{r,v}$, $a_2^{r,v}$, $b_2^{r,v}$, and finally $c_2^{r,v}$. Hence, it is possible to reach an IS partition with IS deviations, starting with the \singleton.

Now, assume that it is possible to reach an IS partition $\partition$ by starting the dynamics from the \singleton. Define by $G = (N,E)$ the graph with edge set $E = \{\{d,e\}\colon \util(d,e)>0\}$, a combinatorial representation of the unweighted version of the FHG under consideration. 
Note that all coalitions of $\partition$ are cliques in $G$, because all agents that get part of a coalition of size at least $2$ have positive utility and would block any further agent that does not award them positive utility. 
Now, consider a set of agents $D :=\{a_w^{r,v},b_w^{r,v},c_w^{r,v}\colon w = 1,\dots, 5\}$ for some $r\in R$, $v\in \{1,\dots, m_r\}$. 
Assume for contradiction that for all agents $d\in D$, $\partition(d)\subseteq D$. This yields an IS partition of the game considered in Theorem~\ref{thm:symFHGnoIS} because the agents would also form an IS partition in the game of Theorem~\ref{thm:symFHGnoIS}. This is due to the fact that no agents with mutual negative utility would form a coalition, and a deviation with negative weights would still be a deviation if these weights are set to $0$. This is a contradiction. Hence, some agent in $D$ forms a coalition with an agent outside $D$. By the fact that all coalitions in $\partition$ are cliques in $G$, the only such agent can be $a_1^{r,v}$. By the same fact, $\partition(a_1^{r,v})\cap D = \{a_1^{r,v}\}$ and there exists a unique $s\in S$ with $r\in s$ such that $\partition(a_1^{r,v}) = \{a_1^{r,v}, s^r\}$.

Next, let $s\in S$. We claim that $\{t^s\}\in \partition$ or $N^s\in \partition$. Otherwise, consider $r\in s$ with $s^r\notin \partition(t^s)$. Again by the clique property, $\util_{s^r}(\partition)\le \frac {304}2$ and $\partition(t^s)\subseteq N^s$. Hence, $\util_{s^r}(\partition(t^s)\cup\{s^r\}) \ge \frac {608}3$, and every agent in $\partition(t^s)$ would welcome $s^r$. This contradicts the individual stability of $\partition$.

Consider the set $T = \{s\in S\colon \{t^s\}\in\partition\}$. Then, for every $s\in T$ and $r\in s$, there exists $v\in \{1,\dots, m_r\}$ with $\partition(s^r) = \{a_1^{r,v}, s^r\}$. Otherwise, $\partition(s^r)\subseteq N^s$ and $t^s$ can perform a deviation by joining $\partition(s^r)$. Hence, the sets in $T$ cover every element in $r\in R$ exactly $m_r$ times (in order to form all the required coalitions of the type $\{a_1^{r,v}, s^r\}$). Since $S$ covers every $r\in R$ exactly $m_r + 1$ times, the set $S' = S\setminus T$ forms a partition of $R$. Hence, $(R,S)$ is a Yes-instance.
\end{proof}

\begin{restatable}{lemma}{convergsymFHG}\label{thm:hardness-convergence-symFHG}
\convergpb{IS}{FHG} is \conp-hard even in symmetric FHGs with non-negative weights.
\end{restatable}

\begin{proof}%
For this purpose, we prove the \np-hardness of the complement problem, which asks whether there exists a cycle of IS deviations. 
	We provide a reduction from \textsc{Exact Cover by $3$-Sets}.

	Let $(R,S)$ be an instance of \textsc{Exact Cover by $3$-Sets}. Let $l = |S| - |R|/3$. Choose $\alpha$ with a polynomial-size representation in the input size satisfying $\frac l{l+1}\alpha < 152 < \frac {l+1}{l+2} \alpha$. For the reduction to work, any number satisfying these boundaries suffices, and for a polynomial-size representation, one can for example use $\alpha = \frac 12 \cdot 152\left(\frac {l+1}l + \frac{l+2}{l+1}\right)$.

	Define the symmetric FHG on agent set $N$ where $N = R \cup \{r^s\colon s\in S, r\in s\} \cup \{s_1,s_2\colon s\in S\} \cup \{a_w,b_w,c_w\colon w = 1,\dots, 5\}$. We define $C = \{a_w,b_w,c_w\colon w = 1,\dots, 5\}$. The utilities are given as follows.

	\begin{itemize}
			\item For all $w\in \{1,\dots, 5\}$, reading indices modulo $5$ (where the modulo function is assumed to map to $\{1,\dots, 5\}$),
		\begin{itemize}
			\item $\util(a_w, b_w) = \util(b_w, c_w) = \util(a_w,c_w) = 228$,
			\item $\util(a_w,a_{w+1}) = 436$, $\util(a_w,b_{w+1}) = 228$, $\util(a_w,c_{w+1}) = 248$,
			\item $\util(b_w,a_{w+1}) = 223$, $\util(b_w,b_{w+1}) = 171$, $\util(b_w,c_{w+1}) = 236$, and
			\item $\util(c_w,a_{w+1}) = 223$, $\util(c_w,b_{w+1}) = 171$, $\util(c_w,c_{w+1}) = 188$.
		\end{itemize}
		\item For all $s\in S$,
		\begin{itemize}
			\item $\util(a_1, s_2) = \util(s_1, s_2) = \alpha$,
			\item $\util(s_1,r^s) = \alpha$, $\util(r^s,r) = 2\alpha$, $r\in S$, and
		\end{itemize}
		\item $\util(x,y) = 0$ for all agents $x,y\in N$ such that the weight is not defined, yet.
	\end{itemize}

	 Finally, define $\partition = \{\{r\}\colon r\in R\}\cup \{\{s_1,i^s,j^s,k^s\}\colon \{i,j,k\} = s\in S\} \cup \{\{a_1\}\cup\{s_2\colon s\in S\}\}\cup \{\{b_1, c_1\},\{a_2,b_2,c_2,a_3,b_3,c_3\},\{a_4,b_4,c_4,a_5,b_5,c_5\}\}$. The reduction is illustrated in Figure~\ref{fig:hardness_cycSymFHG}. 

	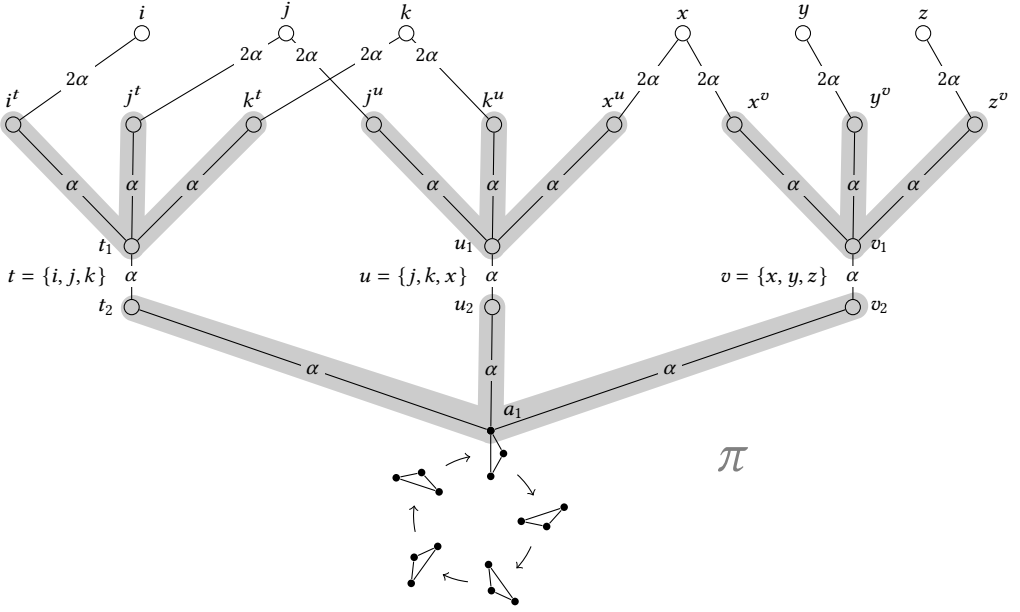
\begin{figure}
		\centering
		\begin{tikzpicture}[auto,scale = 0.8,
			complexnode/.pic={
				\node[indicatorvertex] (b) at (0,-0.05){};
				\node[indicatorvertex] (a) at (0,.55){};
				\node[indicatorvertex] (c) at (0.17,0.25){};
				\path (a) edge (b);
				\path (b) edge (c);
				\path (c) edge (a);},
				complexarc_prev/.pic={
				\node (s) at (-.65,0.3){};
				\node (t) at (.65,0.3){};
				\draw[bend left,->]  (s) to (t);},
				complexarc/.pic={
			\draw[->] (-.15,0.15) arc[radius= 1cm, start angle = 100, end angle= 80];}
			]

		\pgfmathsetmacro{\yshift}{0.11}
		\pgfmathsetmacro{\yshiftt}{0.08}
		\pgfmathsetmacro{\tinyshift}{0.03}

		\node[protovertex,label = 90:\footnotesize $i^t$] (i1s) at (0.86,0.5){};

		\node[protovertex,label = 90:\footnotesize $j^t$] (j1s) at (2.86,0.5){};

		\node[protovertex,label = 90:\footnotesize $k^t$] (k1s) at (4.86,0.5){};

		\node[protovertex,label = 90:\footnotesize $j^u$] (j1t) at (6.86,0.5){};

		\node[protovertex,label = 90:\footnotesize $k^u$] (k1t) at (8.86,0.5){};

		\node[protovertex,label = 90:\footnotesize $x^u$] (x1t) at (10.86,0.5){};

		\node[protovertex,label = 60:\footnotesize $x^v$] (x1u) at (12.86,0.5){};

		\node[protovertex,label = 60:\footnotesize $y^v$] (y1u) at (14.86,0.5){};

		\node[protovertex,label = 60:\footnotesize $z^v$] (z1u) at (16.86,0.5){};

		\node[protovertex,label = 180:{\footnotesize $t_1$}] (s) at (2.83,-1.5){};
		\node[protovertex,label = 180:{\footnotesize $u_1$}] (t) at (8.83,-1.5){};
		\node[protovertex,label = 0:{\footnotesize $v_1$}] (u) at (14.83,-1.5){};
		\node[protovertex,label = 180:{\footnotesize $t_2$}] (s2) at (2.83,-2.5){};
		\node[protovertex,label = 180:{\footnotesize $u_2$}] (t2) at (8.83,-2.5){};
		\node[protovertex,label = 0:{\footnotesize $v_2$}] (u2) at (14.83,-2.5){};

		\draw (s) edge node[fill=white,anchor=center, pos=0.5, inner sep =2pt] {\footnotesize $\alpha$} node[left,midway, xshift = -.2cm] {\footnotesize $t= \{i,j,k\}$} (s2) ;
		\draw (t) edge node[fill=white,anchor=center, pos=0.5, inner sep =2pt] {\footnotesize $\alpha$} node[left,midway, xshift = -.2cm] {\footnotesize $u= \{j,k,x\}$} (t2) ;
		\draw (u) edge node[fill=white,anchor=center, pos=0.5, inner sep =2pt] {\footnotesize $\alpha$} node[left,midway, xshift = -.2cm] {\footnotesize $v= \{x,y,z\}$} (u2) ;

		\foreach \x/\y in {i/s,j/s,k/s,j/t,k/t,x/t,x/u,y/u,z/u}{
		\draw (\y) edge node[fill=white,anchor=center, pos=0.5, inner sep =2pt] {\footnotesize $\alpha$} (\x1\y);
		}

		\node[protovertex,label = {[shift={(0,-0.05)}]90:{\footnotesize $j$}}] (j1) at (5.4,2) {};
		\node[protovertex,label = {[shift={(0,-0.05)}]90:{\footnotesize $k$}}] (k1) at (7.4,2){};
		\node[protovertex,label = {[shift={(0,-0.05)}]90:{\footnotesize $x$}}] (x1) at (12,2){};
		\node[protovertex,label = {[shift={(0,-0.05)}]90:{\footnotesize $i$}}] (i1) at (3,2){};
		\node[protovertex,label = {[shift={(0,-0.05)}]90:{\footnotesize $y$}}] (y1) at (14,2){};
		\node[protovertex,label = {[shift={(0,-0.05)}]90:{\footnotesize $z$}}] (z1) at (16,2){};

		\draw (j1s) edge  node[fill=white,anchor=center, pos=0.8, inner sep =2pt] {\footnotesize $2\alpha$} (j1);
		\draw (j1t) edge  node[fill=white,anchor=center, pos=0.8, inner sep =2pt] {\footnotesize $2\alpha$} (j1);
		\draw (k1s) edge  node[fill=white,anchor=center, pos=0.8, inner sep =2pt] {\footnotesize $2\alpha$} (k1);
		\draw (k1t) edge  node[fill=white,anchor=center, pos=0.8, inner sep =2pt] {\footnotesize $2\alpha$} (k1);
		\draw (x1t) edge  node[fill=white,anchor=center, pos=0.5, inner sep =2pt] {\footnotesize $2\alpha$} (x1);
		\draw (x1u) edge  node[fill=white,anchor=center, pos=0.5, inner sep =2pt] {\footnotesize $2\alpha$} (x1);

		\draw (i1s) edge  node[fill=white,anchor=center, pos=0.5, inner sep =2pt] {\footnotesize $2\alpha$} (i1);
		\draw (y1u) edge  node[fill=white,anchor=center, pos=0.5, inner sep =2pt] {\footnotesize $2\alpha$} (y1);
		\draw (z1u) edge  node[fill=white,anchor=center, pos=0.5, inner sep =2pt] {\footnotesize $2\alpha$} (z1);

		\coordinate (sr) at ($(s)+(.18,-0.12)$){};
		\coordinate (sl) at ($(s)+(-.18,-0.12)$){};
		\coordinate (sa1) at ($(s)+(0.216,0.5)$){};
		\coordinate (sa2) at ($(s)+(-0.216,0.5)$){};

		\coordinate (ksr) at ($(k1s)+(.18,-0.12)$){};
		\coordinate (jsr) at ($(j1s)+(0.216,0)$){};
		\coordinate (isr) at ($(i1s)+(.18,0.12)$){};

		\filldraw[draw = none, fill opacity=0.2,fill] (sr) -- (ksr) arc (-33:147:.22cm) -- (sa1) -- (jsr) arc (0:180:0.216cm) -- (sa2) -- (isr) arc (33:213:0.216cm)-- (sl) arc (213:327:0.215cm) (sr);

		\coordinate (tr) at ($(t)+(.18,-0.12)$){};
		\coordinate (tl) at ($(t)+(-.18,-0.12)$){};
		\coordinate (ta1) at ($(t)+(0.216,0.5)$){};
		\coordinate (ta2) at ($(t)+(-0.216,0.5)$){};

		\coordinate (xtr) at ($(x1t)+(.18,-0.12)$){};
		\coordinate (ktr) at ($(k1t)+(0.216,0)$){};
		\coordinate (jtr) at ($(j1t)+(.18,0.12)$){};

		\filldraw[draw = none, fill opacity=0.2,fill] (tr) -- (xtr) arc (-33:147:.22cm) -- (ta1) -- (ktr) arc (0:180:0.216cm) -- (ta2) -- (jtr) arc (33:213:0.216cm)-- (tl) arc (213:327:0.215cm) (tr);

		\coordinate (ur) at ($(u)+(.18,-0.12)$){};
		\coordinate (ul) at ($(u)+(-.18,-0.12)$){};
		\coordinate (ua1) at ($(u)+(0.216,0.5)$){};
		\coordinate (ua2) at ($(u)+(-0.216,0.5)$){};

		\coordinate (zur) at ($(z1u)+(.18,-0.12)$){};
		\coordinate (yur) at ($(y1u)+(0.216,0)$){};
		\coordinate (xur) at ($(x1u)+(.18,0.12)$){};

		\filldraw[draw = none, fill opacity=0.2,fill] (ur) -- (zur) arc (-33:147:.22cm) -- (ua1) -- (yur) arc (0:180:0.216cm) -- (ua2) -- (xur) arc (33:213:0.216cm)-- (ul) arc (213:327:0.215cm) (ur);

			\node[regular polygon,regular polygon sides=10,minimum size=1.3cm] (p) at (8.55,-6){};
			\draw (p.corner 1) pic {complexnode};
		\node[indicatorvertex, label = 45:{\footnotesize $a_1$}] (b1) at (a) {};
		\draw (b1) edge node[fill=white,anchor=center, pos=0.5, inner sep =2pt] {\footnotesize $\alpha$} (s2);
		\draw (b1) edge node[fill=white,anchor=center, pos=0.5, inner sep =2pt] {\footnotesize $\alpha$} (t2);
		\draw (b1) edge node[fill=white,anchor=center, pos=0.5, inner sep =2pt] {\footnotesize $\alpha$} (u2);		%
			\draw (p.corner 2) pic[rotate = 20] {complexarc};
			\draw (p.corner 3) pic[rotate = 72] {complexnode};		%
			\draw (p.corner 4) pic[rotate = 92] {complexarc};
			\draw (p.corner 5) pic[rotate = 144] {complexnode};		%
			\draw (p.corner 6) pic[rotate = 164] {complexarc};
			\draw (p.corner 7) pic[rotate = 216] {complexnode};		%
			\draw (p.corner 8) pic[rotate = 236] {complexarc};
			\draw (p.corner 9) pic[rotate = 288] {complexnode};		%
			\draw (p.corner 10) pic[rotate = 308] {complexarc};

		\coordinate (br) at ($(b1)+(.12,-0.18)$){};
		\coordinate (bl) at ($(b1)+(-.12,-0.18)$){};
		\coordinate (ba1) at ($(b1)+(0.216,0.42)$){};
		\coordinate (ba2) at ($(b1)+(-0.216,0.42)$){};
		\coordinate (ba3) at ($(b1)+(-0.648,0.216)$){};
		\coordinate (ba4) at ($(b1)+(-0.648,-0.216)$){};

		\coordinate (uur) at ($(u2)+(.12,-0.18)$){};
		\coordinate (tur) at ($(t2)+(0.216,0)$){};
		\coordinate (sur) at ($(s2)+(.12,0.18)$){};
		\coordinate (aao) at ($(a1)+(0,0.216)$){};
		\coordinate (aat) at ($(a2)+(-.18,0.12)$){};

		\filldraw[draw = none, fill opacity=0.2,fill] (br) -- (uur) arc (-67:113:.22cm) -- (ba1) -- (tur) arc (0:180:0.216cm) -- (ba2) -- (sur) arc (67:247:0.216cm)-- (bl) arc(238:312:0.218cm) -- (br);

		\node[opacity = 0.5] at (12.83,-5) {\Huge$\partition$};
	\end{tikzpicture}
		\caption[Cycling on simple asymmetric FHGs]{Schematic of the symmetric FHG of the hardness construction in Lemma~\ref{thm:hardness-convergence-symFHG}. The figure is based on the instance $(\{i,j,k,x,y,z\},\{t,u,v\})$ with $t=\{i,j,k\}$, $u=\{j,k,x\}$, and $v=\{x,y,z\}$. The non-singleton coalitions above $a_1$ of the initial partition $\partition$ are depicted in gray. The only possibility for $a_1$ to deviate is if two of $t_2$, $u_2$, or $v_2$ perform a deviation, which in turn can only happen if the coalition partners of their respective counterparts $t_1$, $u_1$, or $v_1$ have been deviating before.\label{fig:hardness_cycSymFHG}}
	\end{figure}

We claim that $(R,S)$ is a Yes-instance if and only if the IS dynamics starting with $\partition$ can cycle.

First assume that $(R,S)$ is a Yes-instance and let $S'\subseteq S$ be a partition of $R$ by the sets in $S$. We consider three stages of deviations. In the first stage, the agents in a coalition with some $s_1$ for $s\in S'$ join the agents of type $r_v$.
This will leave all agents in $\{s_1\colon s\in S'\}$ in singleton coalitions. In the second stage, agents $s_2$ for $s\in S'$ join their copies $s_1$. This leaves the agent $a_1$ with a utility of $\frac l{l+1}\alpha < 152 = \util_{a_1}(\{a_1,b_1,c_1\})$. Therefore, we can have $a_1$ join $\{b_1,c_1\}$. From now on, we consider the subgame induced by the agents in $C$. In this subgame, the partition is currently $\sigma = \{\{a_1,b_1, c_1\},\{a_2,b_2,c_2,a_3,b_3,c_3\},\{a_4,b_4,c_4,a_5,b_5,c_5\}\}$. We can cycle indefinitely by letting the agents of one coalition of size $6$ in $\sigma$ join the coalition of size $3$. More exactly, let $a_5$, $b_5$, and $c_5$ join $\{a_1,b_1,c_1\}$. Then, we reach the partition $\{\{a_1,b_1, c_1,a_5,b_5,c_5\},\{a_2,b_2,c_2,a_3,b_3,c_3\},\{a_4,b_4,c_4\}\}$ which only differs from $\sigma$ by an index shift. Hence, we can repeat the same triplet of deviations indefinitely.

Conversely, assume that there exists an infinite sequence of deviations starting from $\partition$. Agents of the type $r^s$ can perform at most one deviation joining the agent $r$ if she is still in a singleton coalition. After this deviation, they land in a coalition that cannot be altered anymore. Therefore, agents of the type $r$ for $r\in R$ will never deviate, because they cannot receive positive utility, unless joining an agent of the type $r^s$, which will never leave her coalition with $s_1$ unless joining $r$. Agents of the type $s_1$ will never perform a deviation, because every agent that leaves her coalition can never be joined again, and the agent $s_2$ can only perform a deviation by joining $s_1$. In turn, agents of the type $s_2$ can only deviate if their copy $s_1$ is forced into a singleton coalition. At this point, they can deviate exactly once, forming a coalition that can never be changed again.

Agents in $C\setminus \{a_1\}$ can only perform a deviation after $a_1$ has performed a deviation. Thus, the only possibility for an infinite length of deviations is if $a_1$ performs a deviation. Since $a_1$ cannot join the coalition of agents of the type $s_2$ again, once they left her coalition, the only possible deviation is by joining the coalition $\{b_1,c_1\}$, obtaining a utility of $152$. The utility of $a_1$ for any subset $C\subseteq \partition(a_1)$ that can arise as her coalition before she deviated for the first time is $\util_{a_1}(C) = \frac h{1+h}\alpha$ for $h = C\cap \{s_2\colon s\in S\}$. 
It follows that $a_1$ can only deviate once all except $l$ agents of the type $s_2$ have left her coalition. 

Now let $\partition'$ be the partition right before the first deviation of $a_1$ and define $S' = \{s\in S\colon s_2\in \partition'(s_1)\}$. 
Then, $S'$ consists of exactly $|R|/3$ elements. 
Recall our discussion of which deviations must have occurred such that agents of the type $s_2$ can perform a deviation. For this, all of the agents $r^s$ for $r\in s$ must have deviated. Since the agent $r$ can be joined by at most one such agent, the sets in $S'$ must be disjoint. Hence, the only way that all except $l$ agents of type $s_2$ have left $\partition(a_1)$ is if $S'$ covers precisely the elements of $R$. In total, $S'$ forms a partition of $R$. Consequently, $(R,S)$ is a Yes-instance.
\end{proof}

\hardExPathAsymFHG*

We prove the two hardness results by providing separate reductions for each problem in the next two lemmas.

\begin{lemma}\label{lem:ISnpinasymFHG}
\existpb{IS}{FHG} is \np-hard even in simple asymmetric FHGs.
\end{lemma}

\begin{proof}
	We provide a reduction from \textsc{Exact Cover by $3$-Sets}.

	Let $(R,S)$ be an instance of \textsc{Exact Cover by $3$-Sets}. We may assume that every $r\in R$ occurs in at least one set of $S$. Let $m_r : = |\{s\in S\colon r\in s\}|-1\ge 0$, and $l = |S| - |R|/3$. Define the simple asymmetric FHG based on the directed graph $G = (V,A)$, where $V = \bigcup_{r\in R}\{r_1,\dots, r_{m_r}\} \cup S\cup \bigcup_{s\in S} \{r^s\colon r\in s\}\cup \bigcup_{v = 1}^l \{a_1^v,a_2^v,a_3^v\}$ and $A = \bigcup_{s\in S}(\{(s,r^s),(r^s,r_1),\dots,(r^s,r_{m_r})\colon r\in S\}\cup \{(s,a_1^1),\dots,(s,a_1^l)\})\cup \bigcup_{v=1}^l \{(a_1^v,a_2^v),(a_2^v,a_3^v),(a_3^v,a_1^v)\}$. 
	
	Finally, define the partition $\partition = \bigcup_{a\in V\setminus (S\cup \{r^s\colon s\in S, r\in s\})} \{\{a\}\}\cup \{\{s,i^s,j^s,k^s\}\colon \{i,j,k\} = s\in S\}$. 
	The reduction is illustrated in Figure~\ref{fig:hardness_convSimpleAsymFHG}. 
	It depicts the simple asymmetric directed graph corresponding to a small source instance together with the associated initial partition.

	\begin{figure*}
		\centering
		\begin{tikzpicture}[auto,scale = 0.8]

		\pgfmathsetmacro{\yshift}{0.11}
		\pgfmathsetmacro{\yshiftt}{0.08}
		\pgfmathsetmacro{\tinyshift}{0.03}

		\node[protovertex,label = 90:\footnotesize $i^t$] (i1s) at (0.86,0.5){};

		\node[protovertex,label = 90:\footnotesize $j^t$] (j1s) at (2.86,0.5){};

		\node[protovertex,label = 90:\footnotesize $k^t$] (k1s) at (4.86,0.5){};

		\node[protovertex,label = 90:\footnotesize $j^u$] (j1t) at (6.86,0.5){};

		\node[protovertex,label = 90:\footnotesize $k^u$] (k1t) at (8.86,0.5){};

		\node[protovertex,label = 90:\footnotesize $x^u$] (x1t) at (10.86,0.5){};

		\node[protovertex,label = 90:\footnotesize $x^v$] (x1u) at (12.86,0.5){};

		\node[protovertex,label = 90:\footnotesize $y^v$] (y1u) at (14.86,0.5){};

		\node[protovertex,label = 90:\footnotesize $z^v$] (z1u) at (16.86,0.5){};

		\node (p2) at (0,0){};
		\node (p3) at (0,1){};
		\node (p1) at (0.86,0.5){};

		\node[protovertex,label = 180:{\footnotesize $\{i,j,k\} = t$}] (s) at (2.83,-1.5){};
		\node[protovertex,label = 180:{\footnotesize $\{j,k,x\} = u$}] (t) at (8.83,-1.5){};
		\node[protovertex,label = 0:{\footnotesize $v = \{x,y,z\}$}] (u) at (14.83,-1.5){};

		\foreach \x/\y/\z in {i/s/0,j/s/1.6,k/s/3.2,j/t/4.8,k/t/6.4,x/t/8,x/u/9.6,y/u/11.2,z/u/12.8}{
		\draw (\y) edge[->] (\x1\y);
		}

		\node[protovertex,label = {[shift={(-.1,-0.05)}]45:{\footnotesize $a_1^1$}}] (a1) at (9.25,-3){};
		\node[protovertex,label = 180:{\footnotesize $a_2^1$}] (a2) at (8.85,-3.75){};
		\node[protovertex,label = 0:{\footnotesize $a_3^1$}] (a3) at (9.65,-3.75){};

		\draw (a1) edge[<-] (s);
		\draw (a1) edge[<-] (t);
		\draw (a1) edge[<-] (u);
		\draw (a1) edge[->] (a2);
		\draw (a2) edge[->] (a3);
		\draw (a3) edge[->] (a1);

		\node at (12.7,-3.38) {$|S|-\frac{|R|}3$ many};

		\node[protovertex,label = {[shift={(0,-0.05)}]90:{\footnotesize $j_1$}}] (j1) at (5.4,1.7) {};
		\node[protovertex,label = {[shift={(0,-0.05)}]90:{\footnotesize $k_1$}}] (k1) at (7.4,1.7){};
		\node[protovertex,label = {[shift={(0,-0.05)}]90:{\footnotesize $x_1$}}] (x1) at (12,1.7){};

		\draw (j1s) edge[->] (j1);
		\draw (j1t) edge[->] (j1);
		\draw (k1s) edge[->] (k1);
		\draw (k1t) edge[->] (k1);
		\draw (x1t) edge[->] (x1);
		\draw (x1u) edge[->] (x1);

		\coordinate (sr) at ($(s)+(.18,-0.12)$){};
		\coordinate (sl) at ($(s)+(-.18,-0.12)$){};
		\coordinate (sa1) at ($(s)+(0.216,0.5)$){};
		\coordinate (sa2) at ($(s)+(-0.216,0.5)$){};

		\coordinate (ksr) at ($(k1s)+(.18,-0.12)$){};
		\coordinate (jsr) at ($(j1s)+(0.216,0)$){};
		\coordinate (isr) at ($(i1s)+(.18,0.12)$){};

		\filldraw[draw = none, fill opacity=0.2,fill] (sr) -- (ksr) arc (-33:147:.22cm) -- (sa1) -- (jsr) arc (0:180:0.216cm) -- (sa2) -- (isr) arc (33:213:0.216cm)-- (sl) arc (213:327:0.215cm) (sr);

		\coordinate (tr) at ($(t)+(.18,-0.12)$){};
		\coordinate (tl) at ($(t)+(-.18,-0.12)$){};
		\coordinate (ta1) at ($(t)+(0.216,0.5)$){};
		\coordinate (ta2) at ($(t)+(-0.216,0.5)$){};

		\coordinate (xtr) at ($(x1t)+(.18,-0.12)$){};
		\coordinate (ktr) at ($(k1t)+(0.216,0)$){};
		\coordinate (jtr) at ($(j1t)+(.18,0.12)$){};

		\filldraw[draw = none, fill opacity=0.2,fill] (tr) -- (xtr) arc (-33:147:.22cm) -- (ta1) -- (ktr) arc (0:180:0.216cm) -- (ta2) -- (jtr) arc (33:213:0.216cm)-- (tl) arc (213:327:0.215cm) (tr);

		\coordinate (ur) at ($(u)+(.18,-0.12)$){};
		\coordinate (ul) at ($(u)+(-.18,-0.12)$){};
		\coordinate (ua1) at ($(u)+(0.216,0.5)$){};
		\coordinate (ua2) at ($(u)+(-0.216,0.5)$){};

		\coordinate (zur) at ($(z1u)+(.18,-0.12)$){};
		\coordinate (yur) at ($(y1u)+(0.216,0)$){};
		\coordinate (xur) at ($(x1u)+(.18,0.12)$){};

		\filldraw[draw = none, fill opacity=0.2,fill] (ur) -- (zur) arc (-33:147:.22cm) -- (ua1) -- (yur) arc (0:180:0.216cm) -- (ua2) -- (xur) arc (33:213:0.216cm)-- (ul) arc (213:327:0.215cm) (ur);

		\node[opacity = 0.5] at (11.83,-1.5) {\Huge$\partition$};
	\end{tikzpicture}
		\caption[Convergence on simple asymmetric FHGs]{Schematic of the simple asymmetric FHG of the hardness construction in Lemma~\ref{lem:ISnpinasymFHG}. The figure is based on the instance $(\{i,j,k,x,y,z\},\{t,u,v\})$ with $t=\{i,j,k\}$, $u=\{j,k,x\}$, and $v=\{x,y,z\}$. The non-singleton coalitions of the initial partition $\partition$ are depicted in gray.\label{fig:hardness_convSimpleAsymFHG}}
	\end{figure*}
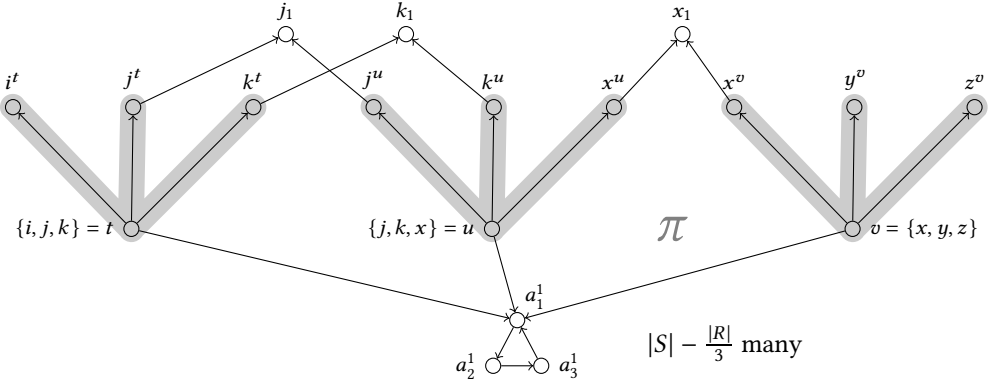

We claim that $(R,S)$ is a Yes-instance if and only if the IS dynamics starting with $\partition$ can converge.

Assume first that $(R,S)$ is a Yes-instance and let $S'\subseteq S$ be a partition of $R$ by sets in $S$. Consider the following deviations. First, the agents in the set $\bigcup_{s\in S\setminus S'} \{r^s\colon r\in s\}$ join one by one the agents in $\bigcup_{r\in R}\{r_1,\dots, r_{m_r}\}$ to end up in coalitions of size $2$. Since $S'$ covers every element of $R$ exactly once, this step can be performed. Next, the agents $\{s\in S\setminus S'\}$ join the agents $\{a_1^1,\dots, a_1^l\}$ in an arbitrary bijective way. Finally, agents $a_2^v$ join agents $a_3^v$. It is quickly checked that the resulting partition is IS.

Conversely, assume that there exists a converging sequence of deviations starting with the partition $\partition$ and terminating in partition $\partition^*$. Then, one agent of every set $\{a_1^v,a_2^v,a_3^v\}$ must form a coalition with an agent outside of this set. The only possibility for this is if $a_1^v$ is joined by an agent of type $s$ corresponding to a set in $s\in S$. Every such agent can only perform a deviation if all the other agents in her initial coalition have deviated before. 
Similar to the proof of Lemma~\ref{thm:hardness-convergence-symFHG}, each agent of the type $r_j$ for $r\in R$ and $j\in [m_r]$ can only joined by exactly one agent of the type $r^s$, while this is the only possible deviation that agents of the type $r_j$ can do.
Hence, the only possibility that $l$ agents of the type $s$ deviate to break cycling is if they correspond to $l$ sets from $S$ which cover each element $r\in R$ for $m_r$ times. Hence, the remaining sets in $S$ form exactly a partition of the elements in $R$. In other words, the set $S' = \{s\in S\colon \partition(s)=\partition^*(s)\}$ forms a partition of $R$. 
Hence, $(R,S)$ is a Yes-instance.
\end{proof}

\begin{lemma}\label{lem:IScoNPinFHG}
\convergpb{IS}{FHG} is \conp-hard even in simple asymmetric FHGs.
\end{lemma}

\begin{proof}
For this purpose, we prove the \np-hardness of the complement problem, which asks whether there exists a cycle of IS deviations. 
	We provide a reduction from \textsc{Exact Cover by $3$-Sets}.

	Let $(R,S)$ be an instance of \textsc{Exact Cover by $3$-Sets}. We may assume that every $r\in R$ occurs in at least one set of $S$. Let $m_r : = |\{s\in S\colon r\in s\}|-1\ge 0$, and $l = |R|/3$. Define the simple asymmetric FHG based on the graph $G = (V,A)$, where $V = \{r_1,\dots, r_{m_r}\colon r\in R\} \cup \{r^s\colon s\in S, r\in s\} \cup \{s_1,s_2\colon s\in S\} \cup \{b_1,b_2,b_3\} \cup \{a_1,\dots, a_l\}$, and $A = \bigcup_{s\in S} (\{(r^s,r_1),\dots, (r^s,r_{m_s}),(s_1,r^s) \colon r\in s\} \cup \{(s_1,s_2)$,$(b_1,s_2)\}) \cup \{(a_v,b_1)\colon v = 1,\dots, l\} \cup \{(b_2,b_3),(b_3,b_1)\}$. 
	
	Finally, define $\partition = \{\{r_1\},\dots,\{r_{m_r}\}\colon r\in R\}\cup \{\{s_1,i^s,j^s,k^s\}\colon \{i,j,k\} = s\in S\} \cup \{\{b_1\}\cup\{s_2\colon s\in S\}\cup \{a_1,\dots, a_l\}\}\cup \{\{b_2\},\{b_3\}\}$. 
	The reduction is illustrated in Figure~\ref{fig:hardness_cycSimpleAsymFHG}. 

	\begin{figure*}
		\centering
		\begin{tikzpicture}[auto,scale = 0.8]

		\pgfmathsetmacro{\yshift}{0.11}
		\pgfmathsetmacro{\yshiftt}{0.08}
		\pgfmathsetmacro{\tinyshift}{0.03}

		\node[protovertex,label = 90:\footnotesize $i^t$] (i1s) at (0.86,0.5){};

		\node[protovertex,label = 90:\footnotesize $j^t$] (j1s) at (2.86,0.5){};

		\node[protovertex,label = 90:\footnotesize $k^t$] (k1s) at (4.86,0.5){};

		\node[protovertex,label = 90:\footnotesize $j^u$] (j1t) at (6.86,0.5){};

		\node[protovertex,label = 90:\footnotesize $k^u$] (k1t) at (8.86,0.5){};

		\node[protovertex,label = 90:\footnotesize $x^u$] (x1t) at (10.86,0.5){};

		\node[protovertex,label = 90:\footnotesize $x^v$] (x1u) at (12.86,0.5){};

		\node[protovertex,label = 90:\footnotesize $y^v$] (y1u) at (14.86,0.5){};

		\node[protovertex,label = 90:\footnotesize $z^v$] (z1u) at (16.86,0.5){};

		\node[protovertex,label = 180:{\footnotesize $t_1$}] (s) at (2.83,-1.5){};
		\node[protovertex,label = 180:{\footnotesize $u_1$}] (t) at (8.83,-1.5){};
		\node[protovertex,label = 180:{\footnotesize $v_1$}] (u) at (14.83,-1.5){};
		\node[protovertex,label = 180:{\footnotesize $t_2$}] (s2) at (2.83,-2.5){};
		\node[protovertex,label = 180:{\footnotesize $u_2$}] (t2) at (8.83,-2.5){};
		\node[protovertex,label = 180:{\footnotesize $v_2$}] (u2) at (14.83,-2.5){};

		\draw (s) edge[<-] node[right,midway] {\footnotesize $t= \{i,j,k\}$} (s2) ;
		\draw (t) edge[<-] node[right,midway] {\footnotesize $u= \{j,k,x\}$} (t2) ;
		\draw (u) edge[<-] node[right,midway] {\footnotesize $v= \{x,y,z\}$} (u2) ;

		\foreach \x/\y in {i/s,j/s,k/s,j/t,k/t,x/t,x/u,y/u,z/u}{
		\draw (\y) edge[->] (\x1\y);
		}

		\node[protovertex,label = {[shift={(-.1,-0.05)}]45:{\footnotesize $b_1$}}] (b1) at (8.83,-4.5){};
		\node[protovertex,label = 0:{\footnotesize $b_2$}] (b2) at (9.33,-5.33){};
		\node[protovertex,label = 0:{\footnotesize $b_3$}] (b3) at (9.83,-4.5){};

		\node[protovertex,label = 180:{\footnotesize $a_1$}] (a1) at (7.03,-4.5){};
		\node[protovertex,label = 180:{\footnotesize $a_2$}] (a2) at (7.33,-5.5){};

		\draw (b1) edge[->] (s2);
		\draw (b1) edge[->] (t2);
		\draw (b1) edge[->] (u2);
		\draw (b1) edge[->] (b2);
		\draw (b2) edge[->] (b3);
		\draw (b3) edge[->] (b1);
		\draw (b1) edge[<-] (a1);
		\draw (b1) edge[<-] (a2);

		\node[protovertex,label = {[shift={(0,-0.05)}]90:{\footnotesize $j_1$}}] (j1) at (5.4,1.7) {};
		\node[protovertex,label = {[shift={(0,-0.05)}]90:{\footnotesize $k_1$}}] (k1) at (7.4,1.7){};
		\node[protovertex,label = {[shift={(0,-0.05)}]90:{\footnotesize $x_1$}}] (x1) at (11.8,1.7){};

		\draw (j1s) edge[->] (j1);
		\draw (j1t) edge[->] (j1);
		\draw (k1s) edge[->] (k1);
		\draw (k1t) edge[->] (k1);
		\draw (x1t) edge[->] (x1);
		\draw (x1u) edge[->] (x1);

		\coordinate (sr) at ($(s)+(.18,-0.12)$){};
		\coordinate (sl) at ($(s)+(-.18,-0.12)$){};
		\coordinate (sa1) at ($(s)+(0.216,0.5)$){};
		\coordinate (sa2) at ($(s)+(-0.216,0.5)$){};

		\coordinate (ksr) at ($(k1s)+(.18,-0.12)$){};
		\coordinate (jsr) at ($(j1s)+(0.216,0)$){};
		\coordinate (isr) at ($(i1s)+(.18,0.12)$){};

		\filldraw[draw = none, fill opacity=0.2,fill] (sr) -- (ksr) arc (-33:147:.22cm) -- (sa1) -- (jsr) arc (0:180:0.216cm) -- (sa2) -- (isr) arc (33:213:0.216cm)-- (sl) arc (213:327:0.215cm) (sr);

		\coordinate (tr) at ($(t)+(.18,-0.12)$){};
		\coordinate (tl) at ($(t)+(-.18,-0.12)$){};
		\coordinate (ta1) at ($(t)+(0.216,0.5)$){};
		\coordinate (ta2) at ($(t)+(-0.216,0.5)$){};

		\coordinate (xtr) at ($(x1t)+(.18,-0.12)$){};
		\coordinate (ktr) at ($(k1t)+(0.216,0)$){};
		\coordinate (jtr) at ($(j1t)+(.18,0.12)$){};

		\filldraw[draw = none, fill opacity=0.2,fill] (tr) -- (xtr) arc (-33:147:.22cm) -- (ta1) -- (ktr) arc (0:180:0.216cm) -- (ta2) -- (jtr) arc (33:213:0.216cm)-- (tl) arc (213:327:0.215cm) (tr);

		\coordinate (ur) at ($(u)+(.18,-0.12)$){};
		\coordinate (ul) at ($(u)+(-.18,-0.12)$){};
		\coordinate (ua1) at ($(u)+(0.216,0.5)$){};
		\coordinate (ua2) at ($(u)+(-0.216,0.5)$){};

		\coordinate (zur) at ($(z1u)+(.18,-0.12)$){};
		\coordinate (yur) at ($(y1u)+(0.216,0)$){};
		\coordinate (xur) at ($(x1u)+(.18,0.12)$){};

		\filldraw[draw = none, fill opacity=0.2,fill] (ur) -- (zur) arc (-33:147:.22cm) -- (ua1) -- (yur) arc (0:180:0.216cm) -- (ua2) -- (xur) arc (33:213:0.216cm)-- (ul) arc (213:327:0.215cm) (ur);

		\coordinate (br) at ($(b1)+(.12,-0.18)$){};
		\coordinate (bl) at ($(b1)+(-.12,-0.18)$){};
		\coordinate (ba1) at ($(b1)+(0.216,0.32)$){};
		\coordinate (ba2) at ($(b1)+(-0.216,0.32)$){};
		\coordinate (ba3) at ($(b1)+(-1.248,0.216)$){};
		\coordinate (ba4) at ($(b1)+(-0.648,-0.216)$){};

		\coordinate (uur) at ($(u2)+(.12,-0.18)$){};
		\coordinate (tur) at ($(t2)+(0.216,0)$){};
		\coordinate (sur) at ($(s2)+(.12,0.18)$){};
		\coordinate (aao) at ($(a1)+(0,0.216)$){};
		\coordinate (aat) at ($(a2)+(-.18,0.12)$){};

		\filldraw[draw = none, fill opacity=0.2,fill] (br) -- (uur) arc (-67:113:.22cm) -- (ba1) -- (tur) arc (0:180:0.216cm) -- (ba2) -- (sur) arc (67:247:0.216cm)-- (ba3) -- (aao) arc (90:270:0.216cm) --(ba4) -- (aat) arc (123:303:0.216cm) -- (br);

		\node[opacity = 0.5] at (11.83,-2) {\Huge$\partition$};
	\end{tikzpicture}
		\caption[Cycling on simple asymmetric FHGs]{Schematic of the simple asymmetric FHG of the hardness construction in Lemma~\ref{lem:IScoNPinFHG}. The figure is based on the instance $(\{i,j,k,x,y,z\},\{t,u,v\})$ with $t=\{i,j,k\}$, $u=\{j,k,x\}$, and $v=\{x,y,z\}$. The non-singleton coalitions of the initial partition $\partition$ are depicted in gray. The only possibility for $b_1$ to deviate is if one of $t_2$, $u_2$, or $v_2$ performs a deviation, which in turn can only happen if the coalition partners of her respective counterparts $t_1$, $u_1$, or $v_1$ have been deviating before.\label{fig:hardness_cycSimpleAsymFHG}}
	\end{figure*}
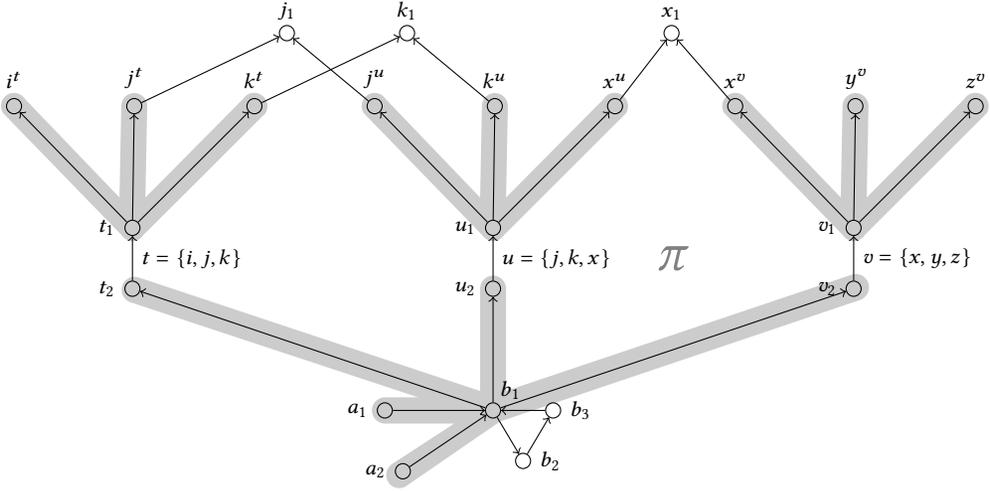

We claim that $(R,S)$ is a Yes-instance if and only if the IS dynamics starting with $\partition$ can cycle.

First assume that $(R,S)$ is a Yes-instance and let $S'\subseteq S$ be a partition of $R$ by the sets in $S$. We consider three stages of deviations. In the first stage, the agents in a coalition with some $s_1$ for $s\notin S'$ join the agents of type $r_v$.
This will leave all agents in $\{s_1\colon s\notin S'\}$ in singleton coalitions. In the second stage, agents $s_2$ for $s\notin S'$ join their copies $s_1$. This leaves the agent $b_1$ with a utility of $l/(2l+1)< \frac 12$. Hence, we start cycling in the final stage by having $b_1$ join $b_2$, $b_2$ join $b_3$, $b_3$ join $b_1$, and repeating these deviations.

Now, assume that there exists an infinite sequence of deviations starting from $\partition$. Agents of the type $r_v$ for $v = 1,\dots, m_r$ will never deviate, because they cannot receive positive utility. Agents of the type $r^s$ for $s\in S, r\in s$ can only deviate once to join an agent of the former type. Then, no agent can join their coalition, because the only agents $r^s$ would allow cannot deviate. In addition, $r^s$ can never improve her utility again. Hence, this coalition will stay the same for the remainder of the dynamics. Agents of the type $s_1$ will never deviate, because they are initially in their best coalition, and every agent that leaves can never be joined again. Next, agents of the type $s_2$ can only deviate if their copy $s_1$ is forced into a singleton coalition. At this point, they can deviate exactly once, forming a coalition that can never be changed again. Hence, after an agent $s_2$ performs a deviation, the agent $s_1$ has utility $0$ for the remainder of the dynamics. We refer to this fact as $(*)$. Moreover, since there are only $3|S|-|R|$ agents of the type $r_v$, at most $\frac{3|S|-|R|}{3} = |S| - \frac{|R|}3$ agents of the type $s_2$ can deviate, which means that at least $\frac{|R|}3$ agents of the type $s_1$ maintain a positive utility. We refer to this fact as $(\Delta)$.

Agents $a_v$ for $1\le v\le l$ can never deviate unless $b_1$ leaves their coalition. Agents $b_2$ and $b_3$ can only be involved in a deviation at most once until $b_1$ forms a coalition of her own or performs a deviation. Since $b_1$ can never form a coalition of her own, the only possibility for an infinite length of deviations is if $b_1$ performs a deviation. 
Since $b_1$ cannot join the coalition of agents of the type $s_2$ again, once they left her coalition, the only possible deviation is by joining the agent $b_2$ obtaining a utility of $\frac 12$. The utility of $b_1$ for any subset $C\subseteq \partition(b_1)$ that can arise before she deviated for the first time is $\util_{b_1}(C) = \frac h{l+1+h}$ for $h = |C\cap \{s_2\colon s\in S\}|$. 
It follows that $b_1$ can only deviate once all except $l$ agents of the type $s_2$ have left her coalition. 

Now let $\partition'$ be the partition right before the first deviation of $b_1$ and define $S' = \{s\in S\colon \util_{s_1}(\partition')>0\}$. By $(*)$ and because at least $|S| - l$ agents have left $b_1$, we know that $|S'| \le |S| - (|S| - l) = |R|/3$. This, together with $(\Delta)$ implies that $S'$ consists of exactly $|R|/3$ elements. By the distribution of agents of type $r_v$, the only way that all except $l$ agents of type $s_2$ have left $\partition(b_1)$ is if $S'$ covers precisely the elements of $R$. Hence, $S'$ forms a partition of $R$. Consequently, $(R,S)$ is a Yes-instance.
\end{proof}

\hardnessSimpleFHG*
\begin{proof}
	The reduction is from \textsc{Exact Cover by $3$-Sets}.

	Let an instance $(R,S)$ of \textsc{Exact Cover by $3$-Sets} be given and set $l= |S| -\frac {|R|} 3$. We construct the simple FHG induced by the following directed graph $G = (V,A)$. Let $V = \{r_1,r_2,r_3\colon r\in R\}\cup\{s_v^i\colon v = 1,2, i\in s\textnormal{ for } s\in S\}\cup \{t_v^w\colon v = 1,2,3, w = 1,\dots, l\}$ and edges given by $A = \{(r_1,r_2),(r_2,r_3),(r_3,r_1)\colon r\in R\}\cup\{(r_1,s_1^r),(s_1^r, r_1)\colon r\in s\textnormal{ for }s\in S\}\cup \{(s_1^i,s_2^i),(s_2^i,s_1^i)\colon i\in s\textnormal{ for }s\in S\}\cup \{(t_1^w,t_2^w),(t_2^w,t_3^w),(t_3^w,t_1^w)\colon w = 1,\dots, l\}\cup\{(t_1^w,s_1^i)\colon w = 1,\dots, l, i\in s\textnormal{ for } s\in S\}$. The construction is illustrated in Figure~\ref{fig:hardness_existSimpleFHG}. We define $T^w = \{t_1^w,t_2^w,t_3^w\}$ for $w = 1,\dots, l$.

	\begin{figure}
		\centering
		\begin{tikzpicture}[auto]

		\pgfmathsetmacro{\yshift}{0}
		\pgfmathsetmacro{\xshift}{0.1}
		\pgfmathsetmacro{\yshiftt}{0}
		\pgfmathsetmacro{\tinyshift}{0}

		\node[protovertex,label = {[shift={(\xshift,\yshift)}]180:\footnotesize $i_2$}] (i2) at (0,0){};
		\node[protovertex,label = {[shift={(-\xshift,-\yshift)}]0:\footnotesize $i_3$}] (i3) at (1,0){};
		\node[protovertex,label = 0:\footnotesize $i_1$] (i1) at (0.5,-0.83){};

		\node[protovertex,label = {[shift={(\xshift,\yshift)}]180:\footnotesize $j_2$}] (j2) at (2,0){};
		\node[protovertex,label = {[shift={(-\xshift,-\yshift)}]0:\footnotesize $j_3$}] (j3) at (3,0){};
		\node[protovertex,label = 0:\footnotesize $j_1$] (j1) at (2.5,-0.83){};

		\node[protovertex] (x2) at (4,0){};
		\node[protovertex] (x3) at (5,0){};
		\node[protovertex] (x1) at (4.5,-0.83){};

		\node[protovertex,label = {[shift={(\xshift,\yshift)}]180:\footnotesize $k_2$}] (k2) at (6,0){};
		\node[protovertex,label = {[shift={(-\xshift,-\yshift)}]0:\footnotesize $k_3$}] (k3) at (7,0){};
		\node[protovertex,label = 0:\footnotesize $k_1$] (k1) at (6.5,-0.83){};

		\node[protovertex] (y2) at (8,0){};
		\node[protovertex] (y3) at (9,0){};
		\node[protovertex] (y1) at (8.5,-0.83){};

		\node[protovertex] (z2) at (10,0){};
		\node[protovertex] (z3) at (11,0){};
		\node[protovertex] (z1) at (10.5,-0.83){};

		\foreach \x in {i,j,k,x,y,z}
		{
		\draw (\x1) edge[->] (\x2);
		\draw (\x2) edge[->] (\x3);
		\draw (\x3) edge[->] (\x1);
		}

		\node[protovertex,label = {[shift={(0.2,0)}]90:\footnotesize $u^i_1$}] (si1) at (1.5,-3.5){};
		\node[protovertex,label = {[shift={(-\xshift,0.2)}]0:\footnotesize $u^j_1$}] (sj1) at (3.75,-2.5){};
		\node[protovertex,label = {[shift={(0.3,-.1)}]90:\footnotesize $u^k_1$}] (sk1) at (6,-3.5){};
		\node[protovertex,label = {[shift={(\xshift,0)}]180:\footnotesize $u^i_2$}] (si2) at (.5,-3.5){};
		\node[protovertex,label = {[shift={(-\xshift,0)}]0:\footnotesize $u^j_2$}] (sj2) at (3.75,-1.5){};
		\node[protovertex,label = {[shift={(-\xshift,0)}]0:\footnotesize $u^k_2$}] (sk2) at (7,-3.5){};

		\node at (8,-2.5) {\footnotesize $u = \{i,j,k\}$};

		\draw (si1) edge[<->] (sj1);
		\draw (sj1) edge[<->] (sk1);
		\draw (sk1) edge[<->] (si1);

		\draw (si1) edge[<->] (i1);
		\draw (sj1) edge[<->] (j1);
		\draw (sk1) edge[<->] (k1);

		\draw (si1) edge[<->] (si2);
		\draw (sj1) edge[<->] (sj2);
		\draw (sk1) edge[<->] (sk2);

		\node[protovertex,label = {[shift={(\xshift,\yshift)}]180:{\footnotesize $t_1^1$}}] (t11) at (1.5,-5){};
		\node[protovertex,label = {[shift={(\xshift,\yshift)}]180:\footnotesize $t_2^1$}] (t12) at (1,-5.83){};
		\node[protovertex,label = {[shift={(-\xshift,\yshift)}]0:\footnotesize $t_3^1$}] (t13) at (2,-5.83){};
		\node[protovertex,label = {180:{\footnotesize $t_1^2$}}] (t21) at (3.7,-5){};
		\node[protovertex,label = {[shift={(\xshift,\yshift)}]180:\footnotesize $t_2^2$}] (t22) at (3.2,-5.83){};
		\node[protovertex,label = {[shift={(-\xshift,\yshift)}]0:\footnotesize $t_3^2$}] (t23) at (4.2,-5.83){};
		\node[protovertex,label = {[shift={(-\xshift,\yshift)}]0:{\footnotesize $t_1^l$}}] (t31) at (7.5,-5){};
		\node[protovertex,label = {[shift={(\xshift,\yshift)}]180:\footnotesize $t_2^l$}] (t32) at (7,-5.83){};
		\node[protovertex,label = {[shift={(-\xshift,\yshift)}]0:\footnotesize $t_3^l$}] (t33) at (8,-5.83){};

		\node at (5.6,-5.42) {\Huge $\dots$};

		\foreach \w in {1,2,3}{
		\draw (t\w1) edge[->] (t\w2);
		\draw (t\w2) edge[->] (t\w3);
		\draw (t\w3) edge[->] (t\w1);
		\draw (t\w1) edge[->] (si1);
		\draw (t\w1) edge[->] (sj1);
		\draw (t\w1) edge[->] (sk1);
		}

		\draw [decorate,decoration={brace,amplitude=10pt,mirror},yshift=0pt]
		(.5,-6) -- (8.5,-6) node [black,midway,below,yshift = -.3cm] {\footnotesize 
		$l = |S|-\frac {|R|}3$ many};

	\end{tikzpicture}
		\caption[Hardness of Existence Convergent Sequences]{Schematic of the simple FHG of the hardness construction in Theorem~\ref{thm:hardness-existencepath-simpleFHG}. Bidirected edges indicate a mutual utility of $1$. \label{fig:hardness_existSimpleFHG}}
	\end{figure}
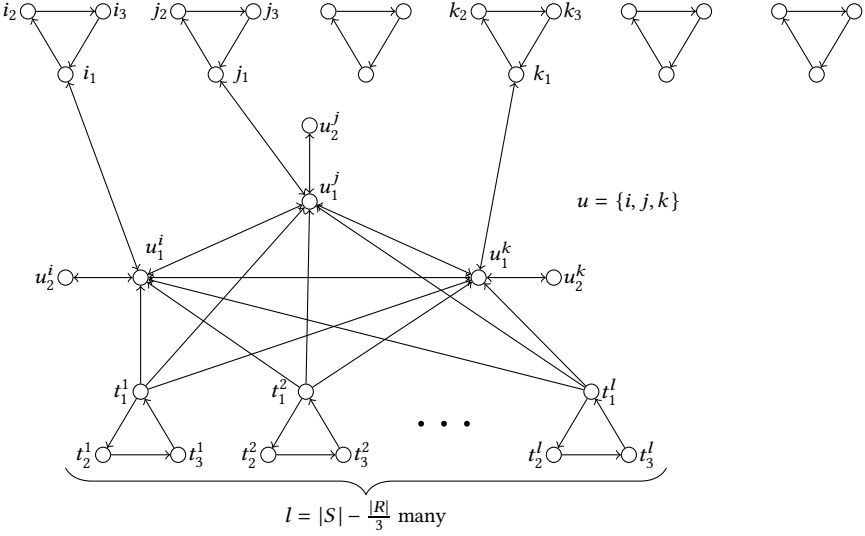

	Assume first that there exists a $3$-cover of $R$ through sets in $S$ and let $S'\subseteq S$ be a set of $3$-element sets partitioning $R$. 
	We will start by defining a partition that we can reach from the \singleton. To this end, we consider two functions the help us to define coalitions. Let $\sigma: \{1,\dots, l\} \to S\setminus S'$ be a bijection and let $\tau: R\to S'$ be the function defined by $\tau(r) = s$ for the unique $s\in S'$ with $r\in s$, i.e., the function that maps an element of $R$ to its partition class.
	By this, we can consider agents $\tau(r)_1^r$ and $\sigma(w)_1^i$ which associate agents in $r$ or elements in $[l]$ with specific agents representing sets in $S'$ and $S\setminus S'$.
	
	We define the partition of agents $\partition = \{\{r_2,r_3\},\{r_1,\tau(r)_1^r\}\colon r\in R\}\cup\{\{t_1^w\}\cup\{\sigma(w)_1^i\colon i\in \sigma(w)\}\colon w= 1,\dots, l\}\cup\{\{s_2^i\}\colon s\in S\}\cup \{\{t_2^w,t_3^w\}\colon w = 1,\dots,l\}$.

	Note that $\partition$ is IS. Let $i\in R$ and $w\in [l]$. Agents of the type $i_2$ or $t_2^w$ are in their best coalitions. Agents of the type $i_3,t_3^w$, or $s_2^i$ could only obtain positive utility by joining a coalition of which at least one agent would get worse if they joined. Agents of the type $i_1$ or $t_1^w$ cannot join another coalition that gives them positive utility because this would be blocked by an agent in that coalition. In particular, $i_1$ cannot join a coalition $\{t_1^j,t_1^k,t_1^v,\sigma^{-1}(t)\}$ for $t\in S\setminus S'$ with $i\in t$, because $\sigma^{-1}(t)$ blocks this. Similarly, $t_1^w$ cannot join a coalition $\{s_1^i,i_1\}$ for $i\in R$ or a coalition $\{t_1^x\}\cup\{\sigma(x)_1^i\colon i\in \sigma(x)\}$ for $x\neq w$, because this is blocked by $i_1$ and $t_1^x$, respectively. Finally, agents of the type $s_1^i$ obtain utility $1/2$ and cannot join $s_2^i$. Any other deviation to a coalition that gives them positive utility is blocked. Hence, $\partition$ is an IS partition of agents.

	Note that $\partition$ can be obtained by IS deviations from the \singleton by forming each of the coalitions in $\partition$. In particular, coalitions of the type $\{t_1^w\}\cup\{\sigma(w)_1^i\colon i\in \sigma(w)\}$ are formed by letting $t_1^w$ join $\sigma(w)_1^i$ for an arbitrary $i\in \sigma(w)$ and then the two $\sigma(w)_1^j$ for $j\in \sigma(w)\setminus \{i\}$ join one after another. This shows that we find a converging sequence if $(R,S)$ is a Yes-instance.

	Conversely, assume that there exists an IS partition $\partition$ of the agents that can be reached by IS deviations starting from the \singleton. We denote the sequence of partitions by $\partition^0,\dots, \partition^L$ for some integer $l$, where $\partition^0 = \{\{v\}\colon v\in V\}$ is the \singleton, $\partition^L = \partition$, and partition $\partition^{p+1}$ can be reached from partition $\partition^p$ by an IS deviation of agent $z^p$ for $0\le p\le l-1$.

	We start with a technical invariant of the IS dynamics that turns out to be very useful in determining the structure of the coalitions that agents of the type $r_1$ and $t_1^w$ eventually will be part of.

To formulate the claim, denote $S_1 = \{s_1^i\colon i\in s \textnormal{ for } s\in S\}$, $V^r = \{r_1,r_2,r_3\}$ for $r\in R$, and $\mathcal N = \{r_1 \colon r\in R\}\cup \{t_1^w\colon w= 1,\dots, l\}\cup\{s_2^i\colon i\in s \textnormal{ for } s\in S\}$. The set $\mathcal N$ contains precisely the agents that have a directed edge to or from an agent in $S_1$, i.e. the outgoing and incoming neighbors of agents in $S_1$. We simultaneously pose the following claims for $0\le p\le L$:

\begin{itemize}
	\item $\partition^p(r_3)\subseteq V^r$ for $r\in R$,
	\item $\partition^p(r_2)\subseteq V^r$ or $\partition^p(r_2) \subseteq \{r_1,r_2\}\cup\{s_1^r\colon s\in S, r\in s\}$ for $r\in R$,
	\item $\partition^p(t_v^w)\subseteq T^w$ for $v = 1,2, w = 1,\dots, l$,
	\item $V^r, T^w \notin \partition^p$ for $r\in R$ and $w = 1,\dots, l$,
	\item $\partition^p(s_2^i)\subseteq \{s_1^i,s_2^i\}$ for $s\in S, i\in s$,
	\item $\partition^p(a)\cap \mathcal N = \{a\}$, for $a\in \mathcal N$, and
	\item $\partition^p(a)\cap S_1\neq \emptyset$ implies $\util_a(\partition^p)>0$, for $a\in\mathcal N$.
\end{itemize}

The claim is initially true for the \singleton $\partition^0$. Assume that it holds after iteration $p$ for $0\le p\le L-1$. Consider the agent $z^p$ that performs the IS deviation to reach $\partition^{p+1}$. 
If $z^p\notin S_1\cup \mathcal N$, the claim holds for $p+1$ because these agents can only join the coalition with agents in their $3$-cycle and if they want to join the coalition of an agent in $\mathcal N$, this agent will block it if she already forms a coalition with an agent in $S^1$. 

If $z^p\in \mathcal N$, she will only deviate if she receives positive utility afterwards. The claim is true by induction if this positive utility comes from an agent outside $S_1$. Otherwise, she joins the coalition of $x\in S_1$. Then, $\partition^p(x)\cap \mathcal N =\emptyset$, because every agent $y\in \partition^p(y)\cap \mathcal N$ would block the inclusion of agent $z^p$ (by the final claim). In addition, since $z^p$ is the deviating agent, she will receive positive utility after this deviation. Hence, all claims hold. 

Finally, if $z^p\in S_1$, she joins an agent in $\mathcal N$ (otherwise she would not receive positive utility in $\partition^{p+1}$). If she joins an agent of type $s_2^i$, the claim follows because $\{s_2^i\}\in\partition^p$ by induction. If she joins an agent of type $i_1$ where $i\in s$, then $i_3\notin \partition^p(i^1)$ (this agent would block the deviation). 
Hence, the claim for the agent $i_2$ follows by induction. In addition, the claim for the agent $i_1$ follows because no agent from $\mathcal N$ joins and she receives positive utility through $s_1^i$ afterwards. Other IS deviations for the agents in $S_1$ are not possible. Together, the claims are established. In particular, they all hold for the IS partition $\partition$.

	We apply the claims to show that for every $w\in \{1,\dots, l\}$, there exists a $s\in S$ and $i\in s$ with $s_1^i\in \partition(t_1^w)$. Otherwise, $\partition(t_v^w)\subseteq T^w$ for $v = 1,2,3$ and $T^w\notin \partition$. Hence, $\partition$ is not IS.

	Now, fix $w\in \{1,\dots, l\}$ and let $s\in S$ and $i\in s$ with $s_1^i\in\partition(t_1^w)$. We claim that $\partition(t_1^w) = \{t_1^w\}\cup \{s_1^j\colon j\in s\}$. By the claims, $\partition(t_1^w)\subseteq \{t_1^w\}\cup S_1$. Under this condition, $\util_{s_1^u}(\partition) \le \frac 12$ and $\util_{s_1^u}(\partition) = \frac 12$ only if $\partition(t_1^w) = \{t_1^w\}\cup \{s_1^j\colon j\in s\}$. Note that $\{s_2^i\}\in \mathcal \partition$. Hence, $\util_{s_1^u}(\partition) \ge \frac 12$ since otherwise $\partition$ is not IS. Hence the claim follows.

	Define $S' = S\setminus \{s\in S\colon t_1^w\in \partition(s_1^i)\textnormal{ for } i\in s\}$. The coalitions of type $\{t_1^w\}\cup \{s_1^j\colon j\in s\}$ imply that $|S'|= |S|- (|S|-|R|/3) = |R|/3$.

	By the above claims, for every $r\in R$, there exists $s\in S$ with $r\in s$ and $s_1^r\in \partition(r_1)$. In particular, $s\in S'$. Hence, $\bigcup_{s\in S'}s = R$ and since $|S'|= |R|/3$ and $|s| =3$ for all $s\in S'$, the sets in $S'$ must be disjoint. Hence, $(R,S)$ is a Yes-instance.
\end{proof}

\section{Dichotomous Hedonic Games}

\existDHG*

We prove the two hardness results by providing separate reductions for each problem in the next two lemmas.

\begin{lemma}\label{lem:hard-exist-DHG}
\existpb{IS}{DHG} is \np-hard even when starting from the \singleton.
\end{lemma}

\begin{proof}
Let us perform a reduction from (3,B2)-SAT~\citep{BKS03a}. 
In an instance of (3,B2)-SAT, we are given a CNF propositional formula $\varphi$ where every clause $C_j$, for $1\leq  j \leq m$, contains exactly three literals and every variable $x_i$, for $1\leq i \leq p$, appears exactly twice as a positive literal and twice as a negative literal. 
From such an instance, we construct an instance of a dichotomous hedonic game with initial partition as follows.

For each clause $C_j$, for $1\leq j\leq m$, we create a clause-agent $k_j$ and agents $k_j^2$ and $k_j^3$.
For each variable $x_i$, for $1\leq i\leq p$, we create a variable-agent $v_i$ and agents $v_i^2$ and $v_i^3$.
The agents $k_j^2$ and $k_j^3$ (or $v_i^2$ and $v_i^3$) are used to form a gadget involving clause-agent $k_j$ (or variable-agent $v_i$) to reproduce the counterexample provided in the proof of Proposition~\ref{prop:noconv-dicho}.
For each $\ell^{\text{th}}$ occurrence ($\ell\in\{1,2\}$) of a positive literal $x_i$ (or negative literal $\overline{x}_i$), we create a literal-agent $y_i^\ell$ (or $\overline{y}_i^\ell$).
The initial partition $\spartition$ is the \singleton, i.e., every agent is initially alone.
The dichotomous preferences of the agents are described in Table~\ref{tab:pref-reduc-exist-DHG}.

\begin{table}
\caption{Coalitions approved by the agents in the reduced instance of Lemma~\ref{lem:hard-exist-DHG}, for every $1\leq j\leq m$, $1\leq i\leq p$ and $\ell\in\{1,2\}$. Notation $lit_j^r$ stands for the literal-agent associated with the $r^\text{th}$ literal of clause $C_j$ ($1\leq r\leq 3$) and $cl(x_i^\ell)$ (or $cl(\overline{x}_i^\ell)$) denotes the index of the clause to which literal $x_i^\ell$ (or $\overline{x}_i^\ell$) belongs. 
All the coalitions that are not mentioned are disapproved by the agents.}
\label{tab:pref-reduc-exist-DHG}
\begin{center}
\begin{tabular}{*{2}{c}}
\toprule
Agents & Approved coalitions \\
\midrule
$k_j$ & $\{k_j,lit_j^1\}$, $\{k_j,lit_j^2\}$, $\{k_j,lit_j^3\}$, $\{k_j,k_j^2\}$ \\
$v_i$ & $\{v_i,y_i^1,y_i^2\}$, $\{v_i,\overline{y}_i^1,\overline{y}_i^2\}$, $\{v_i,v_i^2\}$\\
$y_i^\ell$ & $\{y_i^\ell,y_i^{3-\ell}\}$, $\{y_i^\ell,y_i^{3-\ell},v_i\}$, $\{y_i^\ell,k_{cl(x_i^\ell)}\}$ \\
$\overline{y}_i^\ell$ & $\{\overline{y}_i^\ell,\overline{y}_i^{3-\ell}\}$, $\{\overline{y}_i^\ell,\overline{y}_i^{3-\ell},v_i\}$, $\{\overline{y}_i^\ell,k_{cl(\overline{x}_i^\ell)}\}$ \\
$k_j^2$ & $\{k_j^2,k_j^3\}$ \\
$k_j^3$ & $\{k_j,k_j^3\}$ \\
$v_i^2$ & $\{v_i^2,v_i^3\}$ \\
$v_i^3$ & $\{v_i,v_i^3\}$\\
\bottomrule
\end{tabular}
\end{center}
\end{table}

We claim that there exists a sequence of IS deviations ending in an IS partition iff formula $\varphi$ is satisfiable.

Suppose first that there exists a truth assignment of the variables $\phi$ such that formula $\varphi$ is satisfiable.
Let us denote by $\ell_j$ a chosen literal-agent associated with an occurrence of a literal true in $\phi$ which belongs to clause $C_j$.
Since all the clauses of $\varphi$ are satisfied by $\phi$, there exists such a literal-agent $\ell_j$ for each clause $C_j$.
Now let us denote by $z_i^1$ and $z_i^2$ the literal-agents associated with the two occurrences of the literal of variable $x_i$ which is false in $\phi$.
Since $\phi$ is a truth assignment of the variables that satisfies all the clauses of formula $\varphi$, it holds that $\bigcup_{1\leq j\leq m}\{\ell_j\} \cap \bigcup_{1\leq i\leq n}\{z_i^1,z_i^2\}=\emptyset$.
Let us consider the following sequence of IS deviations starting from the \singleton where every agent has utility 0:
\begin{itemize}
\item For every $1\leq j\leq m$, literal-agent $\ell_j$ joins clause-agent $k_j$, which makes both agents happier since they now belong to an approved coalition;
\item For every $1\leq i\leq n$, literal-agent $z_i^1$ joins literal-agent $z_i^2$, which makes both agents happier since they now belong to an approved coalition (they correspond to two occurrences of the same literal), and then variable-agent $v_i$ joins them, which makes $v_i$ happier without deteriorating the satisfaction of agents $z_i^1$ and $z_i^2$;
\item For every two agents $y_i^1$ and $y_i^2$ (or $\overline{y}_i^1$ and $\overline{y}_i^2$) who were not involved in the previous deviations (i.e., literal $x_i$ (or $\overline{x_i}$) is true in $\phi$ but the two occurrences of this literal have not been used for satisfiability of formula $\varphi$), literal-agent $y_i^1$ joins literal-agent $y_i^2$, which makes both agents happier since they now belong to an approved coalition;
\item For every $1\leq j\leq m$, agent $k_j^2$ joins agent $k_j^3$, which makes agent $k_j^2$ happier and does not deteriorate the satisfaction of agent $k_j^3$ who still belongs to a disapproved coalition;
\item For every $1\leq i\leq n$, agent $v_i^2$ joins agent $v_i^3$, which makes agent $v_i^2$ happier and does not deteriorate the satisfaction of agent $v_i^3$ who still belongs to a disapproved coalition.
\end{itemize}
We claim that the resulting partition is IS. 
Observe that the only dissatisfied agents (who are the only ones who would have an incentive to still perform an IS deviation) are the literal-agents who remained alone, agents $k_j^3$ for every $1\leq j\leq m$ and agents $v_i^3$ for every $1\leq i\leq n$. 
The only better coalition for agent $k_j^3$ is the one she would form with only clause-agent $k_j$. 
However, there is no clause-agent $k_j$ still alone since all the clauses are satisfied by truth assignment $\phi$. 
The only better coalition for agent $v_i^3$ is the one she would form with only variable-agent $v_i$. 
However, there is no variable-agent $v_i$ still alone since $\phi$ is a truth assignment of all variables.
For remaining literal-agents, they must correspond to a true literal in $\phi$ for which the literal-agent associated with the other occurrence of the literal already forms a pair with a clause-agent. 
Therefore, they cannot join this other literal-agent.
Moreover, they cannot join their associated clause-agent because she is not alone anymore.
Hence, there is no IS deviation from this partition, which is then IS.

Suppose now that there does not exist a truth assignment of the variables that satisfies all the clauses of formula $\varphi$.
Suppose that a clause-agent $k_j$ cannot form a coalition with one of the literal-agents associated with the literals of her clause. 
Then, the only possible approved coalition for agent $k_j$ is $\{k_j,k_j^2\}$.
This implies that there will be a cycle of IS deviations among the agents $k_j$, $k_j^2$ and $k_j^3$, as described in the proof of Proposition~\ref{prop:noconv-dicho}, by considering agents $k_j$, $k_j^2$ and $k_j^3$ as agents $1$, $2$, and $3$, respectively, from the counterexample.
Now suppose that a variable-agent $v_i$ cannot form a coalition with either $y_i^1$ and $y_i^2$, or $\overline{y}_i^1$ and $\overline{y}_i^2$. 
Then, the only possible approved coalition for agent $v_i$ is $\{v_i,v_i^2\}$.
Therefore, there will be a cycle among the agents $v_i$, $v_i^2$ and $v_i^3$, as described in the proof of Proposition~\ref{prop:noconv-dicho}, by considering agents $v_i$, $v_i^2$ and $v_i^3$ as agents $1$, $2$, and $3$, respectively, from the counterexample.
Therefore, since there is no possibility to find a truth assignment of the variables which satisfies all the clauses, we cannot simultaneously have that each clause-agent $k_j$ forms a coalition with one of the literal-agents associated with the literals of her clause, and that each variable-agent $v_i$ forms a coalition with either $y_i^1$ and $y_i^2$, or $\overline{y}_i^1$ and $\overline{y}_i^2$. 
Hence, we necessarily get a cycle in a sequence of IS deviations starting from the \singleton.
\end{proof}

\begin{lemma}
\convergpb{IS}{DHG} is \conp-hard.
\end{lemma}

\begin{proof}
For this purpose, we prove the \np-hardness of the complement problem, which asks whether there exists a cycle of IS deviations. 
Let us perform a reduction from the \textsc{Satisfiability} problem which asks the satisfiability of a CNF propositional formula $\varphi$ given by a set of clauses $C_1,\dots,C_m$ over variables $x_1,\dots,x_p$.
We construct an instance of a dichotomous hedonic game with initial partition as follows.

For each clause $C_j$, for $1\leq j \leq m$, we create two clause-agents $k_j$ and $k'_j$. 
Let us denote by $p_i$ (or $n_i$) the number of positive (or negative) literals of variable $x_i$ in formula $\varphi$.
For each $t^{\text{th}}$ occurrence of literal $x_i$ (or $\overline{x_i}$) of variable $x_i$, we create a literal-agent $y_i^t$ (or $\overline{y}_i^t$).
The initial partition is given by $\partition^0:=\{\{x_i^1,\dots,x_i^{p_i},\overline{x}_i^1,\dots,\overline{x}_i^{n_i}\}_{1\leq i\leq p},\{k_j,k'_j\}_{1\leq j\leq m}\}$.   
The dichotomous preferences of the agents over the coalitions to which they belong are summarized below. 
\begin{itemize}
\item Each literal-agent $y_i^t$ (or $\overline{y}_i^t$), for $1\leq i\leq p$ and $1\leq t\leq p_i$ (or $1\leq t\leq n_i$), gives utility 1 to:
\begin{itemize}
\item the coalitions to which agent $k'_j$ belongs, where $k'_j$ refers to the clause $C_j$ to which the $t^\text{th}$ occurrence of literal $x_i$ (or $\overline{x}_i$) belongs, and
\item all coalitions only composed of literal-agents associated with variable $x_i$ where at least one literal-agent associated with $\overline{x_i}$ (or $x_i$) is missing.
\end{itemize}
All the other coalitions are valued 0.
\item Each clause-agent $k_j$, for $1\leq j\leq m$, only gives utility 1 to the coalitions which contain clause-agent $k'_{j+1}$ and one literal-agent associated with a literal belonging to clause $C_{j+1}$ (where $m+1$ refers to $1$).
All the other coalitions are valued 0.
\item Each clause-agent $k'_j$, for $1\leq j\leq m$, only gives utility 1 to the coalitions which contain agent $k_j$.
All the other coalitions are valued 0.
\end{itemize} 

We claim that there exists a cycle of IS deviations iff formula $\varphi$ is satisfiable.

Suppose first that formula $\varphi$ is satisfiable by a truth assignment of the variables denoted by $\phi$.
For each clause $C_j$, for $1\leq j\leq m$, we choose a literal-agent $y_i^t$ (or $\overline{y}_i^t$) such that the $t^{\text{th}}$ occurrence of literal $x_i$ (or $\overline{x}_i$) belongs to clause $C_j$ and literal $x_i$ (or $\overline{x}_i$) is true in $\phi$.
By satisfiability of formula $\varphi$, there always exists such a literal-agent.
Then, literal-agent $y_i^t$ (or $\overline{y}_i^t$) deviates from her coalition of literal-agents associated with variable $x_i$ to coalition $\{k_j,k'_j\}$.
This deviation is beneficial for the literal-agent because she values her new coalition with utility 1 since $k'_j$ belongs to it and her old coalition with utility 0 since no literal-agent associated with her opposite literal has left the coalition of literal-agents associated with variable $x_i$ (we have chosen only literal-agents associated with literals true in $\phi$).
Moreover, this deviation does not decrease the utility of the agents of the joined coalition: agent $k'_j$ still values the coalition with utility 1 since agent $k_j$ belongs to it and agent $k_j$ still values the coalition with utility 0.
Therefore, this deviation is an IS deviation.
After all these deviations, we reach a partition $\partition$ which contains the coalitions $\{k_j,k'_j,\ell_j\}$ for every $1\leq j\leq m$, where $\ell_j$ denotes a literal-agent associated with a literal true in $\phi$ which belongs to clause $C_j$.

We describe below the deviations that lead to come back to partition $\partition$.
\begin{enumerate}
\item For each $1\leq j\leq m$, by increasing order of the indices, clause-agent $k_j$ deviates to coalition $\{k_{j+1},k'_{j+1},\ell_{j+1}\}$ (where $m+1$ refers to $1$).
This deviation is beneficial for clause-agent $k_j$ since she deviates to a coalition containing $k'_{j+1}$ and a literal-agent associated with a literal belonging to clause $C_{j+1}$.
By design of the preferences, this deviation does not hurt the members of the joined coalition, 
therefore it is an IS deviation.  
However, when clause-agent $k_j$ has left her old coalition, this old coalition becomes either $\{k'_j,\ell_j\}$ if $j=1$ or $\{k'_j,\ell_j,k_{j-1}\}$ otherwise.
Therefore, this deviation hurts clause-agent $k'_j$ from the old coalition.
After all these deviations, we reach a partition which contains the coalitions $\{k_j,k'_{j+1},\ell_{j+1}\}$ for every $1\leq j\leq m$ (where $m+1$ refers to $1$).
\item For each $1\leq j\leq m$, by increasing order of the indices, clause-agent $k'_j$ deviates to coalition $\{k_j,k'_{j+1},\ell_{j+1}\}$ (where $m+1$ refers to $1$), in order to recover utility 1 by belonging to the same coalition as clause-agent $k_j$.
By design of the preferences, this deviation does not hurt the members of the joined coalition, therefore it is an IS deviation.
However, when clause-agent $k'_j$ has left her old coalition, this old coalition becomes either $\{k_{j-1},\ell_j\}$ if $j=1$ or $\{k_{j-1},k'_{j-1},\ell_j\}$ otherwise.
Therefore, this deviation hurts clause-agent $\ell_j$ from the old coalition.
After all these deviations, we reach a partition which contains the coalitions $\{k_j,k'_j,\ell_{j+1}\}$ for every $1\leq j\leq m$ (where $m+1$ refers to $1$).
\item For each $1\leq j\leq m$, by increasing order of the indices, literal-agent $\ell_j$ deviates to coalition $\{k_j,k'_j,\ell_{j+1}\}$, in order to recover utility~1 by belonging to the same coalition as clause-agent $k'_j$.
By design of the preferences, this deviation does not hurt the members of the joined coalition, therefore it is an IS deviation.
After all these deviations, we reach again partition $\partition$. 
\end{enumerate}
Hence, there is a cycle in the sequence of IS deviations.

Suppose now that there exists a cycle of IS deviations.
From $\partition^0$, no clause-agent has incentive to deviate: each clause-agent $k'_j$ already values her current coalition with utility 1 since agent $k_j$ belongs to it, and each clause-agent~$k_j$ values her current coalition with utility 0 but there is no coalition containing both agent $k'_{j+1}$ and a literal-agent associated with a literal belonging to clause $C_{j+1}$.
Therefore, some literal-agents must deviate and leave their initial coalition, that they value with utility 0 since no literal-agent has left it yet.
Observe that once a literal-agent associated with variable $x_i$ has left her initial coalition, no literal-agent associated with the opposite literal can leave the coalition because she values it with utility 1.
If a literal-agent deviates, this is for joining coalition $\{k_j,k'_j\}$ where clause $C_j$ refers to the clause where her associated literal occurrence appears.
After such a deviation which is an IS deviation because it does not decrease the utility of the members of the joined coalition, the only agents with incentive to deviate are clause-agents $k_j$ if a literal-agent $\ell_{j+1}$ has joined coalition $\{k_{j+1},k'_{j+1}\}$.
Suppose that there exists a clause coalition $\{k_j,k'_j\}$ such that no literal-agent has joined it.
Then, consider a clause index $j$ such that no literal-agent has joined coalition $\{k_{j+1},k'_{j+1}\}$ and a clause index $b$ such that for all clause coalitions $\{k_r,k'_r\}$, with $b\leq r\leq j$, a literal-agent $\ell_r$ has joined the coalition but this is not the case for coalition $\{k_{b-1},k'_{b-1}\}$ ($m+1$ refers to $1$, and $0$ to $m$).
Then, by progressive IS deviations, all agents belonging to clause coalitions with index between $b$ and $j$ will deviate for joining coalition $\{k_j,k'_j\}$.
Indeed, clause-agent $k_{b}$ will deviate to coalition $\{k_{b+1},k'_{b+1},\ell_{b+1}\}$, and then clause-agent $k'_{b}$ will follow her in this coalition, and then literal-agent $\ell_{b}$ will also follow $k'_{b}$ in this coalition.
But agent $k_{b+1}$ has incentive to do the same for coalition $\{k_{b+2},k'_{b+2},\ell_{b+2}\}$, which leads agents $k'_{b+2}$ and $\ell_{b+2}$ to follow her, as well as agents $k_{b}$, $k'_{b}$ and $\ell_{b}$.
This process of IS deviations then continues in the same way until all these agents group in coalition $\{k_j,k'_j,\ell_j\}$.
However, since clause-agent $k_j$ can never leave this coalition (there is no coalition containing both agent $k'_{j+1}$ and a literal-agent associated with a literal belonging to clause $C_{j+1}$), no other agent will leave this coalition neither.
We will therefore reach a stable state, a contradiction.
It follows that each clause coalition $\{k_j,k'_j\}$, for $1\leq j\leq m$, must be joined by a literal-agent $\ell_j$ associated with a literal belonging to clause $C_j$.
Therefore, by setting to true the literals associated with literal-agents who have joined clause coalitions (we have previously said that no two literal-agents associated with opposite literals can both leave their initial coalition), we get a truth assignment of the variables which satisfies all the clauses of formula~$\varphi$.
\end{proof}

\end{document}